\documentclass[aps,prx, twocolumn, footnote, superscriptaddress, longbibliography]{revtex4-2}

\usepackage{soul}
\usepackage{relsize}
\usepackage{lmodern}
\usepackage{slantsc}
\usepackage{bbm}
\usepackage{graphicx}
\usepackage{amsmath}
\usepackage{amssymb}
\usepackage{amsthm}
\usepackage{amsbsy}
\usepackage{mathrsfs}
\usepackage{varioref}
\usepackage{dsfont}
\usepackage{bm}
\usepackage[usenames,dvipsnames]{xcolor}
\definecolor{myblue}{rgb}{0.2,0.2,0.8}
\definecolor{myblack}{rgb}{0,0,0}
\definecolor{myurl}{rgb}{0.1,0.1,0.4}
\usepackage[colorlinks=true,citecolor=myblue,linkcolor=myblack,urlcolor=myurl]{hyperref}
\usepackage{cleveref}
\usepackage[font = footnotesize, labelfont = bf, justification = RaggedRight]{caption}
\usepackage{subfigure}
\usepackage{booktabs} 
\usepackage{array} 
\usepackage{paralist} 
\usepackage{verbatim} 
\edef\restoreparindent{\parindent=\the\parindent\relax}
\usepackage{multirow}

\usepackage[parfill]{parskip}
\restoreparindent

\setul{}{20 pt}

\newtheorem{definition}{Definition}
\newcommand{\defref}[1]{Definition~\ref{#1}}
\newtheorem{lemma}{Lemma}
\numberwithin{lemma}{section}
\newcommand{\lemref}[1]{Lemma~\ref{#1}}
\newtheorem{prop}{Proposition}
\numberwithin{prop}{section}
\newcommand{\propref}[1]{Proposition~\ref{#1}}
\newtheorem{corollary}{Corollary}
\numberwithin{corollary}{section}
\newcommand{\corref}[1]{Corollary~\ref{#1}}
\newtheorem{example}{Example}
\numberwithin{example}{section}

\newtheorem{theorem}{Theorem}
\numberwithin{theorem}{section}
\newcommand{\thmref}[1]{Theorem~\ref{#1}}

\usepackage{pifont}
\newcommand{\cmark}{\ding{51}}%
\newcommand{\xmark}{\ding{55}}%

\newcommand{\<}{\langle}
\renewcommand{\>}{\rangle}

\newcommand{\rr}{{\mathcal{R}}}
\newcommand{\lo}{{\mathcal{L}}}

\newcommand{\s}{{\mathcal{S}}}
\newcommand{\ee}{{\mathcal{E}}}
\newcommand{\xx}{{\mathcal{X}}}
\newcommand{\yy}{{\mathcal{Y}}}

\renewcommand{\aa}{{\mathcal{A}}}

\newcommand{\ff}{{\mathcal{F}}}

\newcommand{\mm}{{\mathcal{M}}}

\newcommand{\ii}{{\mathcal{I}}}

\newcommand{\co}{\mathds{C}}
\newcommand{\re}{\mathds{R}}

\newcommand{\h}{{\mathcal{H}}}
\newcommand{\kk}{{\mathcal{K}}}
\newcommand{\hs}{{\mathcal{H}\sub{\s}}}
\newcommand{\ha}{{\mathcal{H}\sub{\aa}}}

\newcommand{\E}{\mathsf{E}}
\newcommand{\F}{\mathsf{F}}
\newcommand{\Z}{\mathsf{Z}}
\newcommand{\G}{\mathsf{G}}

\renewcommand{\P}{\mathsf{P}}

\renewcommand{\nat}{\mathds{N}}

\newcommand{\one}{\mathds{1}}
\newcommand{\onesys}{\mathds{1}\sub{\s}}
\newcommand{\oneapp}{\mathds{1}\sub{\aa}}

\newcommand{\zero}{\mathds{O}}
\newcommand{\imag}{\mathfrak{i}}

\newcommand{\tr}{\mathrm{tr}}
\newcommand{\tra}{\mathrm{tr}\sub{\aa}}

\newcommand{\av}{\mathrm{av}}

\newcommand{\sub}[1]{_{\!\mathsmaller{\, #1}}}
\newcommand{\eq}[1]{Eq.~\eqref{#1}}

\newcommand{\app}[1]{Appendix~(\ref{#1})}
\newcommand{\ket}[1]{|{#1}\rangle}

\newcommand{\pr}[1]{P_{#1}}
\newcommand{\rank}[1]{\mathrm{rank}\left( {#1}\right)}

\begin{document}

\title{Quantum measurements constrained by the third law of thermodynamics}
\author{M. Hamed Mohammady  }
\email{mohammad.mohammady@ulb.be}
\affiliation{QuIC, \'{E}cole Polytechnique de Bruxelles, CP 165/59, Universit\'{e} Libre de Bruxelles, 1050 Brussels, Belgium}

\author{Takayuki Miyadera}
\email{miyadera@nucleng.kyoto-u.ac.jp}
\affiliation{Department of Nuclear Engineering, Kyoto University, Nishikyo-ku, Kyoto 615-8540, Japan}


\begin{abstract}
In the quantum regime, the third law of thermodynamics implies the unattainability of pure states. As shown recently, such unattainability implies that a  unitary interaction between the measured system and a measuring apparatus can never implement an ideal projective measurement. In this paper, we introduce an operational formulation of the third law for the most general class of physical transformations, the violation of which is both necessary and sufficient for the preparation of pure states. Subsequently, we investigate how such a law constrains measurements of general observables, or positive operator valued measures. We identify several desirable properties of measurements which are simultaneously enjoyed by ideal projective measurements---and are hence all ruled out by the third law in such a case---and determine if the third law allows for these properties to obtain for general measurements of general observables and, if so, under what conditions. It is shown that while the third law rules out some of these properties for all observables, others may be enjoyed by observables that are sufficiently ``unsharp''.
\end{abstract}

\maketitle

\section{introduction}

One of the standard assumptions of textbook quantum mechanics is the  ``L\"uders rule'' which states that when an observable---represented by a self-adjoint operator with a non-degenerate spectrum---is measured in a system, the state of the  system collapses to the eigenstate of the observable associated with the observed eigenvalue  \cite{Luders2006, Busch2009a, Hegerfeldt2012}. But assuming the universal validity of quantum theory, such a state change must be consistent with a  description of the measurement process as a physical interaction between the system to be measured and a given  (quantum) measuring apparatus.  In his groundbreaking contribution to quantum theory in 1932 \cite{Von-Neumann-Foundations}, von Neumann introduced just such a model for the measurement process, where the system and apparatus interact unitarily.

While L\"uders measurements, and von Neumann's model for their realisation, are always available within the formal framework of quantum theory, they may not always be feasible in practice---technological obstacles and fundamental physical principles must also be accounted for. One such  principle is that of conservation laws and, as shown by the   Wigner-Araki-Yanase theorem \cite{E.Wigner1952,Busch2010, Araki1960,Ozawa2002,Miyadera2006a,Loveridge2011,Ahmadi2013b,Loveridge2020a,Mohammady2021a,Kuramochi2022}, only observables commuting with the conserved quantity admit a L\"uders measurement.  This observation naturally raises the following question: do other physical principles constrain quantum measurements, and if so, how? Given that von Neumann's model for the measurement process assumes that the measuring apparatus is initially prepared in a pure state, an obvious candidate for consideration immediately presents itself: the third law of thermodynamics, or Nernst's unattainability principle, which states that a system cannot be  cooled to absolute zero temperature with finite  time, energy, or control complexity;  in the quantum regime, the third law  prohibits the preparation of pure states \cite{Schulman2005, Allahverdyan2011a, Reeb2013a, Masanes2014, Ticozzi2014, Scharlau2016a, Wilming2017a, Freitas2018, Clivaz2019, Taranto2021, Buffoni2022}.  As argued by  Guryanova \emph{et al.}, such unattainability rules out   L\"uders measurements for  any self-adjoint operator with a non-degenerate spectrum \cite{Guryanova2018}.  

The more modern quantum theory of measurement \cite{PaulBuschMarianGrabowski1995, Busch1996, Heinosaari2011, Busch2016a} states that the properties of a quantum system are not exhausted by its sharp observables, i.e., observables represented by self-adjoint operators. Indeed, observables can be fundamentally unsharp, and are properly represented as positive operator valued measures (POVM) \cite{Busch2010a, Jaeger2019}. Similarly, the state change that results from measurement is more properly captured by the notion of instruments \cite{Davies1970}, which need not obey the L\"uders rule.  Moreover, the interaction between system and apparatus during the measurement process is not necessarily  unitary, and is more generally described as a channel, which   more accurately describes situations where the interaction with the environment cannot be neglected. Therefore, how the third law constrains  general measurements should be addressed; in this paper, we shall thoroughly examine this in the finite-dimensional setting. 

First, we provide a minimal operational formulation of the third law by constraining the class of permissible channels so that the availability of a channel not so constrained is both necessary and sufficient for the preparation of pure states.   The considered class of channels include  those whose input and output spaces are not the same, which is the case when the process considered involves composing and discarding systems, and is more general  than the class of rank non-decreasing channels, such as  unitary channels. Indeed, the rank non-decreasing concept can only be properly applied to the limited cases where the input and output systems of a channel are the same.
 
 Subsequently, we consider the most general class of measurement schemes that are constrained by the third law. That is, we  do not assume that  the measured observable is sharp, or that the pointer observable is sharp, or that the measurement interaction is  rank non-decreasing.  Next, we determine if the instruments realised by such measurement schemes may satisfy several desirable properties and, if so, under what conditions. These properties are:
\begin{enumerate}[\bf(i)]
    \item {\bf Non-disturbance:}  a non-selective measurement does not affect the subsequent measurement statistics of any observable that commutes with the measured observable.
    
    \item {\bf First-kindness:}  a non-selective measurement of an observable does not affect its  subsequent measurement statistics.
    
    \item {\bf  Repeatability:}  successive measurements of an observable are guaranteed to produce the same outcome.
    
    \item {\bf Ideality:}  whenever an outcome is certain from the outset, the measurement does not change the state of the measured system.
    
    \item {\bf  Extremality:}  the instrument cannot be  written as a probabilistic mixture of distinct instruments.
\end{enumerate}
  L\"uders measurements of sharp observables  simultaneously satisfy the above properties.  
In general, however, these properties can be satisfied by instruments   that do not obey the L\"uders rule, and also for observables that are not necessarily sharp. Moreover, they  are in general not equivalent: an instrument can enjoy one while not another \cite{Lahti1991, Heinosaari2010, DAriano2011}. We therefore investigate each such property individually,  providing necessary and sufficient conditions for their fulfilment by a measurement constrained by the third law.

We show that the third law prohibits a measurement of any \emph{small-rank} observable---an observable that has at least one rank-1 effect, or POVM element---from satisfying any of the above properties. On the other hand, extremality is shown to be permitted for an observable if each effect has sufficiently large rank, but only if the interaction between the system and apparatus is non-unitary. Finally, we show that while repeatbility and ideality are  forbidden for all observables, non-disturbance and first-kindness are permitted for observables that are \emph{completely unsharp}: the effects of such observables do not have either eigenvalue 1 or 0, and so such observables do not enjoy the ``norm-1'' property. That is, non-disturbance and first-kindness are only permitted for observables  that cannot have a definite value in any state.  Our results are summarised in Table \ref{table:results}.

\begin{table}[!htb]
\begin{tabular}{ |p{0.5cm}||p{1.5cm}|p{1.5cm}|p{1.5cm}| p{1.5cm}| }
 \hline
 \multicolumn{1}{|c||}{} &  \multicolumn{4}{c|}{Observable} \\
  \hline
& Small-rank  & Sharp & Norm-1 &  Completely unsharp\\
 \hline
 \bf{(i)}   & \xmark    & \xmark &  \xmark & \cmark\\
  \hline
 \bf{(ii)}  &   \xmark   & \xmark  & \xmark & \cmark\\
  \hline
\bf{(iii)}   & \xmark & \xmark &  \xmark & \xmark\\
  \hline
 \bf{(iv)}     & \xmark & \xmark &  \xmark & \xmark \\
  \hline
 \bf{(v)}  &   \xmark  &  \cmark & \cmark &  \cmark\\
 
 \hline 

\end{tabular}
\caption{The possibility (\cmark) or impossibility (\xmark) of an observable to admit the properties  {\bf(i) - (v)} outlined above are indicated  for four classes of observables: small-rank observables have at least one rank-1 effect; sharp observables are such that all effects are projections; norm-1 observables are such that every effect has eigenvalue 1; and completely unsharp observables are such that no effect has eigenvalue  1 or 0. }
\label{table:results}
\end{table}

\section{Operational formulation of the third law for channels}
The third law of thermodynamics states that in the absence of infinite resources of time, energy, or control complexity, a system cannot be cooled to absolute zero temperature. Assuming the universal validity of this law, then it must also hold in the quantum regime \cite{Schulman2005, Allahverdyan2011a, Reeb2013a, Masanes2014, Ticozzi2014, Scharlau2016a, Wilming2017a, Freitas2018, Clivaz2019, Taranto2021}. Throughout, we shall only consider quantum systems with a finite-dimensional Hilbert space $\h$. When such a system is in thermal equilibrium at some temperature, it is in a Gibbs state, and whenever the temperature is non-vanishing, such states are full-rank. Conversely, at absolute zero temperature the system will be in a low-rank state, i.e., it will not have full rank. In the special case of a non-degenerate Hamiltonian, the system will in fact be in a pure state.  A minimal operational formulation of the third law in the quantum regime can therefore be phrased as follows:  the possible transformations of quantum systems must be constrained so that the only attainable states have full rank.

 In the Schr\"odinger picture, the most general transformations of quantum systems are represented by  channels $\Phi: \lo(\h) \to \lo(\kk)$, i.e., completely positive trace-preserving maps from the algebra of linear operators on an input Hilbert space $\h$ to that of an output space $\kk$. In the special case where $\h = \kk$, we say that $\Phi$ acts in $\h$. But in general  $\h$ need not be identical to  $\kk$, and the two systems may have different dimensions. This is because physically permissible transformations include the composition of multiple systems, and discarding of subsystems.  

Previous formulations   of the third law (see, for example,  Proposition 5 of Ref. \cite{Reeb2013a} and Appendix B of Ref. \cite{Taranto2021}) have  restricted the class of available channels to those with the same input and output system, and where the channel does not reduce the rank of the input state of such a system: these are referred to as rank non-decreasing channels, with unitary channels constituting a simple example.  An intuitive argument for such restriction is as follows. Consider the case where we wish to cool the system of interest by an interaction with an infinitely large heat bath. But to utilise all degrees of freedom of such a bath one must either manipulate them all at once, which requires an infinite resource of control complexity, or one must approach the quasistatic limit, which requires an infinite resource of time. It stands to reason that,  in a realistic protocol,  only finitely many degrees of freedom of the bath can be accessed and so the system of interest effectively interacts with a finite, bounded,  thermal bath. Such a bath is represented by a Gibbs state with a non-vanishing temperature, which has full rank. It is a simple task to show that if the interaction between the system of interest and the finite thermal bath is a rank non-decreasing channel---such as a unitary channel---acting in the compound of system-plus-bath,  then the rank of the system  cannot be reduced unless infinite energy is spent.  It follows that if the input state of the system is full-rank, for example if it is a Gibbs state with a non-vanishing temperature, then the third law thus construed will only allow for such a state to be transformed to another full-rank state. 

The above formulation has some drawbacks, however. First, the argument relies on the strong assumption that the system interacts with a thermal environment, which is not justified under purely operational grounds; the environment may in fact be an out of equilibrium system.  Second, the rank non-decreasing condition can only be properly applied to channels  with an identical input and output: the rank of a state  on $\h$ only has meaning in relation to the dimension of $\h$. Indeed,  the partial trace channel (describing the process by which one subsystem is discarded) and the composition channel (describing the process by which the system of interest is joined with an auxiliary system initialised in some fixed state) are physically relevant transformations that must also be addressed, but lead to absurdities when the change in the state's rank is examined. The partial trace channel is  rank-decreasing,  but tracing out one subsystem of a global  full-rank state can only prepare a state that has full rank in the remaining subsystem. On the other hand, the composition channel  is rank-increasing. But it is simple to show that if the rank of the auxiliary state is sufficiently small, then a unitary channel can be applied on the compound so as to purify the system of interest. We thus propose the following minimal definition for channels constrained by the third law, which is conceptually sound, and which does not rely on any assumptions regarding the environment and how it interacts with the system under study, and accounts for the most general class of channels:

\begin{definition}\label{defn:third-law}
A channel $\Phi : \lo(\h) \to \lo(\kk)$ is constrained by the third law if for every full-rank state $\rho$ on $\h$, $\Phi(\rho)$ is a full-rank state on $\kk$. 
\end{definition}
Properties of channels obeying the above definition are given in \app{app:channel-third-law-properties}, and  as shown in \app{app:third-law-channel-proof},  if we are able to implement any channel that is constrained by the third law, then the added ability to implement a channel  not so constrained, that is, a channel that may map some full-rank state to a low-rank state, is both necessary and sufficient for preparing a system in a pure state, given any unknown initial state $\rho$. Moreover, note that while a rank non-decreasing channel acting in $\h$ satisfies \defref{defn:third-law}, a channel acting in $\h$ and which satisfies such a definition need not be rank non-decreasing: a channel constrained by the third law may reduce the rank of some input state, but only if such a state is not full-rank. Finally, \defref{defn:third-law} has the  benefit that for any pair of channels $\Phi_1$ and $\Phi_2$ satisfying such property, where the output of the former corresponds with the input of the latter, so too does their composition $\Phi_2 \circ \Phi_1$; the set of channels constrained by the third law is thus closed under composition.  

\defref{defn:third-law} also allows us to re-examine the constraints imposed by the third law on state preparations, without modeling a finite-dimensional environment prepared in a Gibbs state, or assuming that the system interacts with such an environment by a rank non-decreasing channel. A state preparation is a physical process so that, irrespective of what input state is given, the output is prepared in a unique state $\rho$; indeed, an operational definition of a state is precisely the specification of procedures, or transformations, that produce it.  As stated in  Ref. \cite{Gour2020b}, \emph{``A quantum state can be understood
as a preparation channel, sending a trivial quantum system
to a non-trivial one prepared in a given state''}. That is, state preparations on a Hilbert space $\h$ may be identified with the  set of preparation channels
\begin{align*}
\mathscr{P}(\h) := \{\Phi : \lo(\co^1) \to \lo(\h)   \}   . 
\end{align*} 
Here, the input space is a 1-dimensional Hilbert space $\co^1 \equiv \co |\Omega\>$, and the only state on such a space is the rank-1 projection $|\Omega\>\<\Omega|$. The triviality of the input space captures the notion that the output of the channel $\Phi$ is independent of the input, and so the prepared state $\rho = \Phi(|\Omega\>\<\Omega|)$ is uniquely identified with the channel  itself.    Without any constraints, all states $\rho$ on $\h$ may be prepared by some $\Phi \in \mathscr{P}(\h)$. But now we may restrict the class of preparations by the third law as follows:  $\mathscr{P}(\h)$ is constrained by the third law if all $\Phi \in \mathscr{P}(\h)$ map full-rank states to full-rank states as per \defref{defn:third-law}. But note that $|\Omega\>\<\Omega|$ has full rank in $\co^1$, and so $\rho $ is guaranteed to be full-rank in $\h$.

\section{Quantum measurement}
Before investigating how the third law constrains quantum measurements, we shall first cover briefly some basic elements of quantum measurement theory which will be used in the sequel  \cite{PaulBuschMarianGrabowski1995, Busch1996, Heinosaari2011, Busch2016a}.  

\subsection{Observables}
Consider a quantum system $\s$ with a Hilbert space $\hs$ of finite dimension $2 \leqslant \dim(\hs) < \infty$. We denote by   $\zero$ and $\onesys$ the null and identity operators on $\hs$, respectively, and an operator $E$ on $\hs$ is called an \emph{effect} if it holds that $\zero \leqslant E \leqslant \onesys$. An observable of  $\s$  is represented by  a normalised positive operator valued measure (POVM) $\E : \Sigma \to \mathscr{E}(\hs)$, where $\Sigma$ is a sigma-algebra of some value space $\xx$, representing the possible measurement outcomes, and $\mathscr{E}(\hs)$ is the space of effects on $\hs$.  We  restrict ourselves to discrete observables for which  $\xx := \{x_1, x_2, \dots \}$ is countable.   In such a case we may identify an observable with the set  $\E:= \{\E_x: x \in \xx\}$, where  $\E_x \equiv \E(\{x\})$ are the (elementary) effects of $\E$ (also called POVM elements) which satisfy  $\sum_{x\in \xx} \E_x = \onesys$. The probability of observing outcome $x$ when measuring $\E$ in the state $\rho$ is given by the Born rule as $p^\E_\rho(x) := \tr[\E_x \rho]$.

Without loss of generality, we shall always assume that $\E_x \ne \zero$, since for any $x$ such that $\E_x = \zero$, the outcome $x$ is never observed, i.e., it is observed with probability zero; in such a case we may simply replace $\xx$ with the smaller value space $\xx \backslash \{x\}$. Additionally, we shall always assume that the observable is non-trivial, as trivial observables cannot distinguish between any states, and are thus uninformative; an effect is trivial if it is proportional to the identity, and an  observable  is non-trivial if at least one of its effects is not trivial.

We shall employ the short-hand notation $[\E, A]=\zero$ to indicate that the operator $A$ commutes with all effects of $\E$, and  $[\E, \F]=\zero$ to indicate that all the effects of observables $\E$ and $\F$ mutually commute. An observable $\E$ is commutative if  $[\E,\E]=\zero$, and a commutative observable is also sharp if additionally  $\E_x \E_y = \delta_{x,y} \E_x$, i.e.,  if  $\E_x$ are mutually orthogonal projection operators. Sharp observables are also referred to as projection valued measures, and by the spectral theorem a sharp observable may be represented by a self-adjoint operator   $A = \sum_x \lambda_x \E_x$, where $\{\lambda_x\} \subset \re$ satisfies $\lambda_x \neq \lambda_y$ for $x\neq y$. 
An observable that is not sharp will be called unsharp. An observable $\E$ has the norm-1 property  if it holds that $\|\E_x\|=1$ for all $x$, where $\| \cdot \| $ denotes the operator norm. In finite dimensions,  each effect of a norm-1 observable has at least one eigenvector with eigenvalue 1. While sharp observables are trivially norm-1, this property may also be enjoyed by some unsharp observables.

We now introduce definitions for  classes of observables that are of particular significance to our results:

\begin{definition}\label{defn:small-rank}
An observable $\E:= \{\E_x : x\in \xx\}$ is called  ``small-rank'' if there exists some $x\in \xx$ such that $\E_x$ has rank 1. An observable is called ``large-rank'' if it is not small-rank.  
\end{definition}
In particular, a sub-class of small-rank observables are called rank-1, for which every effect has rank 1 \cite{Holland1990, Pellonpaa2014}. For example, the effects of a sharp observable represented by a non-degenerate self-adjoint operator $A = \sum_x \lambda_x |\psi_x\>\<\psi_x|$ are the rank-1 projections $\E_x = |\psi_x\>\<\psi_x|$. Such observables are therefore rank-1, and hence small-rank. On the other hand,  a sharp observable represented by a degenerate self-adjoint operator such that the eigenspace corresponding to each (distinct) eigenvalue has dimension larger than 1 is large-rank, as each effect is a projection with rank larger than 1. 

\begin{definition}\label{defn:non-degenerate}
An observable $\E:= \{\E_x : x\in \xx\}$ is called ``non-degenerate'' if there exists some $x\in \xx$ such that  there are no multiplicities in the strictly positive eigenvalues of $\E_x$. An observable is called ``degenerate'' if it is not non-degenerate. 
\end{definition}
An example of a non-degenerate observable is a small-rank observable, since in such a case there exists an effect that has exactly one strictly positive eigenvalue. On the other hand,  a large-rank sharp observable  is degenerate, since in such a case  each effect has more than one eigenvector with eigenvalue 1.

\begin{definition}\label{defn:complete-unsharp}
An observable $\E:= \{\E_x : x\in \xx\}$ is called ``completely unsharp'' if for each $x\in \xx$, the spectrum of $\E_x$ does not contain either 1 or 0.
\end{definition}
Completely unsharp observables evidently do not have the norm-1 property, since it holds that $\| \E_x\| < 1$ for all $x$. But since the effects  also do not have eigenvalue 0, then  the effects are in fact full-rank. It follows that completely unsharp observables are also large-rank. But a completely unsharp observable may be either degenerate or  non-degenerate.

As a simple illustrative example of when an observable may or may not satisfy the aforementioned properties, let us consider the case where the system is a qubit, $\hs = \co^2$, with the family of binary observables $\E^{(\lambda)}:= \{\E^{(\lambda)}_+, \E^{(\lambda)}_-\}$ defined by
\begin{align}\label{eq:example-qubit-binary}
\E^{(\lambda)}_\pm := \frac{1}{2} \left(\onesys \pm \lambda \sigma_z \right),   
\end{align}
where $0< \lambda  \leqslant 1$ and $\sigma_z$ is the Pauli-Z operator. Note that if $\lambda =0$, then $\E^{(\lambda)}$ is a trivial observable since in such a case $\E^{(\lambda)}_\pm  = \onesys /2$.  Since $\E^{(\lambda)}$ are binary, so that $\E^{(\lambda)}_- = \onesys - \E^{(\lambda)}_+$,  then they are always commutative, i.e., $[\E^{(\lambda)}_+, \E^{(\lambda)}_-]=\zero$. Now note that the spectrum of each effect is $\{(1 + \lambda)/2, (1 - \lambda )/2\}$. If $\lambda =1$, then the spectrum simplifies to  $\{1,0\}$, i.e., each effect has one eigenvector with eigenvalue 1, and one eigenvector with eigenvalue 0. In such a case,  $\E^{(\lambda)}$ is a norm-1, sharp, small-rank, and non-degenerate observable. On the other hand, for any $0 < \lambda < 1$,  $\E^{(\lambda)}$ is a completely unsharp, large-rank, and non-degenerate observable. Note that for qubits, the only situation where an effect can be degenerate is when the effect is trivial, i.e., when $\lambda = 0$.

\subsection{Instruments}

An instrument \cite{Davies1970}, or operation valued measure, describes how a system is transformed upon measurement, and is  given as a collection of operations (completely positive trace non-increasing linear maps) $\ii:= \{\ii_x : x\in \xx\}$ such that $\ii_\xx(\cdot) := \sum_{x\in \xx} \ii_x(\cdot)$ is a channel. Throughout, we shall always assume that the instrument acts in $\hs$, i.e.,  that both the input and output space of $\ii_x$ is $\hs$. An instrument $\ii$ is identified with a unique observable  $\E $ via the relation  $\tr[\ii_x(\rho)] = \tr[\E_x\rho]$ for all outcomes $x$ and states $\rho$, and we shall refer to such $\ii$ as an $\E$-compatible instrument, or an $\E$-instrument for short, and to $\ii_\xx$ as the corresponding $\E$-channel \cite{Heinosaari2015}. Note that while every instrument is identified with a unique observable, every observable $\E$ admits infinitely many $\E$-compatible instruments; the operations of the L\"uders instrument $\ii^L$ compatible with $\E$ are written as  
\begin{align}\label{eq:Luders}
\ii^L_x(\cdot) := \sqrt{\E_x} \cdot \sqrt{\E_x},    
\end{align} 
and it holds that the operations of every $\E$-compatible instrument $\ii$  can be constructed as $\ii_x = \Phi_x \circ \ii^L_x$, where $\Phi_x$ are arbitrary channels acting in $\hs$ that may depend on outcome $x$ \cite{Ozawa2001,Pellonpaa2013a}.

\subsection{Measurement schemes}

A quantum system is measured when it undergoes an appropriate physical interaction with a measuring apparatus so that the transition of some variable of the apparatus---such as the position of a pointer along a scale---registers the outcome of the measured observable. The most general description of the measurement process is given by a \emph{measurement scheme}, which is a  tuple $\mm:= (\ha, \xi, \ee, \Z)$ where $\ha$ is the Hilbert space for (the probe of)  the apparatus $\aa$ and  $\xi$ is a fixed state of  $\aa$, $\ee$ is a channel acting in  $\hs\otimes \ha$ which serves to correlate   $\s$ with $\aa$, and $\Z := \{\Z_x : x\in \xx\}$ is a POVM acting in $\ha$ which is referred to as a   ``pointer observable''. Throughout, we shall always assume that $2 \leqslant \dim(\ha) < \infty$. For all outcomes $x$,  the operations of the instrument $\ii$ implemented by $\mm$ can be written as
\begin{align}\label{eq:instrument-dilation}
    \ii_x (\cdot) = \tra[(\one\sub{\s}\otimes \Z_x) \ee(\cdot \otimes \xi)],
\end{align}
where  $\tra[\cdot]$   is the partial trace  over $\aa$. The channel implemented by $\mm$ is thus $\ii_\xx(\cdot)  = \tra[\ee(\cdot \otimes \xi)]$.  Every $\E$-compatible instrument admits infinitely many \emph{normal} measurement schemes, where $\xi$ is chosen to be a pure state, $\ee$ is chosen to be a unitary channel, and $\Z$ is chosen to be sharp \cite{Ozawa1984}. Von Neumann's model for the measurement process is one such example of a normal measurement scheme. However, unless stated otherwise, we shall consider the most general class of measurement schemes, where $\xi$ need not be pure, $\ee$ need not be unitary, and $\Z$ need not be sharp. 

\section{Measurement schemes  constrained by the third law}
We  now consider how the third law  constrains measurement schemes, and subsequently examine how such constraints limit the possibility of a measurement to satisfy the proprieties  {\bf (i) - (v)}  outlined in the introduction. 

Since the third law only pertains to channels and state preparations,  the only elements of a measurement scheme $\mm:= (\ha, \xi, \ee, \Z)$ that will be limited by the third law are the interaction channel $\ee$, and the apparatus state preparation $\xi$.  By \defref{defn:third-law} and the proceeding discussion,   we therefore introduce the following definition:
\begin{definition}\label{defn:third-law-measurement}
A measurement scheme $\mm:= (\ha, \xi, \ee, \Z)$  is constrained by the third law if the following hold:
\begin{enumerate}[(i)]

    \item $\xi$ is a full-rank state on $\ha$.
    
    \item For every full-rank state $\varrho$ on $\hs \otimes \ha$,  $\ee(\varrho)$ is also a full-rank state.
\end{enumerate}
\end{definition}
Properties of measurement schemes constrained by the third law are given in \app{app:third-law-measurement} and \app{app:fixed-point-measurement}. Note that the third law does not impose any constraints on the measurability of  observables; we may always choose $\mm$ to be a ``trivial'' measurement scheme, where $\ha \simeq \hs$ and $\ee$ is a unitary swap channel, in which case the observable $\E$ measured in the system is identified with the pointer observable $\Z$ of the apparatus, which can be chosen arbitrarily. This is in contrast to the case where a measurement is constrained by conservation laws; by the Yanase condition, the pointer observable is restricted so that it commutes with the apparatus part of the conserved quantity, and it follows that an observable not commuting with the system part of the conserved quantity is measurable  only if it is unsharp, and only if the apparatus preparation has a large coherence in the conserved quantity \cite{Mohammady2021a}. 
 
The measurability of observables notwithstanding, let us note that an instrument implemented by a trivial measurement scheme is also trivial, i.e., it will hold that for all outcomes $x$ and states $\rho$, the operations of $\ii$ satisfy $\ii_x(\rho) = \tr[\E_x \rho] \xi$. Irrespective of what outcome is observed and what the initial state is, the final state is always   $\xi$. In such a case, $\ii$  fails all the properties {\bf (i) - (v)} that are the subject of our investigation. Therefore, whether or not an observable admits an instrument---realisable by a measurement scheme constrained by the third law as per \defref{defn:third-law-measurement}---with such properties remains to be seen: we shall now investigate this.

\subsection{Non-disturbance}

\begin{figure}[htbp!]
\begin{center}
\includegraphics[width=0.45\textwidth]{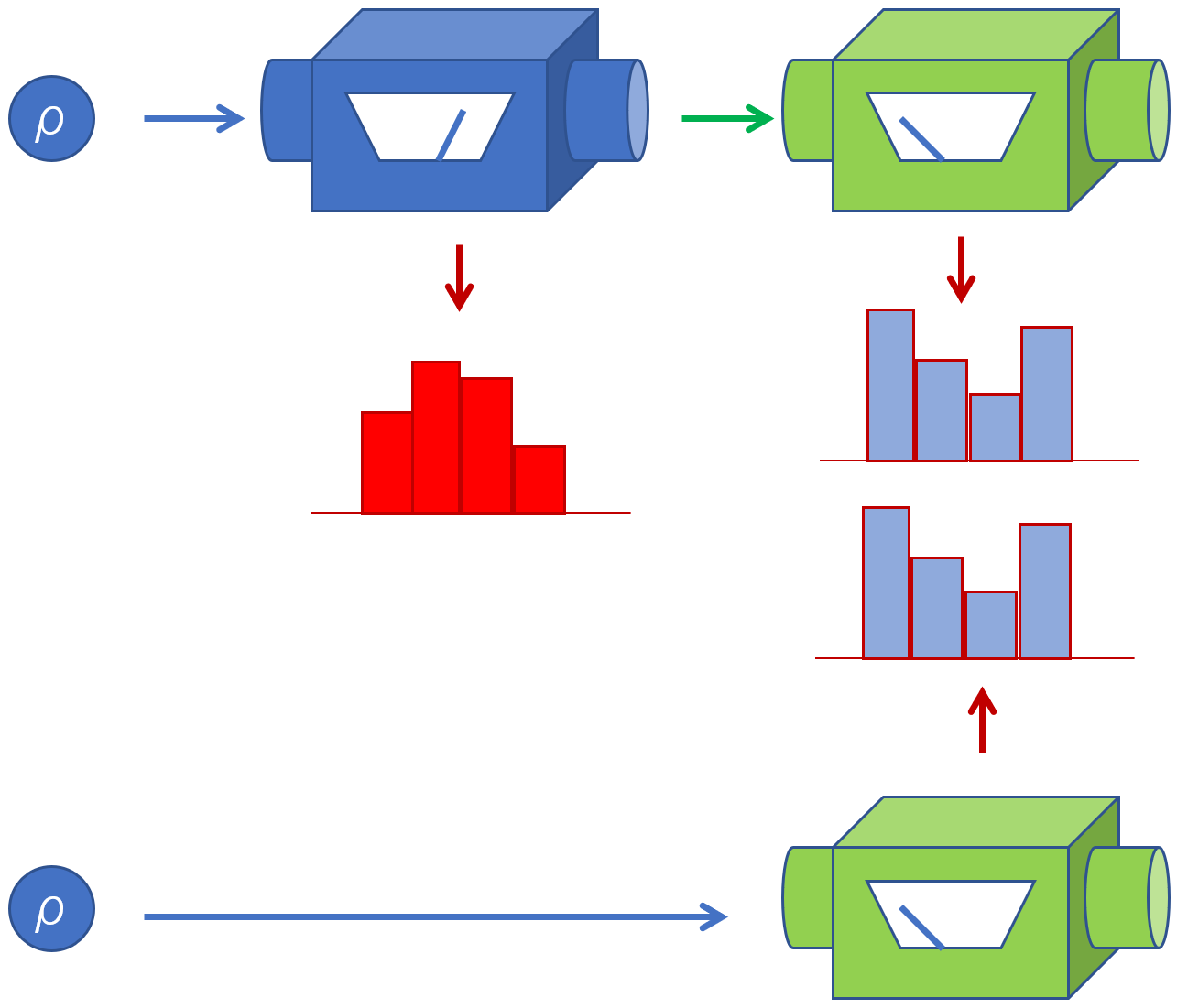}
\vspace*{-0.2cm}
\caption{The top half of the figure represents a sequential measurement of possibly different observables in a system initially prepared in state $\rho$, with the  histograms representing the  statistics obtained for each measurement in the sequence. The bottom half shows the case where the first measurement in the sequence is removed, and only the second measurement takes place. When the statistics of such a measurement are the same in both scenarios, for all states $\rho$, then the first measurement is said to not disturb the second. }\label{fig:non-disturbance}
\vspace*{-0.5cm}
\end{center}
\end{figure}

An $\E$-compatible instrument $\ii$ does not disturb an observable $\F:= \{\F_y : y \in \yy\}$ if it holds that 
\begin{align*}
    \tr[\F_y \ii_\xx(\rho)] = \tr[\F_y \rho]
\end{align*} 
for all states $\rho$ and outcomes $y$ \cite{Heinosaari2010}. In other words,  $\ii$ does not disturb $\F$ if the statistics of $\F$ are not affected by a prior non-selective measurement of $\E$ by $\ii$. Non-disturbance is only possible for \emph{jointly measurable} observables, since in such a case the sequential measurement of $\E$ by $\ii$, followed by a measurement of $\F$, defines a joint observable for $\E$ and $\F$ \cite{Heinosaari2015}. In the absence of any constraints, commutation of $\E$ with $\F$ is sufficient for non-disturbance. That is, if $\F$ commutes with $\E$, 
there exists an $\E$-instrument $\ii$ that does not disturb $\F$. Moreover, the L\"uders $\E$-instrument $\ii^L$ does not disturb \emph{all} $\F$ commuting with $\E$ \cite{Busch1998}. This can be easily shown by the following: if all effects of $\F$ and $\E$ mutually commute, then by \eq{eq:Luders} we may write
\begin{align*}
    \tr[\F_y \ii_\xx^L(\rho)] &= \sum_x \tr[\F_y \sqrt{\E_x} \rho \sqrt{\E_x}] \\
    & = \sum_x \tr[\sqrt{\E_x} \F_y \sqrt{\E_x} \rho ] \\
    & = \sum_x \tr[\E_x \F_y  \rho ]  = \tr[\F_y \rho].
\end{align*} 
In the second line we have used the cyclicity of the trace, and in the third line we use $[\E_x, \F_y] = \zero \iff [\sqrt{\E_x}, \F_y] = \zero$.

While commutation is sufficient for non-disturbance, it is in general not necessary; if $\E$ and $\F$ do not commute but are both sufficiently unsharp so as to be jointly measurable \cite{Miyadera2008},   then it \emph{may} be possible  for a measurement of $\E$ to not disturb $\F$, but not always: while non-disturbance requires joint measurability, joint-measurability does not guarantee non-disturbance. Let us consider an example where non-disturbance is permitted for two non-commuting observables. Consider the case that $\hs = \co^2 \otimes  \co^2$, with the  orthonormal basis $\{|k\rangle \otimes |m\rangle : k,m = 0,1\}$, and define the following family of operators on $\co^2$:
\begin{align*}
A_0 &=|0\rangle \langle 0|, \qquad  A_1=\frac{1}{2}|0\rangle \langle 0|, \qquad  A_2 =\frac{1}{2}|1\rangle \langle 1|,\\
A_3 &=\frac{1}{2} |+\rangle \langle + |, \qquad  A_4 =\frac{1}{2} |-\rangle \langle -|, \qquad  A_5= |1\rangle \langle 1|,  
\end{align*}
where $\ket{\pm} := \frac{1}{\sqrt{2}}(\ket{0} \pm \ket{1})$. Now consider the binary observables $\E:= \{\E_0, \E_1\}$ and $\F:= \{\F_0, \F_1\}$ acting in $\hs$, defined by 
\begin{align*}
\E_0 &= A_0 \otimes |0\rangle \langle 0|
+ (A_2+ A_4) \otimes |1\rangle \langle 1| , \\
\E_1&=(A_1+A_3) \otimes |1\rangle \langle 1|
+A_5 \otimes |0\rangle \langle 0|,
\end{align*}
and 
\begin{align*}
\F_0 &= A_0\otimes |0\rangle \langle 0| 
+ (A_1+A_4)\otimes |1\rangle \langle 1|, \\
\F_1 &= (A_2+A_3)\otimes |1\rangle \langle 1|
+A_5 \otimes |0\rangle \langle 0|.
\end{align*}
One can confirm that $[\E, \F] \ne \zero$. But, we can construct an $\E$-instrument $\ii$ with operations
\begin{align*}
 \ii_0(\rho) &=\tr[\rho (A_0\otimes |0\>\<0| + A_4 \otimes |1\>\<1|)]
 |0 \>\< 0|\otimes |0 \>\< 0| \\
& \quad + \tr[\rho(A_2 \otimes |1\>\<1|)]|1\>\<1|\otimes  |0\>\<0|,
\\
 \ii_1(\rho) &= \tr[\rho( A_5\otimes |0\rangle \langle 0| + A_3 \otimes |1\rangle \langle 1|)]
 |1\rangle \langle 1|\otimes |0\rangle \langle 0|
 \\
 &\quad + \tr[\rho(A_1 \otimes |1\rangle \langle 1|)]|0\rangle \langle 0|\otimes |0\rangle \langle 0|   , 
 \end{align*}
which does not disturb $\F$.

However, we show that under the third law constraint, commutation is in fact necessary for non-disturbance. That is, if an $\E$-instrument $\ii$ can  be 
implemented by a measurement scheme constrained by the third law, such that $\ii$ does not disturb $\F$, then 
$[\E, \F]=\zero$ must be satisfied. In \app{app:third-law-measurement} we show that for any instrument $\ii$, implemented by a measurement scheme constrained by the third law,  there exists at least one full-rank state $\rho_0$ such that $\ii_\xx(\rho_0) = \rho_0$. In such a case, non-disturbance of $\F_y$ (i.e., $\tr[\F_y \ii_\xx(\rho)] = \tr[\F_y \rho]$ for all $\rho$)
implies non-disturbance of a sharp observable $\P=\{\P_z\}$, where $\P_z$ are the spectral projections of $\F_y$. That is, a sequential measurement of $\E$ by the instrument $\ii$, followed by a measurement of $\P$, is a joint measurement of $\E$ and $\P$. Since joint measurability implies commutation when either observable is sharp, it follows that $\E$ must commute with $\P$, and hence with $\F_y$, for all $y$. In other words, given the existence of a full-rank fixed state $\rho_0$, then a measurement of $\E$ does not disturb $\F$ only if they commute. See also Proposition 4 of Ref. \cite{Heinosaari2010}.

But  when the measurement of $\E$ is constrained by the third law, we show that  $[\E,\F]=\zero$ is  not sufficient for non-disturbance: the properties of $\E$ impose further constraints. We now present our first main result:

\begin{theorem}
Under the third law constraint, a completely unsharp observable $\E$ admits a measurement that does not disturb any  observable $\F$ that commutes with $\E$. On the other hand, if an observable $\E$ satisfies  $\|\E_x\| = 1$ for any outcome $x$, then there exists $\F$ which commutes with $\E$ but is disturbed by any measurement of $\E$ that is constrained by the third law. 
\end{theorem}
That is, an $\E$-compatible instrument $\ii$ admits a measurement scheme $\mm$ that is constrained by the third law, such that $[\E,\F] = \zero \implies \tr[\F_y \ii_\xx(\rho)] = \tr[\F_y \rho]$ for all $y$ and $\rho$, if $\E$ is completely unsharp and only if $\|\E_x\| <1$ for all outcomes $x$. Note that an observable can satisfy $\| \E_x\| <1$ for all $x$ without being completely unsharp, since such effects can still have 0 in their spectrum.  The proof is presented in \app{app:non-disturbance} (\propref{prop:non-disturbance}). To show sufficiency of complete unsharpness we prove that, given the third law constraint,  an observable admits a L\"uders instrument if and only if it is completely unsharp (\propref{prop:luders-completely-unsharp}). But since L\"uders measurements are guaranteed to not disturb any commuting observable, the claim immediately follows. On the other hand, the necessity that the effects have norm smaller than 1 follows from the following: if any effect of $\E$ has eigenvalue 1, the projection onto such   eigenspace commutes with $\E$ but is shown to be disturbed. In particular, this implies that  when a norm-1 observable (such as a sharp observable) is measured under the third law constraint, then there exists some observable $\F$ that commutes with $\E$ but is nonetheless disturbed.

Of course, even if a sharp or norm-1 observable $\E$ fails the strict non-disturbance condition, this does not imply that some non-disturbed observables do not exist. In  \app{app:non-disturbance},  we show that if an observable is small-rank as per \defref{defn:small-rank}, then it holds that a third-law constrained measurement of such an observable will disturb all  observables, even if they commute.  Second, we show that if $\E$ is a non-degenerate observable as per \defref{defn:non-degenerate},  then the class of non-disturbed observables will be commutative, and  any pair of non-disturbed observables will commute. That is, for any $\F$ and $\G$ that are non-disturbed, then it will hold that $[\F, \G] = [\F,\F] = [\G,\G] = \zero$.  In other words, non-degeneracy of the measured observable will spoil the ``coherence'' of the measured system. Therefore, to ensure that a measurement of $\E$ does not disturb a non-trivial class of (possibly non-commutative) observables, then   $\E$ must be a large rank (and degenerate) observable.  

For the binary qubit observables $\E^{(\lambda)}$ introduced in \eq{eq:example-qubit-binary},  there exists a third law constrained  measurement of $\E^{(\lambda)}$ such that all commuting observables $\F$ are non-disturbed if and only if $\lambda < 1$, in which case $\E^{(\lambda)}$ are completely unsharp.  But note that since $\E^{(\lambda)}$ is always a degenerate observable when $0 < \lambda <1$, then the non-disturbed observables $\F$ must also be commutative. On the other hand, if $\lambda = 1$, then $\E^{(\lambda)}$ is a small-rank observable and so its measurement will disturb all observables.   In \app{app:non-disturbance},  we construct an explicit example  where the measurement of a sharp observable that is large-rank, and hence degenerate,  will not disturb a non-trivial class of possibly non-commutative observables. This is a binary observable $\E$ acting in a two-qubit system $\hs = \co^2 \otimes \co^2$, defined by $\E_x = \one \otimes |x\>\<x|$. Note that these effects have rank 2, and are hence also degenerate.  In such a case, any observable $\F$ with effects $\F_y \otimes \one$ will be non-disturbed, and it may be the case that $[\F,\F] \ne \zero$.

\subsection{First-kindness}

\begin{figure}[htbp!]
\begin{center}
\includegraphics[width=0.45\textwidth]{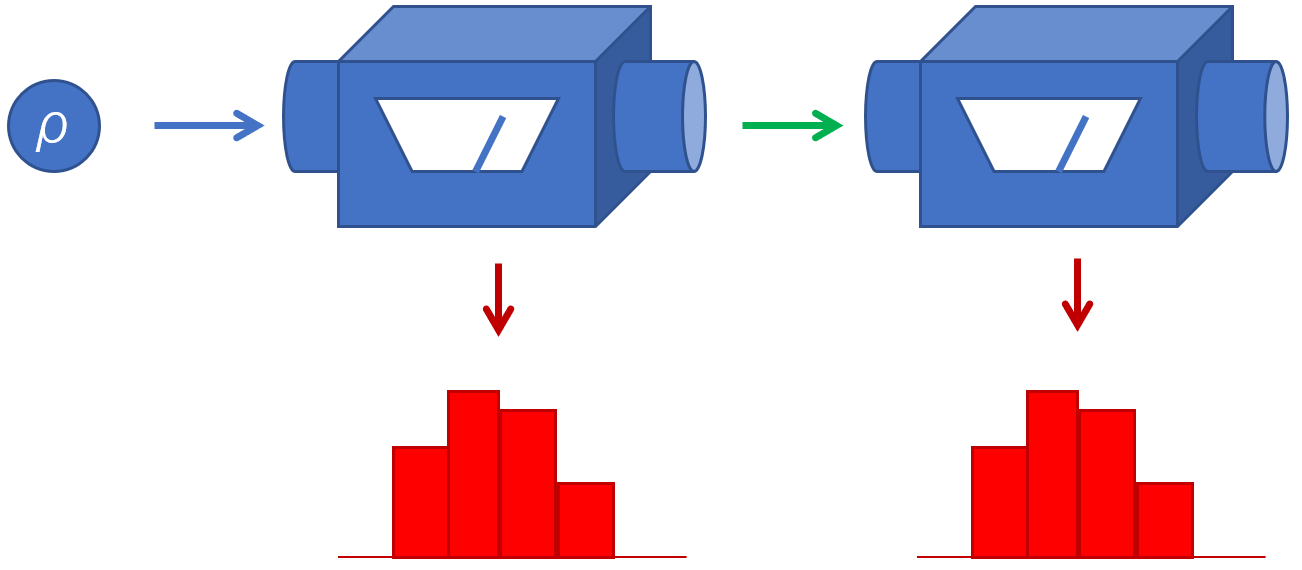}
\vspace*{-0.2cm}
\caption{When the same observable is measured in succession, and when the statistics of the second measurement are the same as those of the first, for all input states $\rho$, then such a measurement is said to be of the first kind.  }\label{fig:first-kindness}
\vspace*{-0.5cm}
\end{center}
\end{figure}

An $\E$-compatible instrument $\ii$ is a measurement of the first kind if $\ii$ does not disturb $\E$ itself, i.e., if it holds that
\begin{align*}
    \tr[\E_x \ii_\xx(\rho)] = \tr[\E_x \rho]
\end{align*} 
for all states $\rho$ and outcomes $x$ \cite{Lahti1991}.  In the absence of any constraints, commutativity of an observable is sufficient for it to admit a first-kind measurement; for any  observable $\E$ such that $[\E,\E]=\zero$ holds, the corresponding  L\"uders instrument is a  measurement of the first kind. This follows from analogous reasoning to that given above.  But we show that,  under the third law constraint, commutativity is necessary for first-kindness, but not sufficient. We now present our second main result:
\begin{theorem}\label{thm:first-kind}
Under the third law constraint, an observable $\E$ admits a measurement of the first kind if and only if  $\E$ is  commutative and completely unsharp.
\end{theorem}
In particular, note that a third law constrained measurement of any norm-1 observable, such as a  sharp observable, necessarily  disturbs itself. The proof is given in \app{app:first-kindness} (\propref{prop:first-kindness}). The sufficiency follows from the fact that any completely unsharp observable admits a L\"uders instrument, as discussed above. On the other hand,  the following is a sketch of the proof for the necessity of such a condition: a non-selective measurement constrained by the third law always leaves some full-rank state $\rho_0$ invariant.  Non-disturbance of $\E$ therefore demands commutativity, as discussed above. But every commutative observable $\E$ is a classical post processing of a sharp observable $\P$, i.e., we may write $\E_x = \sum_y p(x|y) \P_y$ where $\{p(x|y)\}$ is a family of non-negative numbers satisfying $\sum_x p(x|y) = 1$ for every $y$ \cite{Heinosaari2011a}. Given that $\ii_\xx$ has a full-rank fixed state, then if $\ii$ is a first-kind measurement,  $\P$ is also not disturbed \cite{Mohammady2021a}. Therefore,  a sequential measurement of $\E$ by $\ii$ followed by measurement of $\P$ defines a joint measurement of $\E$ and $\P$. By \eq{eq:instrument-dilation}, we obtain for every $\rho$ the following: 
\begin{align*}
  \tr[\P_y \E_x \P_y  \rho]  = \tr[\P_y \otimes \Z_x \ee(\rho \otimes \xi)].
\end{align*}
Now assume that $\rho$ is full-rank. Given that a third law constrained measurement employs a full-rank apparatus preparation $\xi$, while $\ee$ obeys \defref{defn:third-law}, then $\ee(\rho \otimes \xi)$ is full-rank. It follows that the term on the right hand side is strictly positive, and hence so too is the term on the left. But this implies that    $\P_y \E_x \P_y > \zero$, and so  $0 < p(x|y) < 1$, for all $x,y$. Therefore,  $\E$ is completely unsharp.
 
For the binary qubit observables $\E^{(\lambda)}$ introduced in \eq{eq:example-qubit-binary},   there exists a third law constrained  measurement of the first kind  if and only if  $\lambda < 1$, in which case $\E^{(\lambda)}$ are commutative and completely unsharp.  
 In  \app{app:first-kindness}, we construct an explicit example of a first-kind measurement (not given by a L\"uders instrument) of a  commutative and completely unsharp observable.   We consider a system  $\hs = \co^N$ with orthonormal basis $\{|n\>: n=1, \dots, N\}$, and an observable $\E:=\{\E_x: x=1,\dots,N\}$ acting in $\hs$ given by the effects $\E_x = \sum_n p(n|x) |n\>\<n|$. Here,  $ p(n|x) = q(x \ominus n)$, where $\ominus$ denotes subtraction modulo $N$, with $q(n)$ some arbitrary probability distribution satisfying $0<q(n)<1$ for all $n$.   Such an observable is commutative and completely unsharp.

\subsection{Repeatability}

\begin{figure}[htbp!]
\begin{center}
\includegraphics[width=0.45\textwidth]{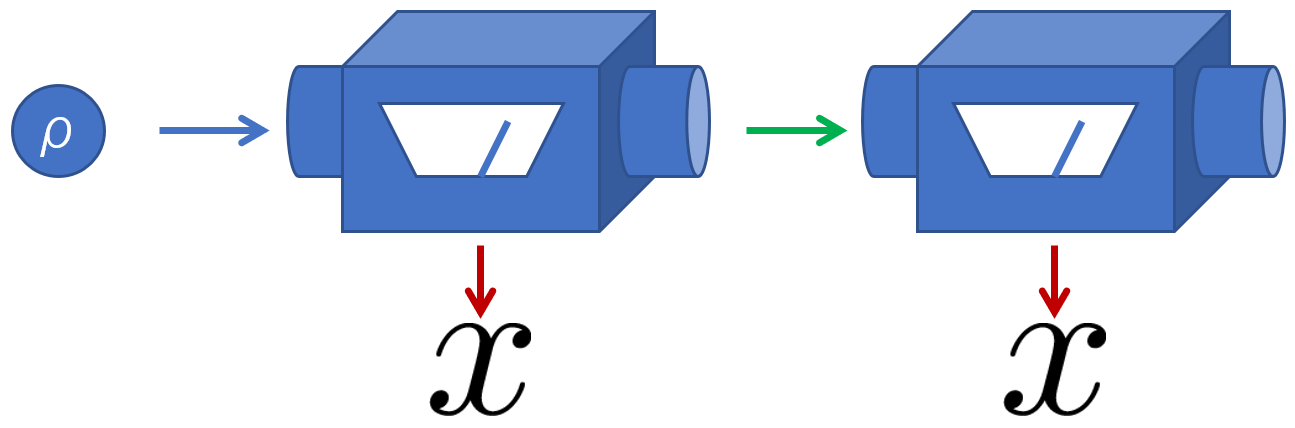}
\vspace*{-0.2cm}
\caption{When the same observable is measured  in succession, and when the outcome obtained by the second is guaranteed with probabilistic certainty to coincide with that of the first, for all input states $\rho$, then such a measurement is said to be repeatable.}\label{fig:repeatability}
\vspace*{-0.5cm}
\end{center}
\end{figure}

An $\E$-compatible instrument $\ii$ is  repeatable if it holds that
 \begin{align*}
\tr[\E_y \ii_x(\rho)] = \delta_{x,y}\tr[\E_x \rho]    
 \end{align*}
 for all states $\rho$ and outcomes $x,y$ \cite{Busch1995,Busch1996b}. In other words, an instrument $\ii$ is a repeatable measurement of $\E$ if a second measurement of $\E$ is guaranteed (with probabilistic certainty) to produce the same outcome as $\ii$. It is simple to verify that repeatability implies first-kindness, since if $\ii$ is repeatable, then we have
 \begin{align*}
\tr[\E_y \ii_\xx(\rho)] = \sum_x \tr[\E_y \ii_x(\rho)] = \tr[\E_y \rho].    
 \end{align*}
While a first-kind measurement need not be repeatable in general, repeatability and first-kindness coincide for the class of sharp observables (Theorem 1 in Ref. \cite{Lahti1991}). For example, if $\E$ is commutative then the corresponding L\"uders instrument is a measurement of the first kind, but such an instrument is repeatable if and only if $\E$ is sharp; note that $\tr[\E_x \ii^L_x(\rho)] = \tr[\E_x^2 \rho]$, which satisfies the repeatability condition if and only if $\E_x^2 = \E_x$. 

An observable $\E$ admits a repeatable instrument only if it is   norm-1,  and in the absence of any constraints, all  norm-1 observables admit a repeatable instrument. For example,  if $\E$ is a possibly unsharp observable with the norm-1 property, and if $\ket{\psi_x}$ are eigenvalue-1 eigenvectors of the effects $\E_x$, then an instrument with operations $\ii_x(\rho) = \tr[\E_x \rho] |\psi_x\>\<\psi_x|$ is repeatable. Note that if the system is a qubit, then only sharp observables admit repeatable measurements. For example,    the binary qubit observables $\E^{(\lambda)}$ introduced in \eq{eq:example-qubit-binary} are norm-1  if and only if     $\lambda = 1$, in which case the observable is also sharp.   Now we present  our third main result:
\begin{theorem}
Under the third law constraint, no observable admits a repeatable measurement.
\end{theorem}
This is an immediate consequence of  \thmref{thm:first-kind} which shows that, under the third law constraint,  norm-1 observables do not admit a measurement of the first kind.  Since repeatability is only admitted for norm-1 observables, and since repeatability  implies first-kindness, then the statement follows. In fact, we can show that for every sequence of outcomes $x$ and $y$, there exists a state $\rho$ such that $\tr[\E_y \ii_x(\rho)] >0$. See \corref{cor:repeatability} for further details.

\subsection{Ideality}

\begin{figure}[htbp!]
\begin{center}
\includegraphics[width=0.45\textwidth]{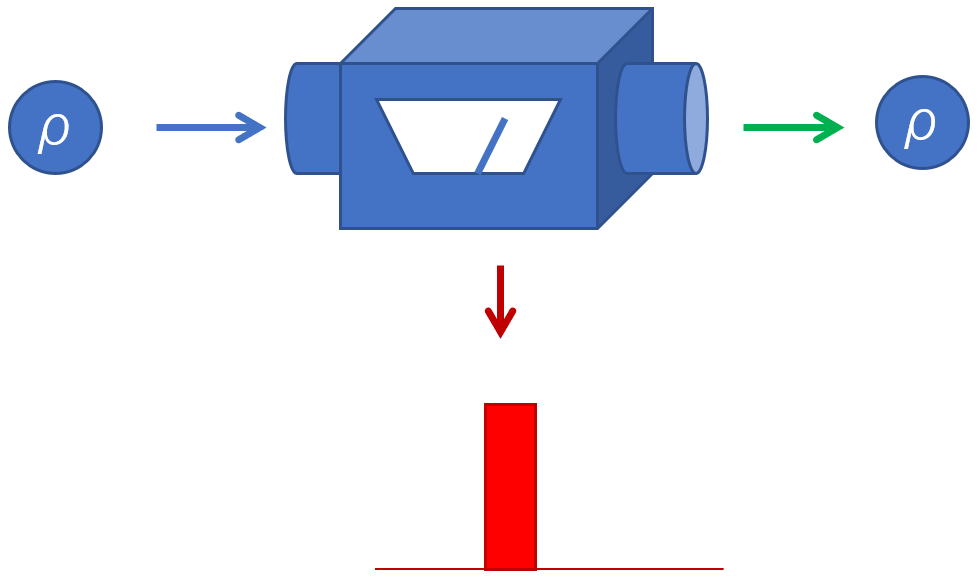}
\vspace*{-0.2cm}
\caption{When an observable is measured in a system such that whenever an outcome can be predicted with certainty, the state of the measured system is unperturbed, then such a measurement is said to be ideal. }\label{fig:ideality}
\vspace*{-0.5cm}
\end{center}
\end{figure}

An instrument  $\ii$ is said to be an ideal measurement of $\E$ if for every outcome $x$ there exists a state $\rho$ such that $\tr[\E_x \rho]=1$, and if for every outcome $x$ and every state $\rho$    the following implication holds: 
\begin{align*}
\tr[\E_x \rho]=1 \implies \ii_x(\rho) = \rho.    
\end{align*} 
That is, $\ii$ is an ideal measurement if it does not change the state of the system whenever the outcome can be predicted with certainty \cite{Busch1990}. Note that ideality can only be enjoyed by norm-1 observables; since  $\tr[\E_x \rho] \leqslant \| \E_x\|$, then any $\E$ that does not enjoy the norm-1 property fails the antecedent of the ideality condition, in which case such condition becomes void. Conversely, in the absence of any constraints all norm-1 observables admit an ideal measurement;  the condition $\tr[\E_x \rho] = 1$ holds if and only if $\rho$ only has support in the eigenvalue-1 eigenspace of $\E_x$, which implies that $\E_x \rho = \rho \E_x = \rho$. But in such a case, we obtain $\ii^L_x(\rho) = \sqrt{\E_x} \rho \sqrt{\E_x} = \E_x \rho = \rho$, and so the L\"uders measurement of a norm-1 observable is ideal.

For the class of sharp observables, the ideal measurements are precisely the L\"uders instruments (see Theorem 10.6 in Ref. \cite{Busch2016a}). Since the third law only permits L\"uders instruments  for completely unsharp observables, then we may immediately infer that ideal measurements of any sharp observable, even those represented by a possibly degenerate self-adjoint operator, are prohibited by the third law.    For example,    the binary qubit observables $\E^{(\lambda)}$ introduced in \eq{eq:example-qubit-binary} are norm-1 if and only if $\lambda=1$, in which case the observable is also sharp. Therefore, such observables never admit an ideal measurement when constrained by the third law.  

However, unsharp observables admit ideal measurements that are not given by the L\"uders instrument. For example, consider a system $\hs = \co^3$ with orthonormal basis $\{\ket{-1}, \ket{0}, \ket{1}\}$. Let $\E := \{\E_+, \E_-\}$ be a binary norm-1 observable acting in $\hs$,   defined by  $\E_\pm = |\pm 1\rangle \langle \pm 1| + \frac{1}{2}|0\rangle \langle 0 |$. It can easily be verified that an instrument with operations 
\begin{align*}
\ii_\pm (\cdot) = \<\pm 1| \cdot |\pm 1\> |\pm 1\>\< \pm 1| + \< 0 | \cdot |0 \> \frac{\onesys }{6}    
\end{align*} 
is an ideal measurement of $\E$.  Therefore, the restriction imposed by the third law on the realisability of  L\"uders instruments does not by itself rule out the possibility of ideal measurements for unsharp norm-1 observables. Now we present our fourth main result:
\begin{theorem}
Under the third law constraint, no observable admits an ideal measurement.
\end{theorem}
The proof is given in \app{app:ideality} (\propref{prop:appendix-ideality}), and the following is a rough sketch. If $\ii$ is an ideal measurement of $\E$, and if $\rho$ is a state for which outcome $x$ can be predicted with certainty, then $\ii_y(\rho) = \zero$ for all $y\ne x$, which implies that $\ii_\xx(\rho) = \rho$. But given the third law constraint, for every state $\rho$ such that $\tr[\E_x \rho] = 1$, it is shown that $\rho$ cannot be a fixed state of $\ii_\xx$, and so $\ii$ cannot be ideal.

\subsection{Extremality}

\begin{figure}[htbp!]
\begin{center}
\includegraphics[width=0.45\textwidth]{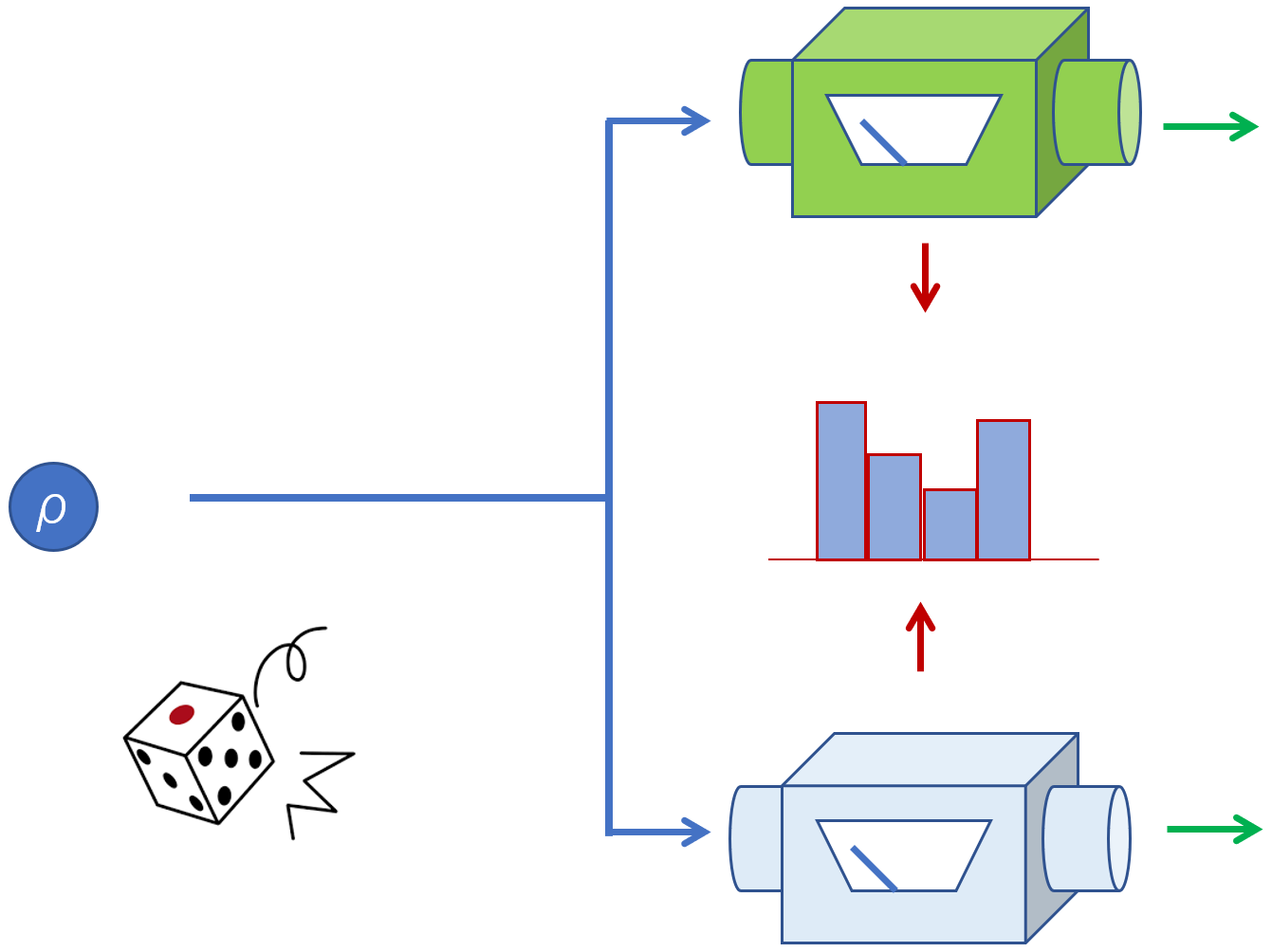}
\vspace*{-0.2cm}
\caption{ A  system may be measured by an instrument  obtained by a probabilistic mixture of two distinct instruments. An extremal instrument is such that   cannot be recovered as a probabilistic mixture of two distinct instruments.     }\label{fig:extremality}
\vspace*{-0.5cm}
\end{center}
\end{figure}

 For any fixed value space $\xx$, the set of  instruments is convex. That is, given any $\lambda \in [0,1]$, and any pair of instruments $\ii^{(i)}:= \{\ii^{(i)}_x : x\in \xx\}$, $i=1,2$,  we can construct an instrument $\ii$ with the operations 
\begin{align*}
    \ii_x(\cdot) = \lambda \,  \ii^{(1)}_x(\cdot) + (1-\lambda) \ii^{(2)}_x(\cdot).
\end{align*}  
An instrument $\ii$ is \emph{extremal} when for any $\lambda \in (0,1)$ such a decomposition is only possible if $\ii = \ii^{(1)} = \ii^{(2)}$. Intuitively, this implies that an extremal instrument is ``pure'', whereas a non-extremal instrument suffers from ``classical noise''. For an in-depth analysis of extremal instruments and their properties, see Refs. \cite{DAriano2011,Pellonpaa2013}.  A simple example of an extremal instrument is the L\"uders instrument compatible with an observable  with linearly independent effects. Since such linear independence is trivially satisfied for norm-1 observables, then  their corresponding  L\"uders instruments  are   extremal.  But it is also possible for the effects of a completely unsharp observable to be linearly independent.  For example, a binary observable $\E:= \{\E_0, \E_1\}$ acting in  $\hs = \co^2$, defined as $\E_0= 3/4|0\rangle \langle 0|+1/4 |1\rangle \langle 1|$ and $\E_1=\one - \E_0$, is completely unsharp with linearly independent effects. Indeed, the effects of the qubit observables $\E^{(\lambda)}$ defined in \eq{eq:example-qubit-binary} are linearly independent for any $0< \lambda < 1$.  Since the L\"uders instruments for such  observables are extremal, and can be implemented under the third law constraint, then we can immediately infer that extremality is permitted by the third law.  Now we present our final main result:
\begin{theorem}
Under the third law constraint, an observable $\E$ acting in $\hs$ admits an extremal instrument only if $\rank{\E_x} \geqslant \sqrt{\dim(\hs)}$ for all outcomes $x$, and a measuring apparatus can only implement an extremal instrument  if it interacts with the system with a non-unitary  channel  $\ee$.
\end{theorem}
The proof is given in \app{app:extremality} (\propref{prop:extremality-third-law}). It follows that under the third law constraint, extremality is only permitted for large-rank observables.  Note in particular that since L\"uders measurements of completely unsharp observables may be extremal, then the above result indicates that they are realisable under the third law constraint  only with non-unitary measurement interactions; indeed, our proof for the sufficiency of complete unsharpness for the realisability of L\"uders instruments (\propref{prop:luders-completely-unsharp}) uses a non-unitary interaction channel. Furthermore, note that in contradistinction to the other properties discussed above,  unsharpness of $\E$ is not a necessary condition for extremality. Indeed,  sharp observables with sufficiently large rank admit an extremal instrument, albeit such instruments cannot be L\"uders due to the previous results.  In \app{app:extremality},   we provide a concrete model for an extremal instrument compatible with a binary sharp observable $\E$ acting in $\hs = \co^2 \otimes \co^2$, defined by $\E_x = \one \otimes |x\>\<x|$. Since $\rank{\E_x} = 2 = \sqrt{\dim(\hs)}$, we see that the bound provided in the above theorem is in fact tight.

\section{Discussion}

 We have generalised and strengthened the results of   Ref. \cite{Guryanova2018}, in the finite-dimensional setting, in several ways.  We have considered the most general class of (discrete) observables---both the observable to be measured and the pointer observable for the measuring apparatus---and not just those that are  sharp and rank-1. Moreover, we have considered a more general class of measurement interactions, between the measured system and measuring apparatus, constrained only by our operational formulation of the third law and thus not restricted to the standard  unitary or rank non-decreasing framework. Within the extended setting thus described, we have shown that ideal measurements are categorically prohibited by the third law for all observables and,   \emph{a fortiori}, we showed that the third law dictates that whenever a measurement outcome  can be predicted with certainty, then the state of the measured system is necessarily perturbed upon measurement. Moreover, we showed that the third law also forbids repeatable measurements, where we note that repeatability and ideality coincide only in the case of sharp rank-1 observables. In addition to the aforementioned impossibility statements, however, our results also include possibility statements as regards extremality and non-disturbance:  the third law  allows for an extremal instrument that measures an observable with sufficiently large rank, and for a measurement of a completely unsharp observable so that such a measurement will not disturb any observable that commutes.  
 
Our results have interesting consequences for the role of unsharp observables in the foundations of quantum theory, and  the question: what is real? There are two deeply connected traditional paradigms for the assignment of reality to a system, both of which are formulated with respect to sharp observables: the Einstein-Podolsky-Rosen (EPR) criterion  \cite{Einstein1935}, and the macrorealism criteria of Legget-Garg \cite{Leggett1985}.

The EPR criterion for a  physical property to correspond to an element of reality reads: \emph{``If, without in any way disturbing a system, we can predict with certainty (i.e., with probability equal to unity) the value of a physical quantity, then there exists an element of physical reality corresponding to this physical quantity''}. In other words, the EPR criterion rests on the possibility of ideal measurements: an eigenvalue of some self-adjoint operator exists in a system when the system is in the corresponding eigenstate, so that an ideal measurement of the observable reveals the eigenvalue  while leaving the system in the same state. But the EPR criterion is shown to be in conflict with the third law of thermodynamics: it is in fact  \emph{not} possible to ascertain any property of the system, with certainty,  without changing its state.  As argued by Busch and Jaeger  \cite{Busch2010a, Jaeger2019}, however,  the EPR criterion is  sufficient, but not necessary;  a necessary condition for a property of a system to correspond to an element of reality  is that it must have the capacity of influencing other systems, such as a measuring apparatus, in a way that is characteristic of such property.  Indeed, since the influence the system has on the apparatus may come in degrees---quantified by the probability, or ``propensity'',  for the apparatus to register that such property  obtains a given value in the system---then even an unsharp observable may correspond to an element of ``unsharp reality''. But note that this weaker criterion makes no stipulation as to how the state of the system changes upon measurement, and does not rely on the possibility of ideal measurements:  a property may exist in a system even if its measurement changes the state of the system. Consequently, our results  provide  support for the unsharp reality program of Busch and Jaeger from a thermodynamic standpoint, as it is shown to be compatible with the third law.

On the other hand, Legget and Garg proposed Macrorealism as the conjunction of two postulates:{ \bf(MR)} \emph{Macrorealism per se}, and {\bf (NI)} \emph{Noninvasive measurability}. { \bf(MR)} rests on the notion of \emph{definiteness}, i.e., that at any given  time, a system can only  be in one out of a set of  states that are perfectly distinguishable by measurement of the observable describing the system---for example, an eigenstate corresponding to some eigenvalue  of a self-adjoint operator. On the other hand,  {\bf (NI)} requires that measurement of such observable not influence the statistics of other observables    at later times. In other words, {\bf (NI)} relies on the possibility of a non-disturbing measurement. But we showed that the third law permits non-disturbance only for unsharp observables without the norm-1 property. Since such observables do not admit definite values in any state, i.e., no two states can be perfectly distinguished by a measurement of such observables, the third law is incompatible with the conjunction of  { \bf(MR)} and {\bf (NI)}. It follows that if we want to keep  {\bf (NI)}, then we must drop { \bf(MR)};   once again we are forced to adopt the notion of an unsharp reality.

To be sure, the third law of thermodynamics should not be considered in isolation; a complete analysis of how the laws of thermodynamics constrain channels and  quantum measurements demands that the third law be considered in conjunction with the first  (conservation of energy) and with the second (no perpetual motion of the second kind). Indeed, our operational formulation of the third law is independent of any notion of temperature, energy, or time. We expect that in the complete picture, that is, when the other laws are also taken into account, our generalised formulation will recover the standard notions of the third law in the literature.     It is also an interesting question to ask how our formulation of the third law, and the constraints imposed by such law on measurements, can be extended to the infinite-dimensional setting. A complete operational formulation of channels constrained by the laws of thermodynamics, and for more general systems than those of finite dimension, is thus still an open problem; our work constitutes one part of such a program,  which  extends the research discipline devoted to the ``thermodynamics of quantum measurements''  \cite{Sagawa2009b, Miyadera2011d, Jacobs2012a, Funo2013, Navascues2014a, Miyadera2015a, Abdelkhalek2016, Hayashi2017, Lipka-Bartosik2018, Solfanelli2019, Purves-2020, Hovhannisyan2021, Aw2021, Beyer2021, Danageozian2022, Mohammady2022, Stevens2021}.  While our impossibility results  are expected to carry over to the more complete framework, the question remains as to how our positive claims must be adapted in light of the other laws of thermodynamics: the combined laws may  impose further constraints. Indeed,  as witnessed by the Wigner-Araki-Yanase theorem,  conservation of energy imposes constraints  on the  measurability  of observables that do not commute with the Hamiltonian \cite{E.Wigner1952,Busch2010, Araki1960,Ozawa2002,Miyadera2006a,Loveridge2011,Ahmadi2013b,Loveridge2020a, Mohammady2021a,Kuramochi2022}. This is in contradistinction to the third law, which imposes no constraints on  measurability. We leave such open questions for future work.

\begin{acknowledgments}
The authors wish to thank Leon Loveridge for his invaluable comments, which greatly improved the present manuscript. M.H.M. acknowledges funding from the European Union’s Horizon 2020 research and innovation programme under the Marie Skłodowska-Curie grant agreement No. 801505. T.M. acknowledges financial support from JSPS KAKENHI Grant No. JP20K03732.
\end{acknowledgments}



\appendix

\section{Preliminaries}

Before presenting the proofs for our main results, let us establish some basic notation and definitions. We denote the algebra of linear operators on a finite-dimensional complex Hilbert space $\h$ as $\lo(\h)$. For any subset $\mathscr{A} \subseteq \lo(\h)$, 
we denote  the \emph{commutant} as 
\begin{align*}
\mathscr{A}':=\{B\in \lo(\h): [A, B]=\zero \, \forall \, A\in \mathscr{A}\}.    
\end{align*}

A ``Schr\"odinger picture''  operation is defined as a completely positive (CP), trace non-increasing linear map  $\Phi : \lo(\h) \to \lo(\kk)$, where $\h$ is the input space and $\kk$ is the output space.  When both input and output spaces are the same, i.e.,  $\h = \kk$, we say that $\Phi$ acts in $\h$. The associated  ``Heisenberg picture'' dual operation is  a completely positive  linear map $\Phi^* : \lo(\kk) \to \lo(\h)$, defined by the trace duality $\tr[A \Phi(B)] = \tr[\Phi^*(A) B]$ for all $A\in \lo(\kk)$ and $B \in \lo(\h)$. $\Phi^*$ is sub-unital, i.e., $\Phi^*(\one\sub{\kk}) \leqslant \one\sub{\h}$, and is unital when the equality holds, which is the case exactly when $\Phi$ is a channel, i.e., when $\Phi$ preserves the trace. We shall also refer to unital CP maps $\Phi^*$ as channels.

The rank of an  operator $A \in \lo(\h)$  is defined as $\rank{A}:= \dim (\mbox{Im}A) = \dim (A\h)$, which coincides with the number of strictly positive eigenvalues of $A^*A$. A positive operator $A > \zero$ is said to be full-rank, or to have full rank in $\h$, if all of its eigenvalues are strictly positive. That is,  $A$ is full-rank in $\h$ when  $\rank{A} = \dim(\h)$. A  state on $\h$ is a positive operator of unit trace, and a full-rank state is also faithful, where a state $\rho$ is called faithful if for any $A \in \lo(\h)$ the following implication holds: $\tr[\rho A^* A] = 0 \iff A = \zero$. Now we  show a useful property of full-rank states which shall be frequently employed in this paper.
\begin{lemma}\label{lemma:full-rank-state-inequality}
Consider the states $\rho$ and $\sigma$ on $\h$. If $\rho$ is full-rank, then there exists $\lambda \in (0,1)$ such that $\rho \geqslant \lambda \sigma$.
\end{lemma}
\begin{proof}
Let  us first note that $\rho - \lambda \sigma \geqslant \zero$ if and only if $\tr[E(\rho - \lambda \sigma)] \geqslant 0$ for all $E \in \mathscr{E}(\h)$, where we define $\mathscr{E}(\h):= \{E \in \lo(\h) : \zero < E \leqslant \one\}$ as the space of (non-vanishing) effects on $\h$. Therefore, let us define 
\begin{align*}
\epsilon &:= \inf_{E \in \mathscr{E}(\h)}\{\tr[E \rho]\} \in [0,1],  \\
\delta &:= \sup_{E \in \mathscr{E}(\h)}\{\tr[E \sigma]\} =1.
\end{align*}
If $\rho$ is full-rank, then it holds that $ \epsilon >0$. Choosing  $\lambda = \epsilon / \delta$,  we thus have $\tr[E(\rho - \lambda \sigma)] \geqslant  \epsilon - \epsilon = 0$ for all $E \in \mathscr{E}(\h)$.  Therefore, there exists $0< \lambda <1$ such that  $\rho \geqslant \lambda \sigma$. 
\end{proof}

\section{Properties of channels constrained by  the third law}\label{app:channel-third-law-properties}

Recall from \defref{defn:third-law} that a channel is constrained by the third law if it maps full-rank states to full-rank states. We now show that the the set of channels constrained by the third law is closed under composition:

\begin{lemma}\label{lemma:third-law-composition}
Consider the channels $\Phi_1 : \lo(\h) \to \lo(\kk)$ and $\Phi_2 : \lo(\kk) \to \lo(\rr)$. If each channel is constrained by the third law, then so too is their composition $\Phi_2 \circ \Phi_1 : \lo(\h) \to \lo(\rr)$.
\end{lemma}
\begin{proof}
Assume that $\Phi_i$ are constrained by the third law. Let  $\rho$ be a full-rank state on $\h$. It follows that $\Phi_1(\rho)$ is a full-rank state on $\kk$, and so $\Phi_2 \circ \Phi_1 (\rho)$ is a full-rank state on $\rr$. As such, $\Phi_2 \circ \Phi_1$ maps all full-rank states on $\h$ to full-rank states on $\rr$, and so by \defref{defn:third-law} is constrained by the third law.
\end{proof}

We shall now prove a  useful result that will allow for equivalent formulations of channels constrained by the third law:

\begin{lemma}\label{lemma:third-law-channel-condition}
Consider a channel $\Phi : \lo(\h) \to \lo(\kk)$. The following statements are equivalent.  
\begin{enumerate}[(i)]
     \item For every full-rank state $\rho$ on $\h$, $\Phi(\rho)$ is a full-rank state on $\kk$. 

    \item There exists a  state $\sigma$ on $\h$ such that $\Phi(\sigma)$ is a full-rank state on $\kk$.

    \item For every $A \in \lo(\kk)$, $\Phi^*(A^* A) = \zero \iff A = \zero$.
    
    \item In the case that $\h = \kk$, there exists a full-rank state $\rho_0$ on $\h$ such that $\Phi(\rho_0) = \rho_0$.

\end{enumerate}
 
\end{lemma}
\begin{proof}
(i) $\implies$ (ii): This is trivial.
\\
(ii) $\implies$ (i):  By \lemref{lemma:full-rank-state-inequality}, for any full-rank state $\rho$, and for any state $\sigma$, there exists $\lambda \in (0,1)$ such that $\rho \geqslant \lambda \sigma$.  Assume that $\Phi(\sigma)$ is full-rank.  It holds that for any full-rank state $\rho$, and for any unit-vector $\ket{\phi} \in \kk$, we have $\<\phi| \Phi(\rho) |\phi\> \geqslant \lambda \<\phi| \Phi(\sigma) |\phi\> >0$, which implies that $\Phi(\rho)$ is full-rank.  
\\
(ii) $\implies$ (iii): For any channel $\Phi^*$, it holds that $\|\Phi^*(A)\| \leqslant \| A\|$, and so $A = \zero \implies \Phi^*(A^* A) = \zero$ follows. Now assume that for some $A$, it holds that $\Phi^*(A^* A) = \zero$. Then  $\tr[\Phi(\sigma) A^*A] = \tr[\sigma \Phi^*(A^* A)]
=0$ for any state $\sigma$. But if $\Phi(\sigma)$ is full-rank, this holds only if $A = \zero$. 
\\
(iii) $\implies$ (ii):  Consider the complete mixture $\sigma=\one/\dim(\h)$, and  let $P \leqslant \one$ be the minimal support projection on $\Phi(\sigma)$, with $P^\perp := \one - P \geqslant \zero$ its orthogonal complement. It follows that 
$\tr[ \Phi^*(P^\perp) \sigma] =  \tr[ P^\perp \Phi(\sigma)]=0$. But by  (iii), this holds if and only if 
$P^\perp =\zero$, and hence $P=\one$. That is, $\Phi$ maps the complete mixture to a full-rank state. 
\\
(i) $\implies$ (iv): Assume that $\h = \kk$, so that for any $n \in \nat$ we may define the channel $\Phi^n : \lo(\h) \to \lo(\h)$ as $n$ consecutive applications of the channel $\Phi$. Now consider the channel
\begin{align*}
    \Phi_\av(\cdot) := \lim_{N \to \infty} \frac{1}{N} \sum_{n=1}^N \Phi^n(\cdot),
\end{align*}
where we note that this limit exists as $\dim(\h) < \infty$. Let us define the state $\rho_0:= \Phi_\av(\sigma)$, where $\sigma = \one / \dim(\h)$ is the complete mixture. Since the complete mixture  is full-rank, then  $\Phi^n(\sigma)$ will be full-rank for all $n \in \nat$, and so it holds that $\rho_0$ is full-rank. But note that
\begin{align*}
 \Phi(\rho_0) &= \lim_{N \to \infty} \frac{1}{N} \sum_{n=1}^N \Phi^{n+1}(\sigma) \nonumber \\
 &=  \lim_{N \to \infty} \frac{1}{N} \bigg(\sum_{n=1}^N \Phi^{n}(\sigma) + \Phi^{N+1}(\sigma) - \Phi(\sigma) \bigg) \nonumber \\
 & = \Phi_\av(\sigma) = \rho_0.
\end{align*}
and so there exists a full-rank state $\rho_0$ such that $\Phi(\rho_0) = \rho_0$. 
\\
(iv) $\implies$ (ii): This is trivial.

\end{proof}

\begin{corollary}\label{corollary:third-law-channel-examples}
The following channels are constrained by the third law:
\begin{enumerate}[(i)]
    \item Rank non-decreasing channels acting in $\h$.
    \item Bistochastic channels acting in $\h$. 
    \item Gibbs-preserving channels acting in $\h$.
    \item Partial trace channels $\tra : \lo(\hs \otimes \ha) \to \lo(\hs)$.
    \item Composition channels     $\Lambda_\xi : \lo(\hs) \to \lo(\hs\otimes \ha)$,  $\rho \mapsto  \rho \otimes \xi$, where  $\xi$ is a full-rank state on $\ha$. 
\end{enumerate}
\end{corollary}
\begin{proof}
\begin{enumerate}[(i)]
    \item This is trivial.

    \item Bistochastic channels preserve both the trace and the identity. Consider the complete mixture $\sigma := \one/\dim(\h)$. It trivially holds that if $\Phi$ is  bistochastic, then $\Phi(\sigma) = \sigma$. Such a channel satisfies property (ii) of \lemref{lemma:third-law-channel-condition}, and is hence constrained by the third law.
    
    \item A Gibbs-preserving channel  $\Phi$ acting in $\h$ satisfies $\Phi(\tau_\beta) = \tau_\beta$ for some $\tau_\beta:= e^{-\beta H} / \tr[e^{-\beta H}]$, where $H = H^*$, $\beta >0$.     Gibbs states $\tau_\beta$, with $\beta >0$,  are full-rank.  It follows that a Gibbs-preserving channel satisfies property (ii) of \lemref{lemma:third-law-channel-condition}, and is hence constrained by the third law.
    
    \item Let $\rho$ and $\sigma$ be full-rank states on $\hs$ and $\ha$, respectively, so that $\omega = \rho \otimes \sigma$ is full-rank in $\hs\otimes \ha$. Given that $\tra[\omega] = \rho$, which is full-rank in $\hs$, it follows that the partial trace satisfies property (ii) of \lemref{lemma:third-law-channel-condition}, and is hence constrained by the third law.
    
    \item If $\xi$ is full-rank, then $\Lambda_\xi(\rho) = \rho \otimes \xi$ is full-rank for all full-rank states $\rho$ on $\hs$. As such, $\Lambda_\xi$ is constrained by the third law.
    
\end{enumerate}
\end{proof}

Given that unitary channels preserve the spectrum, they are clearly rank non-decreasing channels. Since unitary channels are a special subclass of bistochastic channels, one may wonder whether or not all bistochastic channels are rank non-decreasing. The following lemma answers such a conjecture in the affirmative:

\begin{lemma}
Bistochastic channels $\Phi$ acting in $\h$ are rank non-decreasing.
\end{lemma}
\begin{proof}
Consider a state $\rho$, whose minimal support projection is $P$. Since $\rho$ and $ P/\tr[P]$ are both full-rank states in $P \h$,  then by \lemref{lemma:full-rank-state-inequality} there exists $\lambda_1, \lambda_2 >0$ such that 
$\lambda_1 P/\tr[P] \leqslant \rho \leqslant \lambda_2 P/\tr[P]$. Now consider an arbitrary channel $\Phi$, and define the state $\sigma:= \Phi(P/\tr[P])$. By complete positivity, it follows that $\lambda_1 \sigma \leqslant \Phi(\rho) \leqslant \lambda_2 \sigma$. 
Now let us note that $A\geqslant B\geqslant \zero$ implies that $\rank{A} \geqslant \rank{B}$. This follows from the fact that for any $A\geqslant \zero $, the condition 
$|\psi\rangle \in \ker( A)$  is equivalent to $\langle \psi |A |\psi\rangle =0$. We thus obtain 
 $\rank{ \sigma} \leqslant  \rank{\Phi(\rho)} \leqslant \rank{\sigma}$, which implies that 
\begin{align}\label{eq:rank-condition-support-projection}
    \rank{\Phi(\rho)} = \rank{\sigma}.
\end{align} 
Assume that $\Phi$ is bistochastic, so that it preserves the identity. In such a case, if  $\rho$ is full-rank so that $P=\one$ and $\sigma=\Phi(\one/\dim(\h)) = \one / \dim(\h)$ hold, we find by \eq{eq:rank-condition-support-projection} that  $\rank{\Phi(\rho)}=\rank{\rho}$. Now let us  suppose that $\rho$ is not full-rank. By decomposing the identity as 
\begin{eqnarray*}
\one = \tr[P] \frac{P}{\tr[P]}
+ \tr[P^{\perp}] \frac{P^{\perp}}{\tr[P^{\perp}]}, 
\end{eqnarray*}
we  obtain for any bistochastic $\Phi$ the following:  
\begin{eqnarray*}
\one = \Phi(\one) = \tr[P] \sigma +\tr[P^{\perp}] \mu, 
\end{eqnarray*}
where we define $\mu:= \Phi(P^{\perp}/\tr[P^{\perp}])$, which implies that
\begin{eqnarray*}
\tr[P^{\perp}]\mu = \one - \tr[P] \sigma.
\end{eqnarray*}
Since $\mu$ is non-negative, it must hold that $1 \geqslant \tr[P] \| \sigma\| = \rank{\rho} \| \sigma\|$, where the equality follows from the fact that $\rank{\rho} = \tr[P]$. This gives $\rank{\rho} \leqslant 1/ \| \sigma\|$. But it holds that $1 = \tr[\sigma] \leqslant \| \sigma\| \rank{\sigma} = \|\sigma\| \rank{\Phi(\rho)}$, where the equality follows from \eq{eq:rank-condition-support-projection}. This  gives $\rank{\Phi(\rho)} \geqslant 1/ \| \sigma\|$.  We finally arrive at
\begin{align*}
    \rank{\Phi(\rho)} \geqslant \frac{1}{\|\sigma\|} \geqslant \rank{\rho},
\end{align*}
which proves our claim.

\end{proof}

Let us note that while any rank non-decreasing channel $\Phi$ acting in $\h$  is constrained by the third law,  a channel acting in $\h$ that is constrained by the third law may decrease the rank for some input states. To see this, consider $\h \simeq \co^3$, with orthonormal basis $\{\ket{i} : i =0,1,2\}$. Consider the projections   $\pr{i}=|i\>\< i|$, and a full-rank state  $\rho_0$  on $\h$. Now define the channel $\Phi(\rho) = \tr[\rho \pr{0}] \rho_0 + \tr[\rho \pr{0}^\perp] \pr{1}$. This channel clearly maps all full-rank states $\rho$ to a full-rank state, and so is constrained by the third law. However, a rank-2 state $\pr{0}^{\perp}/2$ is mapped to a rank-1 state $\pr{1}$.

\section{Operational justification for  Definition 1}\label{app:third-law-channel-proof}

Consider a system $\h$ with a Hamiltonian that has the spectral decomposition $H = \sum_{i=0}^N \epsilon_i P_i$, with the (distinct) energy eigenvalues arranged in increasing order, i.e., $\epsilon_i < \epsilon_{i+1}$. The  state of such a system at inverse temperature $\beta = 1/T$ is the Gibbs state $\tau_\beta := e^{- \beta H}/\tr[e^{-\beta H}]$. In such a case, $\rank{\tau_\beta} = \dim(\h)$. When cooled to absolute zero temperature, i.e., $\beta=\infty$, the state reads $\tau_\infty = P_0/ \tr[P_0]$, and it holds that $\rank{\tau_\infty} < \dim(\h)$. Indeed, any state whose support is in $P_0 \h$ can be considered as having zero temperature. Therefore, the third law can be seen as prohibiting any channel $\Phi$ such that $\Phi(\tau_\beta) = \tau_\infty$. \defref{defn:third-law} is a generalisation of such an intuitive idea, and below, we shall show that the existence of a channel which is unconstrained by the third law as defined by \defref{defn:third-law}, i.e., a channel that may map a full-rank state on some system to a non-full-rank state of a possibly different system, is both necessary and sufficient for the preparation of some system in a pure state.

\begin{prop}
Let us consider a system $\s$ described by $\hs$ with $2 \leqslant \dim( \hs) < \infty$, and 
suppose that the system is given with an unknown input state $\rho$. 
Assume that we can employ arbitrary channels constrained by 
the third law as per \defref{defn:third-law}.  The following statements are equivalent:  
\begin{enumerate}[(i)]
    \item
    We may implement some known channel that is unconstrained by the third law. 
    
    \item
    We may prepare the system $\s$ in an arbitrary desired state (in particular, an arbitrary 
    pure state), from any given unknown prior state $\rho$. 
\end{enumerate}
\end{prop}

\begin{proof}
(ii) $\implies$ (i): Assume that we can only implement channels that are constrained by the third law.  By \lemref{lemma:third-law-composition},  the class of third-law constrained channels is closed under composition, and so a full-rank input state $\rho$ can only be transformed to another full-rank state. Therefore, the ability to prepare an arbitrary state, such as a  pure state, from any input $\rho$, including one that is full-rank, requires a channel that is unconstrained by the third law. 
\\
(i) $\Rightarrow$ (ii): Consider a system $\hs = \co^N$, with the orthonormal basis $\{\ket{0}, \dots, \ket{N-1}\}$, initially prepared in some unknown state $\rho$.  Consider a channel $\Phi$ acting in $\hs$, defined by $\Phi(\cdot) = \tr[\cdot  \one ]\rho_0$,  where $\rho_0=  \sum_{n=0}^{N-1} \lambda_n |n\rangle \langle n| $ is a known full-rank state. Such a channel maps full-rank states to full-rank states, and so is constrained by the third law. We may therefore  use this channel to prepare $\hs$ in the state  $\rho_0$, independently of the input $\rho$.

Now let us assume that we may implement a known channel  $\ee : \lo(\kk) \to \lo(\rr)$ which is unconstrained by the third law, i.e., assume that there exists a full-rank state $\varrho$  on $\kk$ which is mapped to  state  $\xi$ on $\rr$ that is not full-rank. But preparation channels that prepare a full-rank state are constrained by the third law. As such, we may prepare $\kk$ in such a state $\varrho$, which is known, so that by applying the channel $\ee$, we may  prepare $\rr$ in the state $\xi$, which is also known. 
Let us write such a state as $\xi = \sum_{m=0}^{M -1} p_m |m\rangle \langle m|$, where  $\{\ket{0}, \dots, \ket{M-1}\}$ is an  orthonormal basis that spans $\rr = \co^M$, such that $p_m =0$ for at least one $m$. 

By preparing $D$ copies of $\xi$, we may  write the  state on the total system $\hs\otimes \rr^{\otimes D}$, i.e.,  $\rho_0 \otimes \xi^{\otimes D} = \rho \otimes \xi \otimes \dots \otimes \xi$,  as
\begin{align*}
\rho_0 \otimes \xi^{\otimes D} 
& = \sum_{n=0}^{N-1} \sum_{m_1=0}^{M -1} \cdots \sum_{m_D=0}^{M -1} C_{n,m_1,\ldots, m_D} \\
& \qquad  |n, m_1, \cdots, m_D\rangle \langle n, m_1, \cdots, m_D|, 
\end{align*}
where we define $C_{n,m_1,\ldots, m_D}:=\lambda_n p_{m_1}\cdots p_{m_D}$. Now choose a unitary channel acting in the total system, with a unitary operator $U$ 
which permutes the basis $\{|n, m_1,  \ldots, m_D\rangle \}$. We obtain 
\begin{align*}
U(\rho_0 \otimes \xi^{\otimes D})U^* 
&= \sum_{n=0}^{N-1} \sum_{m_1=0}^{M -1} \cdots \sum_{m_D=0}^{M -1} C_{ \pi(n,m_1,\ldots, m_D)} \\
& \qquad |n, m_1, \cdots, m_D\rangle \langle n, m_1, \cdots, m_D|.  
\end{align*}
By choosing $D$ to be sufficiently large so that  $\rank{\xi}^D N \leqslant M^D$ is satisfied, then  the permutation can be chosen so that $C_{ \pi(n,m_1,\ldots, m_D)} =0$ for $n\neq 0$, and so the above state will read $U(\rho_0 \otimes \xi^{\otimes D})U^* = |0\>\<0| \otimes \omega$.
The restriction of the final state $|0\>\<0| \otimes \omega$ on $\hs$ is thus a pure state $|0\rangle \langle 0|$. To obtain an arbitrary target state $\sigma$, we apply a channel $\Lambda (\cdot) = \langle 0 | \cdot |0\rangle \sigma + \tr[ (\one - |0\rangle \langle 0|) \cdot ]\frac{\one}{N}$, which maps full-rank states to full-rank states, and so is constrained by the third law. 
\end{proof}

\section{Measurements constrained by the third law}\label{app:third-law-measurement}

Consider again the composition channel $\Lambda_\xi:  \lo(\hs) \mapsto \lo(\hs\otimes \ha ),  \rho \mapsto \rho \otimes \xi$. We now define the  restriction map $\Gamma_{\xi}: \lo(\hs\otimes \ha ) \to \lo(\hs)$ as dual to the  composition channel, i.e., $\Gamma_\xi = \Lambda_\xi^*$. The restriction map therefore satisfies $\tr[\Gamma_{\xi} (B) \rho] = \tr[B (\rho \otimes \xi)]$ for all $B \in \lo(\hs\otimes \ha)$ and states $\rho$ on $\hs$.  Using such a map, and a channel $\ee$ acting in $\hs\otimes \ha$, we  define the channel $\Gamma_\xi^\ee: \lo(\hs\otimes \ha) \to \lo(\hs)$ as
\begin{align}\label{eq:gamma-u}
\Gamma_\xi^\ee(\cdot) &:= \Gamma_\xi\circ \ee^*(\cdot).
\end{align}
We may thus write the dual of the operations in \eq{eq:instrument-dilation} implemented by the measurement scheme $\mm:=(\ha, \xi, \ee, \Z)$ as
\begin{align*}
    \ii^*_x (\cdot) &=  \Gamma_\xi^\ee(\cdot \otimes \Z_x),
\end{align*}
to hold for all $x\in \xx$. As such, we may write the channel implemented by $\mm$ as $\ii_\xx^*(\cdot) = \Gamma_\xi^\ee(\cdot \otimes\oneapp)$. Moreover, recall that an instrument is compatible with $\E$ if $\tr[\ii_x(\rho)] = \tr[\E_x \rho]$ for all outcomes $x$ and for all states $\rho$, which may equivalently be written as $\E_x =  \ii^*_x(\onesys)$ for all $x$. We may therefore write the effects of the observable implemented by the measurement scheme $\mm$ as
\begin{align*}
    \E_x = \Gamma_\xi^\ee(\onesys \otimes \Z_x).
\end{align*}

\begin{lemma}\label{lemma:measurement-third-law}
Let $\mm:= (\ha, \xi, \ee, \Z)$ be a measurement scheme for an $\E$-compatible instrument $\ii$ acting in $\hs$. Assume that $\mm$ is constrained by the third law. The following hold: 

\begin{enumerate}[(i)]
\item For every $A \in \lo(\hs \otimes \ha)$, it holds that  $\Gamma_\xi^\ee(A^* A) = \zero \iff A = \zero$. 

\item There exists at least one full-rank state $\rho_0$ on $\hs$ such that $\ii_\xx(\rho_0) = \rho_0$. 

\item For every $A \in \lo(\hs)$ and $x$, it holds that $\ii_x^*(A^*A) = \zero \iff A = \zero$.
    
\item For every $x$, let  $P_x$ be the minimal projection on the support  of $\E_x$. For every  state $\rho$ such that $P_x \rho P_x$  has full rank in $P_x \hs$, $\ii_x(\rho)$ has  full rank in $\hs$. 

\end{enumerate}  
  
\end{lemma}
\begin{proof}

\begin{enumerate}[(i)]
\item By  \defref{defn:third-law-measurement},  $\xi$ is full-rank and $\ee$ maps full-rank states to full-rank states. Since $\Gamma_\xi$ is dual to the composition channel, then by  \corref{corollary:third-law-channel-examples}   $\Gamma_\xi$ is constrained by the third law. By \lemref{lemma:third-law-composition}, it follows that $\Gamma_\xi^\ee = \Gamma_\xi \circ \ee^*$ is constrained by the third law. The statement follows from  \lemref{lemma:third-law-channel-condition}.

    \item Since $\ii_\xx^*(\cdot) = \Gamma_\xi^\ee(\cdot \otimes \oneapp)$, then by the above it follows that $\ii_\xx$ satisfies property (iii), and hence (iv), of \lemref{lemma:third-law-channel-condition}.

    \item Given that $0< \|\E_x\| = \| \Gamma_\xi^\ee(\onesys \otimes \Z_x)\|\leqslant \| \Z_x\|$, it holds that $\Z_x > \zero$ for all $x$.  Now note that $\ii_x^*(A^*A) = \Gamma_\xi^\ee(A^*A \otimes \Z_x)$. By item (i), it follows that $\ii_x^*(A^*A) = \zero$ if and only if $A^*A \otimes \Z_x = \zero$, which holds if and only if $A = \zero$. 
    
    \item We may always write $\ii_x^*(\cdot) = \sqrt{\E_x} \Phi_x^*(\cdot) \sqrt{\E_x}$ for some channel $\Phi_x$ acting in $\hs$. It follows that $\ii_x^*(\cdot) = P_x\ii_x^*(\cdot)P_x$. 
   By (iii), for any $\zero \ne A \in \lo(\hs)$, it holds that $\lo(P_x \hs) \ni \ii_x^*(A^* A) = P_x \ii_x^*(A^* A) P_x > \zero$. But for   any $\rho$ for which $P_x \rho P_x$ has full-rank in $P_x \hs$, it follows that $\tr[\ii_x^*(A^* A) \rho] >0$.  By writing  $\tr[A^* A \ii_x(\rho)] = \tr[\ii_x^*(A^* A) \rho]$, it follows that $\tr[A^* A \ii_x(\rho)] =0 \iff A = \zero$,  and so $\ii_x(\rho)$ must be full-rank in $\hs$.

\end{enumerate}

\end{proof}  

Condition (iv) of the above lemma shows that the third law restricts the possible conditional state transformations by measurements. That is, it is impossible to prepare a system in a state of low rank, given an arbitrary input state, by  measurement and selection of an outcome. Moreover, the class of input states that may be prepared in a state of low rank diminishes as the rank of the observable's effects decrease. Indeed, if an effect $\E_x$ has rank 1, so that it may be written as $\E_x = \lambda |\psi\>\<\psi|$ for some unit-vector $\ket{\psi}$ and $\lambda \in (0,1]$, then for any state $\rho$ such that $\<\psi| \rho | \psi\> >0$,  $\ii_x(\rho)$ will be full-rank.   In particular, we obtain the following result:

\begin{prop}\label{prop:luders-completely-unsharp}
Let $\ii^L$ be a L\"uders instrument compatible with a non-trivial observable $\E$. $\ii^L$ admits a measurement scheme $\mm:=(\ha, \xi, \ee, \Z)$, which is constrained by the third law, if and only if $\E$ is completely unsharp.
\end{prop}
\begin{proof}
First, let us show the only if statement. By item (iv) of \lemref{lemma:measurement-third-law}, $\ii^L_x(\cdot) = \sqrt{\E_x} \cdot \sqrt{\E_x}$ must map full-rank states
to full-rank states. Consider the complete mixture $\sigma = \onesys /\dim(\hs)$, which is full-rank. $\ii_x^L(\sigma) = \E_x/\dim(\hs)$ is full-rank if and only if $\E_x$ is full-rank.   Since $\E$ is non-trivial, and there must be at least two outcomes $x$ for which $\zero < \E_x < \onesys$,  then  the spectrum of all effects $\E_x$ must be contained in $(0,1)$. $\E$ is thus completely unsharp.

Now we shall show the if statement. Let us consider a completely unsharp observable $\E=\{\E_x\}_{x=1}^N$ acting in  
$\hs$. Let the apparatus Hilbert space be $\ha=\co^N$, with an orthonormal basis $\{|y\rangle \}_{y=0}^{N-1}$, and choose the apparatus preparation as the complete mixture $\xi=\oneapp/N$.  
Choose an interaction channel $\ee(\cdot) = \sum_x K_x \cdot K_x^*$, with Kraus operators 
\begin{eqnarray}\label{ChanLuders}
K_x := \sum_a \sqrt{\E_{x\oplus a}}\otimes |x\oplus a \rangle \langle a|, 
\end{eqnarray}
where $\oplus$ represents  summation modulo $N$.  Then  for an arbitrary state $\rho$ on $\hs$, we obtain 
\begin{eqnarray*}
\ee(\rho \otimes \xi)
&=&\sum_x \sum_a \sum_b 
\\
&&
\sqrt{\E_{x\oplus a}}\rho \sqrt{\E_{x\oplus b}}\otimes 
\frac{1}{N} \sum_c |x\oplus a\rangle \langle x\oplus b | \delta_{ac}\delta_{bc}
\\
&=& \sum_x \sqrt{\E_x}\rho \sqrt{\E_x} \otimes |x\rangle \langle x|.
\end{eqnarray*}
To show that  $\ee$ is constrained by the third law, let us note that for $\rho = \onesys/\dim (\hs)$, we have 
\begin{eqnarray*}
\ee(\rho \otimes \xi) = \frac{1}{\dim (\hs)}\sum_x \E_x \otimes |x\rangle \langle x|. 
\end{eqnarray*}
To show that this state is full rank, we observe that for an arbitrary 
vector $|\varphi \rangle = \sum_x |\varphi_x \rangle \otimes |x\rangle$, it holds that
\begin{eqnarray*}
\langle \varphi | \ee(\rho \otimes \xi)|\varphi\rangle 
= \frac{1}{\dim (\hs)} \sum_x \langle \varphi_x |\E_x|\varphi_x\rangle,  
\end{eqnarray*}
which is non-vanishing if every $\E_x$ is full-rank. By item (ii) of \lemref{lemma:third-law-channel-condition}, $\ee$ is constrained by the third law. Finally, choosing the pointer observable as $\Z_x=|x\rangle \langle x|$, then by \eq{eq:instrument-dilation} we obtain
\begin{eqnarray*}
\tr\sub{\aa}[(\onesys \otimes \Z_x) \ee(\rho\otimes \xi)] 
= \sqrt{\E_x} \rho \sqrt{\E_x} =: \ii^L_x(\rho), 
\end{eqnarray*}
and so the above measurement scheme implements   the L\"uders instrument. 
\end{proof}

We obtain the following as an immediate consequence of the above:

\begin{corollary}
Let $\E$ be a completely unsharp observable. Define $\mathscr{I}$ as the set of all $\E$-compatible instruments $\ii$ such that, for every full-rank state $\rho$  and for every outcome $x$, $\ii_x(\rho)$ is full-rank. It holds that every $\ii \in \mathscr{I}$ admits a measurement scheme $\mm = (\ha, \xi, \ee, \Z)$ that is constrained by the third law. 

\end{corollary}
\begin{proof}
The operations of every $\E$-instrument may be written as $\ii_x(\cdot ) = \Phi_x \circ \ii^L_x(\cdot)$, where 
$\ii^L$ is the L\"uders instrument  for $\E$ and $\Phi_x$ is an arbitrary channel. It is easy to show that $\ii \in \mathscr{I}$ if and only if $\Phi_x$ is constrained by the third law. That is, $\Phi_x(\omega)$ is full-rank for every full-rank $\omega$. To see this, let $\sigma = \onesys/ \dim(\hs)$ be the complete mixture, and define $\omega := \ii^L_x(\sigma)/\tr[\ii^L_x(\sigma)] = \E_x / \tr[\E_x]$ which, given complete unsharpness of $\E$, is guaranteed to be full-rank. Therefore, $\ii_x(\sigma)$ is full-rank if and only if $\Phi_x(\omega)$ is full-rank. The claim follows from item (ii) of \lemref{lemma:third-law-channel-condition}. 

Now let us define a channel $\ee_2$ by 
\begin{eqnarray*}
\ee_2(A\otimes B)=\sum_x \Phi_x(A)\otimes |x\rangle \langle x|B|x\rangle \langle x|,
\end{eqnarray*}
where $\Phi_x$ is an arbitrary channel which is constrained by the third law. In such a case, $\ee_2$ is also constrained by the third law; it is easily verified that $\ee_2(\onesys\otimes \oneapp)$ has full rank.  Therefore, a concatenated channel 
$\ee:= \ee_2 \circ \ee_1$, where $\ee_1$ is a channel defined by the Kraus operators in \eq{ChanLuders}, is  also constrained by the  the third law. But the measurement scheme $\mm = (\ha, \xi,  \ee, \Z)$, with $\xi$ the complete mixture and  $\Z_x = |x\rangle \langle x|$, implements $\ii_x(\cdot) = \Phi_x \circ \ii^L_x(\cdot)$. 

\end{proof}

\section{Fixed-point structure of measurements constrained by the third law}\label{app:fixed-point-measurement}

We define the fixed-point sets of the $\E$-channel $\ii_\xx$ and its dual $\ii_\xx^*$ as
\begin{align*}
&\ff(\ii_\xx):= \{ A \in \lo(\hs) : \ii_\xx(A) = A\},  \\
&\ff(\ii_\xx^*):= \{ A \in \lo(\hs) : \ii_\xx^*(A) = A\}.
\end{align*}
  Now let us define the channels
\begin{align*}
 \ii_{\av}(\cdot)&:=\lim_{N\to \infty} \frac{1}{N}\sum_{n=1}^{N}(\ii_\xx)^n(\cdot), \\
    \ii^*_{\av}(\cdot)&:=\lim_{N\to \infty} \frac{1}{N}\sum_{n=1}^{N}(\ii^*_\xx)^n(\cdot).
\end{align*}
$\ii^*_{\av}$ is a CP projection on $\ff(\ii_\xx^*) = \ff(\ii_\av^*)$, i.e., it holds that $\ii_\av^* = \ii_\av^* \circ \ii_\xx^* = \ii_\xx^* \circ \ii_\av^* = \ii_\av^* \circ \ii_\av^*$.   Similarly,  $\ii_\av$ is  a CP projection on $\ff(\ii_\xx) = \ff(\ii_\av)$. Now let us assume that the measurement scheme for $\ii$ is constrained by the third law. It follows from item (ii) of \lemref{lemma:measurement-third-law} that $\ff(\ii_\xx)$ contains a full-rank state $\rho_0$. This in turn  implies that $\ff(\ii_\xx^*)$ is a  von Neumann algebra, i.e., the fixed points of $\ii_\xx^*$ satisfy multiplicative closure  \cite{Bratteli1998, Arias2002}. But since $\hs$ is finite-dimensional, then $\ff(\ii_\xx^*)$ is a finite von Neumann algebra $\mathscr{A}$, which may have an Abelian non-trivial center $\mathscr{Z} := \mathscr{A} \cap \mathscr{A}'$ generated by the set of ortho-complete projections $\{P_\alpha\}$. That is, every self-adjoint $B \in \mathscr{Z}$ can be written as $B = \sum_\alpha  \lambda_\alpha P_\alpha$.  We may therefore decompose $\mathscr{A}$ into a finite direct sum $\mathscr{A} = \oplus_\alpha \mathscr{A}_\alpha$, where each $\mathscr{A}_{\alpha} = P_\alpha \mathscr{A}$ is a 
type-I factor (a finite dimensional von Neumann algebra with a trivial center) on 
$P_{\alpha}\hs = \kk_\alpha \otimes \rr_\alpha$, 
written as $\mathscr{A}_{\alpha} = \lo(\kk_\alpha)\otimes \one\sub{\rr_\alpha}$. It follows  that  we may write 
\begin{align}\label{eq:factor-fixed-point-set}
\ff(\ii_\xx) &= \bigoplus_\alpha \lo(\kk_\alpha) \otimes \omega_\alpha, \nonumber \\
\ff(\ii_\xx^*)&= \bigoplus_{\alpha} \lo(\kk_{\alpha}) \otimes \one\sub{\rr_\alpha},  
\end{align}
and 
\begin{align}\label{eq:factor-av-channel}
\ii_\av(\cdot) &= \sum_{\alpha}
\tr\sub{\rr_\alpha}[P_{\alpha} \cdot P_{\alpha}]\otimes \omega_{\alpha},  \nonumber \\
\ii^*_\av(\cdot) &= \sum_{\alpha}\Gamma_{\omega_{\alpha}}
(P_{\alpha} \cdot P_{\alpha})\otimes \one_{\mathcal{R}_{\alpha}}, 
\end{align}
where:  $\omega_\alpha$ are states on $\rr_\alpha$;   $\Gamma_{\omega_\alpha} : \lo(\kk_{\alpha}\otimes \rr_{\alpha}) \to \lo(\kk_\alpha)$ are restriction maps; and $\tr\sub{\rr_\alpha}: \lo(\kk_\alpha \otimes \rr_\alpha) \to \lo(\kk_\alpha)$ are partial traces \cite{Lindblad1999a}.

Note that the third-law constraint implies that $\omega_\alpha$ are full-rank states on $\rr_\alpha$. This can be immediately inferred by noting that, given the complete mixture $\sigma = \onesys / \dim(\hs)$, it holds that $\ii_\av(\sigma) \propto \oplus_\alpha \one\sub{\kk_\alpha} \otimes \omega_\alpha$. By property (i) of \lemref{lemma:third-law-channel-condition}, this state must be full-rank, which holds if and only if $\omega_\alpha$ are full-rank for all $\alpha$.

Finally, let us note that since $\ff(\ii_\xx^*)$ is a von Neumann algebra, then it holds that  
\begin{align}\label{eq:fixed-point-commutant-E}
\ff(\ii_\xx^*) \subseteq \E' := \{A \in \lo(\hs) : [\E, A] = \zero \},    
\end{align} 
that is, the fixed points of $\ii_\xx^*$ are contained in the commutant of $\E$ \cite{Heinosaari2010}. We now provide a useful result indicating the form that the effects of $\E$ must take in light of the fixed-point structure of the measurement channel:

\begin{lemma}\label{lemma:effect-factor-decomposition}
 Let  $\mm:=(\ha, \xi, \ee, \Z)$ be a measurement scheme for a non-trivial observable $\E$, with instrument $\ii$, acting in $\hs$. Assume that $\mm$ is constrained by the third law. Then the effects of $\E$ are of the form 
     \begin{align*}
    \E_x = \bigoplus_\alpha \one\sub{\kk_\alpha} \otimes E_{x, \alpha},
\end{align*}
where $\zero < E_{x, \alpha} < \one\sub{\rr_\alpha}$ for all $x$ and $\alpha$. 
\end{lemma}
\begin{proof}
By the channel $\Gamma_\xi^\ee$ defined in \eq{eq:gamma-u}, we may write  $\ii_x^*(\cdot) =  \Gamma_\xi^\ee (\cdot \otimes \Z_x)$, and so we may write   $\E_x = \ii_x^*(\onesys) = \Gamma_\xi^\ee (\onesys \otimes \Z_x)$ and $\ii_\xx^*(\cdot) = \Gamma_\xi^\ee (\cdot \otimes \oneapp)$.  Since  $\ff(\ii_\xx^*)$ is a von Neumann algebra, for any $A \in \ff(\ii_\xx^*)$ it holds that $A^*A, A A^* \in \ff(\ii_\xx^*)$.     By the multiplicability theorem \cite{Choi1974}, this implies that   $A \Gamma_\xi^\ee (B) = \Gamma_\xi^\ee ( (A\otimes \oneapp) B)$ and $ \Gamma_\xi^\ee (B) A = \Gamma_\xi^\ee ( B (A\otimes \oneapp))$ for all $B \in \lo(\hs \otimes \ha)$. By choosing $B = \onesys \otimes \Z_x$,  we may therefore write
\begin{align*}\label{eq:multiplicative-non-disturbance-Q}
\ii_x^*(A) =  \Gamma_\xi^\ee (A \otimes \Z_x) =  A \E_x = \E_x A
\end{align*}
for all $A \in \ff(\ii_\xx^*)$. Now assume that $A \E_x = \zero$. By the above equation this implies that $\ii_x^*(A^*A) = A^*A \E_x = \zero$. By item (iii) of  \lemref{lemma:measurement-third-law}, it follows that for any $A \in \ff(\ii_\xx^*)$, it holds that $A \E_x = \zero \iff A = \zero$.

Now note that the condition $\ff(\ii_\xx^*) \subset \E'$ implies that $\E \subset \ff(\ii_\xx^*)'$. By  \eq{eq:factor-fixed-point-set} it holds that $ \ff(\ii_\xx^*)' = \bigoplus_{\alpha} \one\sub{\kk_{\alpha}} \otimes \lo(\rr_\alpha)$. That the effects of $\E$ are decomposed as in the statement of the lemma directly follows. Now assume that $E_{x,\alpha} = \zero$ for some $\alpha$. It will hold that an operator  $A = A_\alpha \otimes \one\sub{\rr_\alpha} \in \ff(\ii_\xx^*)$ exists,   with $A_\alpha \ne \zero$, such that   $A \E_x = \zero$. But this contradicts what we showed above. Therefore, all $E_{x,\alpha}$ must be strictly positive. Finally, since $\E$ is non-trivial, then there exists at least two distinct outcomes, and so by normalisation it holds that  $E_{x,\alpha} <  \one\sub{\rr_\alpha}$.

\end{proof}

\section{Non-disturbance}\label{app:non-disturbance}

An observable $\F := \{\F_y : y \in \yy\}$ is non-disturbed by an $\E$-compatible instrument $\ii$ if  $\tr[\F_y \ii_\xx(\rho)] = \tr[\F_y \rho]$ holds for all states $\rho$ and outcomes $y$. This can equivalently be stated as  $\ii_\xx^*(\F_y) = \F_y$ for all $y$, which we denote as $\F \subset \ff(\ii_\xx^*)$. If the measurement scheme for $\ii$ is constrained by the third law, then as discussed surrounding \eq{eq:fixed-point-commutant-E} it holds that $\ff(\ii_\xx^*) \subseteq \E'$, and so a necessary condition for non-disturbance of $\F$ is for $\F$ to commute with $\E$. As we show below, however, commutation is not sufficient; properties of the measured observable impose further constraints. 

\begin{prop}\label{prop:non-disturbance}
Let $\mm:= (\ha, \xi, \ee, \Z)$ be a measurement scheme for an $\E$-compatible instrument $\ii$ acting in $\hs$. The following hold:
\begin{enumerate}[(i)]
    \item  If $\| \E_x \| =1$ for any $x$, then    there exists a projection $P \in \E'$ such that $P\notin 
     \ff(\ii_\xx^*)$ for any instrument $\ii$ that can be implemented by a scheme $\mm$ that is constrained by the third law.

    \item If $\E$ is completely unsharp, then a scheme $\mm$ that is constrained by the third law can be chosen so that $\ff(\ii_\xx^*) = \E'$.
\end{enumerate}

\end{prop}
\begin{proof}
\begin{enumerate}[(i)]
    \item 
By \lemref{lemma:effect-factor-decomposition}, we may write
 \begin{align*}                   
    \E' = \bigoplus_\alpha  \lo(\kk_{\alpha})
    \otimes E_{\alpha}',
\end{align*}
where $E_{\alpha}' := \{A \in \lo(\rr_\alpha) : [E_{x,\alpha}, A] = \zero \, \forall \, x \in \xx\}$. If $\ff(\ii_\xx^*) = \E'$, then by \eq{eq:factor-fixed-point-set} it must hold that $E_\alpha '= \co \one\sub{\rr_\alpha}$ for all $\alpha$.

Recall that for each $\alpha$, $\{E_{x,\alpha} : x \in \xx\}$ is a POVM acting in $\rr_\alpha$, where $\zero < E_{x,\alpha} <\one\sub{\rr_\alpha}$ and $\sum_x E_{x,\alpha} = \one\sub{\rr_\alpha}$. 
Assume that for some $\alpha$, there exists $x$ such that $E_{x,\alpha}$ has eigenvalue 1. Let $P$ be the projection on the eigenvalue-1 eigenspace of $E_{x,\alpha}$. Since $\E$ is non-trivial, then $P < \one\sub{\rr_\alpha}$. By normalisation, it follows that $P E_{x',\alpha} =  E_{x',\alpha} P = \delta_{x,x'}P$ for all $x'$, and so  there exists $ P \not\propto  \one\sub{\rr_\alpha} \in  E_{\alpha}'$.  Therefore, $\ff(\ii_\xx^*) = \E'$ holds only if $E_{x,\alpha}$ does not have eigenvalue 1, and so $\| \E_x\| <1$ for all $x$. 

\item By \propref{prop:luders-completely-unsharp}, a completely unsharp observable $\E$  admits a L\"uders instrument $\ii^L$, given the third law constraint. In finite dimensions, for the L\"uders instrument compatible with $\E$, it holds that $\ff({\ii_\xx^{L}}^*) = \E'$ \cite{Busch1998}.
\end{enumerate}

\end{proof}
In other words, an observable $\E$ admits an instrument $\ii$ so that $\ff(\ii_\xx^*) = \E'$, with such instrument realisable by a measurement scheme $\mm$ that is constrained by the third law,  if $\E$ is completely unsharp, and only if $\|\E_x\| < 1$ for all $x$. Note that by item (i), if $\|\E_x\|=1$ for some $x$, then there exists a POVM $\{P, P^\perp := \onesys - P\}$  that commutes with $\E$, but is disturbed by any realisable $\E$-instrument $\ii$, since $\ii_\xx^*(P) \ne P$. In particular, the above proposition implies that for any possible measurement of a norm-1 observable $\E$, such as a sharp observable, there exists some $\F\subset \E'$ that is  disturbed.

Of course, while a measurement of a norm-1 observable $\E$ is guaranteed to disturb some observable that commutes, this does not imply that there are no non-disturbed observables. Below, we provide necessary conditions on  $\E$ so that its measurement allows for a non-trivial class of non-disturbed observables. 

\begin{prop}\label{prop:rank-deg}
Let $\mm:= (\ha, \xi, \ee, \Z)$ be a measurement scheme for an $\E$-compatible instrument $\ii$ acting in $\hs$, and assume that $\mm$ is constrained by the third law. The following hold:
\begin{enumerate} [(i)]
    \item If $\E$ is a small-rank observable, then  $\ff(\ii_\xx^*) = \co \onesys$.
    
    \item If $\E$ is a non-degenerate observable, then $\ff(\ii_\xx^*)$ is Abelian.  
\end{enumerate}
\end{prop}
\begin{proof}

\begin{enumerate}[(i)]
    \item Recall from \defref{defn:small-rank} that a small-rank observable has at least one effect that is rank-1.     By \lemref{lemma:effect-factor-decomposition}, the rank of every effect of $\E$ is bounded as $\rank{\E_x} \geqslant \sum_\alpha \dim(\kk_\alpha)$.  Therefore, if any effect of $\E$ is rank-1, then it must hold that the number of indices $\alpha$ is 1, and that  $\dim(\kk_\alpha) = 1$, so that by \eq{eq:factor-fixed-point-set} we have $\ff(\ii_\xx^*) = \co \onesys$.
    
    \item Recall from \defref{defn:non-degenerate} that $\E$ is non-degenerate if one of its effects has no multiplicities in its strictly positive eigenvalues. By \lemref{lemma:effect-factor-decomposition},  non-degeneracy of such an effect implies that $\dim (\kk_\alpha)=1$ for every $\alpha$. It follows from \eq{eq:factor-fixed-point-set} that  $\ff(\ii_\xx^*) = \oplus_{\alpha} \co \one\sub{\rr_{\alpha}}$, i.e., for any $A, B \in \ff(\ii_\xx^*)$, it holds that $[A, B] = \zero$. \end{enumerate}

\end{proof} 

In other words, for the class of non-disturbed observables to be non-trivial, then the measured observable must be large-rank. Additionally, for the non-disturbed observables to be non-commutative, then the measured observable must be degenerate. We now provide a concrete example for a measurement scheme constrained by the third law, which measures a sharp observable that is large-rank and hence degenerate,  that does not disturb a non-trivial class of possibly non-commutative  observables. 
\begin{example}
Consider $\hs :=\h_1\otimes \h_2$ and $\ha := \h_3$,  with $\h_i = \co^2$. Let $\ee$ be a unitary channel which acts trivially in $\h_1$ and implements a  swap in $\h_2 \otimes \h_3$, i.e., $\ee^*(A
\otimes B\otimes C) = A\otimes C\otimes B$. As shown in \corref{corollary:third-law-channel-examples} this channel is constrained by the third law. Now define a pointer observable $\Z_x :=|x\rangle \langle x|$ acting in $\ha$. A measurement scheme $\mm:= (\ha, \xi, \ee, \Z)$, with $\xi$ a full-rank state on $\ha$, is constrained by the third law, and measures the sharp, large-rank and degenerate observable with effects   
$\E_x = \one_1 \otimes |x\rangle \langle x|$ in $\hs$. The fixed-point set of the instrument $\ii$ implemented by $\mm$ is easily verified to be $\ff(\ii_\xx^*)= \lo(\h_1)\otimes 
\one_2 \subset \E'$, which is a non-trivial and non-commutative proper subset of $\E'$. That is, any possibly non-commutative observable $\F$ with effects $\F_y = F_y \otimes \one_2$ will be non-disturbed. 
\end{example}

\section{First-kindness and repeatability}\label{app:first-kindness}

An $\E$-compatible instrument $\ii$ is a measurement of the first kind if $\E \subset \ff(\ii_\xx^*)$. A subclass of first-kind measurements are repeatable, satisfying the additional condition  $\ii_y^*(\E_x) = \delta_{x,y} \E_x$. Only norm-1 observables admit a repeatable instrument.  Repeatability implies first-kindness, since $\ii_\xx^*(\E_x) = \sum_{y}\ii_y^*(\E_x) = \E_x$. We now show that the third law only permits first-kindness for completely unsharp observables, and so categorically  prohibits repeatability.

\begin{prop}\label{prop:first-kindness}
A non-trivial observable $\E$ admits a measurement of the first kind, given a measurement scheme $\mm:=(\ha, \xi, \ee, \Z)$ that is constrained by the third law, if and only if $\E$ is commutative and completely unsharp. 
\end{prop}
\begin{proof}
Let us first show the only if statement. An $\E$-compatible instrument $\ii$ is a measurement of the first kind if $\E \subset \ff(\ii_\xx^*)$.  Now, recall that $\ff(\ii_\xx^*) = \ff(\ii_\av^*)$, so  that the first-kind condition also reads as $\ii_\av^*(\E_x) = \E_x$ for all outcomes $x$.  By \lemref{lemma:effect-factor-decomposition} and \eq{eq:factor-av-channel},  it follows that $\E_x = \oplus_\alpha \lambda_\alpha(x)  \one\sub{\kk_\alpha} \otimes \one\sub{\rr_\alpha}$ with $\lambda_{\alpha}(x):=  \tr[E_{x,\alpha} \omega_\alpha]$. Since $\zero < E_{x,\alpha} < \one\sub{\rr_\alpha}$ and $\omega_\alpha$ is full-rank  for all $\alpha$ and $x$, then $\lambda_{\alpha}(x) \in (0,1)$. The claim immediately follows. 

To show the  if statement, recall from \propref{prop:luders-completely-unsharp} that a completely unsharp observable admits a L\"uders instrument under the third law constraint, and that if such an observable is also commutative, then the L\"uders instrument is a first-kind measurement. 

\end{proof}
\begin{corollary}\label{cor:repeatability}
Let $\mm:=(\ha, \xi, \ee, \Z)$ be a measurement scheme for an $\E$-compatible instrument $\ii$. If $\mm$ is constrained by the third law, then $\ii$ cannot be repeatable. Moreover, for every pair of outcomes $x,y$, there exists a state $\rho$ such that $\tr[\E_y \ii_x(\rho)] >0$.
\end{corollary}
\begin{proof}
If $\ii$ is repeatable, then it is also first-kind. By \propref{prop:first-kindness}, given the third law constraint only a completely unsharp observable admits a first-kind measurement. Since repeatability is only admitted for norm-1 observables, and completely unsharp observables lack the norm-1 property, then $\ii$ cannot be repeatable.  Now note that we may write $\ii_x^*(\E_y) = \Gamma_\xi^\ee(\E_y \otimes \Z_x)$. Since $\E_y \otimes \Z_x > \zero$ for all $x,y$, then by item (i) of \lemref{lemma:measurement-third-law} it holds that $\ii_x^*(\E_y) > \zero$, and so   $\tr[\E_y \ii_x(\rho)] = \tr[\ii_x^*(\E_y) \rho] >0$ for some $\rho$.
\end{proof}

Below we provide a model for a measurement scheme that is constrained by the third law, and which implements a first-kind measurement of a completely unsharp observable. Note that the model does not implement a L\"uders instrument. 
\begin{example}
Consider $\hs=\ha= \co^N$ with  an orthonormal basis $\{|n\> : n=0,\dots, N-1 \}$ for each system. Consider the unitary channel $\ee(\cdot) = U \cdot U^*$ acting in $\hs\otimes \ha$, with the unitary operator $U$ defined as 
\begin{align*}
U = \sum_{m,n} |n\>\<n| \otimes  |m \oplus n\>\<m|,  
\end{align*}
where $\oplus$ denotes addition modulo $N$. As shown in \corref{corollary:third-law-channel-examples} this channel is constrained by the third law. Consider a full-rank state on $\ha$ given as $\xi = \sum_n q(n) |n\rangle \langle n |$
 with $q(n) >0$ for all $n$.   Let $\Z_n = |n\rangle \langle n|$ be a pointer observable acting in $\ha$. The measurement scheme $\mm:= (\ha, \xi, \ee, \Z)$ is therefore  constrained by the third law. Moreover,  the operations of the instrument $\ii$ implemented by $\mm$  satisfy  
\begin{align*}
\ii_x^*(A)
=\Gamma^{\mathcal{E}}_{\xi}
(A\otimes \Z_x)
= \sum_n q(x\ominus n) \< n | A | n\> |n\>\< n|,    
\end{align*} 
where $\ominus$ denotes subtraction modulo $N$. The measured observable is therefore commutative and completely unsharp, with effects $\E_x= \sum_n q(x\ominus n) |n\rangle \langle n |$, 
whose eigenvalues are in $(0,1)$.  Additionally, the fixed-point set of the $\E$-channel is 
$\ff(\ii_\xx^*) = \oplus_n \co |n\>\<n| \subset \E'$,  which is nontrivial. In particular, we have $\E \subset \ff(\ii_\xx^*)$, and so  this model describes a first-kind measurement for $\E$. 
\end{example}

\section{Ideality}\label{app:ideality}

\begin{prop}\label{prop:appendix-ideality}
Let  $\mm:=(\ha, \xi, \ee, \Z)$ be  a measurement scheme for a non-trivial norm-1 observable   $\E$, with the instrument $\ii$,  acting in $\hs$. If $\mm$ is constrained by the third law, then $\ii$ cannot be ideal.   
\end{prop}
\begin{proof}
Assume that $\ii$ is ideal, so that for every state $\rho$ and for  every outcome $x$ such that $\tr[\E_x \rho]=1$, it holds that $\ii_x(\rho) = \rho$. Since $\tr[\E_x \rho] =1$ is equivalent to  $\tr[\E_y \rho] = \delta_{x,y}$, and hence $\ii_y(\rho) = \zero$ for all $y\ne x$,  this implies that $\ii_\xx(\rho) = \sum_y \ii_y(\rho) = \ii_x(\rho) = \rho$, and hence  $\rho \in \ff(\ii_\xx) = \ff(\ii_\av)$. But by \eq{eq:factor-av-channel},  we have
\begin{equation}\label{eq:ideal-proof-1}
 \rho =    \ii_\av(\rho) = \bigoplus_\alpha \sigma_\alpha \otimes \omega_\alpha, 
\end{equation}
where $\omega_\alpha$ are fixed full-rank states on $\rr_\alpha$ and $\sigma_\alpha := \tr\sub{\rr_\alpha}[P_{\alpha} \rho P_{\alpha}]$ are sub-unit-trace positive operators on $\kk_\alpha$.

Since $\E$ is norm-1,  then by \lemref{lemma:effect-factor-decomposition} it follows that for each outcome $x$,  there exists at least one $\alpha$ such that $\| E_{x,\alpha} \| =1$. Any state on $\kk_\alpha \otimes \rr_\alpha$ written as $\rho = \sigma_\alpha \otimes \mu_\alpha$, such that $\tr[E_{x,\alpha} \mu_\alpha] = 1$, will give  $\tr[\E_x \rho] = \tr[ (\one\sub{\kk_\alpha} \otimes E_{x,\alpha}) (\sigma_\alpha \otimes \mu_\alpha)] = \tr[E_{x,\alpha} \mu_\alpha] = 1$.  By \eq{eq:ideal-proof-1}, if $\ii$ is ideal we must have
\begin{align*}
    \sigma_\alpha \otimes \mu_\alpha = \sigma_\alpha \otimes \omega_\alpha,
\end{align*}
that is,    $\tr[\E_{x,\alpha} \mu_\alpha] =1 \iff \mu_\alpha = \omega_\alpha$. But given that $\omega_\alpha$ are full-rank states on $\rr_\alpha$, $\tr[\E_{x,\alpha} \mu_\alpha] =1$ if and only if   $\dim(\rr_\alpha) = 1$, so that $ E_{x,\alpha} = \one\sub{\rr_\alpha}$. But by \lemref{lemma:effect-factor-decomposition}, if $\E$ is non-trivial, $E_{x,\alpha} < \one\sub{\rr_\alpha}$ must hold.  We therefore have a contradiction, and so $\ii$ cannot be ideal. 
\end{proof} 

In particular, let us highlight the fact that if $\E$ is a non-trivial norm-1 observable, then for every outcome $x$,  and for every state $\rho$ such that $\tr[\E_x\rho] = 1$, then  a third law constrained measurement will give $\ii_x(\rho) \ne \rho$.

\section{Extremality}\label{app:extremality}

Let $\{K_i^{(x)} : i=1,\dots, M_x\}$ be a minimal Kraus representation for the operation $\ii_x$ of an $\E$-compatible instrument $\ii$, with $M_x$ the Kraus rank of $\ii_x$.  The following is a series of necessary conditions for extremality of such an instrument:

\begin{lemma}\label{lemma:extremality-conditions-general}
The instrument $\ii$ is extremal only if the following conditions are met:
\begin{enumerate}[(i)]
    \item The operations $\ii_x$ are all  extremal.
    \item $M_x \leqslant  \rank{\E_x}$ for all $x$.
     \item $\{K_i^{(x)*} K_j^{(x)} : x \in \xx; i, j=1,\dots, M_x\}$ are linearly independent (this condition is necessary and sufficient).
\end{enumerate}
\end{lemma}
\begin{proof}
\begin{enumerate}[(i)]
\item This trivially follows from the definition of extremal instruments. 

\item By  Remark 6 of Ref. \cite{Choi1975}, a sub-unital CP map $\Phi^* : \lo(\kk) \to \lo(\h)$, with minimal Kraus representation $\{K_i : i=1,\dots, M\}$, is extremal only if $\{K_i^*K_j \in \lo(\h)\}$ are linearly independent. Since the cardinality of this set is $M^2$, while $\dim(\lo(\h)) = \dim(\h)^2$, it follows that  $M \leqslant \dim(\h)$ must hold. Now, since we may always write $\ii_x^*(\cdot) = \sqrt{\E_x} \Phi^*(\cdot) \sqrt{\E_x}$,  it holds that $\ii_x^* : \lo(\hs) \to \lo(P_x \hs)$, with $P_x$ the minimal  projection on the support of $\E_x$ so that $\dim( P_x \hs) = \rank{\E_x}$. By above, extremality of  $\ii_x$ implies that  $M_x \leqslant  \rank{\E_x}$. 

 \item Theorem 5 of Ref. \cite{DAriano2011}.

\end{enumerate}

\end{proof}

Now we shall provide some useful results that will allow us to investigate how the third law constrains extremality in the sequel:

\begin{lemma}\label{lemma:Gamma-independent-factorisation}
Let $\Gamma_\xi : \lo(\hs \otimes \ha) \to \lo(\hs)$ be a restriction map. Assume that  $\Gamma_\xi(A)$ is independent of the state $\xi$. It holds that $A = B \otimes \oneapp$.
\end{lemma}
\begin{proof}
Let $\{\ket{e_n}\}$ and $\{\ket{b_j}\}$ be orthonormal bases that span $\hs$ and $\ha$, respectively.  If $\Gamma_\xi(A)$ is independent of $\xi$, it follows that for any unit vector $|\psi\> \in \ha$,  
\begin{align*}
    \<e_m| \Gamma_{|\psi\>\<\psi|}(A)|e_n\> = \langle e_m \otimes \psi | A |e_n \otimes \psi\rangle 
= C_{m,n}. 
\end{align*}
Now, for  any choice of $j$ and $k$,  define  $|\phi_\pm\rangle := 
\frac{1}{\sqrt{2}}(|b_j \rangle \pm |b_k\rangle)$ and $|\varphi_\pm\rangle := 
\frac{1}{\sqrt{2}}(|b_j \rangle \pm \imag |b_k\rangle)$. We obtain 
\begin{align*}
&\<e_m| \Gamma_{|\phi_\pm\>\<\phi_\pm|}(A)|e_n\> = C_{m,n} \nonumber \\
 & \,  \pm \frac{1 - \delta_{j,k}}{2}\bigg( \langle e_m \otimes b_j |A|e_n \otimes b_k\rangle 
+ \langle e_m \otimes b_k |A|e_n \otimes b_j\rangle\bigg),    
\end{align*}
and 
\begin{align*}
&\<e_m| \Gamma_{|\varphi_\pm\>\<\varphi_\pm|}(A)|e_n\> = C_{m,n} \nonumber \\
 & \,  \pm \imag \frac{1 - \delta_{j,k}}{2}\bigg( \langle e_m \otimes b_j |A|e_n \otimes b_k\rangle 
- \langle e_m \otimes b_k |A|e_n \otimes b_j\rangle\bigg).    
\end{align*}
Thus we conclude that $\langle e_m \otimes b_j |A|e_n 
\otimes b_k \rangle = \delta_{j,k}C_{m,n}$. It follows that 
\begin{align*}
    A &= \sum_{m,n}\sum_{j,k} \langle e_m \otimes b_j |A|e_n 
\otimes b_k \rangle |e_m\>\<e_n| \otimes |b_j\>\<b_k|, \nonumber \\
& = \sum_{m,n}\sum_{j} C_{m,n} |e_m\>\<e_n| \otimes |b_j\>\<b_j| = B \otimes \oneapp.
\end{align*}.  
\end{proof}

\begin{lemma}\label{lemma:extremal-interaction-channel-identity}
Let  $\mm:=(\ha, \xi, \ee, \Z)$ be  a measurement scheme for an instrument $\ii$  acting in $\hs$. Assume that $\mm$ is constrained by the third law. If $\ii$ is extremal, then it holds that 
\begin{align*}
    \ee^*(\cdot \otimes \Z_x) = \ii_x^*(\cdot) \otimes \oneapp
\end{align*}
for all $x$. 
\end{lemma}
\begin{proof}
If $\mm$ is constrained by the third law, then $\xi$ is full-rank. By \lemref{lemma:full-rank-state-inequality},  for an arbitrary unit vector $\ket{\phi} \in \ha$, there exists a $0< \lambda < 1$ such that $\xi \geqslant \lambda |\phi\>\<\phi|$. Defining the state $\sigma:= (\xi - \lambda |\phi\>\<\phi|)/ (1 - \lambda)$, we may thus decompose $\xi$ as $\xi = \lambda |\phi\>\<\phi| + (1-\lambda) \sigma$. Using the map $\Gamma^\ee_\xi$ defined in \eq{eq:gamma-u}, we may thus write 
\begin{align*}
    \ii_x^*(\cdot) &= \Gamma_\xi^\ee(\cdot \otimes \Z_x), \nonumber \\
    & = \lambda \Gamma_{|\phi\>\<\phi|}^\ee(\cdot \otimes \Z_x) + (1-\lambda) \Gamma_\sigma^\ee(\cdot \otimes \Z_x).
\end{align*}
Since $\ii$ is assumed to be extremal, it follows that 
\begin{align*}
    \ii_x^*(\cdot) &= \Gamma_{|\phi\>\<\phi|}^\ee(\cdot \otimes \Z_x) = \Gamma_{|\phi\>\<\phi|} \circ \ee^*(\cdot \otimes \Z_x)
\end{align*}
must hold for arbitrary unit vectors $\ket{\phi} \in \ha$. The claim follows from  \lemref{lemma:Gamma-independent-factorisation}.
\end{proof}

We are finally ready to provide necessary conditions for extremality of an instrument constrained by the third law:

\begin{prop}\label{prop:extremality-third-law}
Let  $\mm:=(\ha, \xi, \ee, \Z)$ be  a measurement scheme for a non-trivial observable $\E$, with  the instrument $\ii$,  acting in $\hs$. Assume that $\mm$ is constrained by the third law, and that $\ii$ is extremal. The following hold:

\begin{enumerate}[(i)]
    \item $\rank{\E_x} \geqslant \sqrt{\dim(\hs)}$ for all $x$.

    \item $\ee$ cannot be a unitary channel.
    
\end{enumerate}
\end{prop}
\begin{proof}
Let us first prove (i). Let $P_x$ be the minimal projection on the support of $\E_x$, where we note that $\dim(P_x \hs) = \rank{\E_x}$. Consider a  full-rank state  $\rho = \sum_{n=1}^{\rank{\E_x}} p_n |\psi_n\>\<\psi_n|$ on $P_x\hs$, where $p_n>0$ and $\ket{\psi_n} \in P_x \hs$ for all $n$. By item (iv) of \lemref{lemma:measurement-third-law}, it holds that $\ii_x(\rho)$ has full-rank in $\hs$, i.e., $\rank{\ii_x(\rho)} = \dim(\hs)$. Now, for any unit vector $\ket{\psi} \in P_x\hs$ it holds that
\begin{align*}
    \ii_x(|\psi\>\<\psi|) = \sum_{i=1}^{M_x} |K_i^{(x)} \psi\>\<K_i^{(x)} \psi|.
\end{align*} 
We thus have $\dim(\hs) = \rank{\ii_x(\rho)} \leqslant M_x \rank{\E_x}$. By item (ii) of \lemref{lemma:extremality-conditions-general}, extremality implies that $M_x \leqslant \rank{\E_x}$. It follows that $\dim(\hs) \leqslant \rank{\E_x}^2$.

Now we shall prove (ii). If $\ii$ is extremal, then by \lemref{lemma:extremal-interaction-channel-identity} it holds that  $\ee^*(\cdot \otimes \oneapp)
= \ii_{\xx}^*(\cdot)\otimes \oneapp$. If $\ee^*(\cdot) = U^* \cdot U$ is a unitary channel, then $\ii_{\xx}^*$ is 
also a unitary channel, and so there exists a unitary operator $V \in \lo(\hs)$ such that $\ii_\xx^*(\cdot) = V^* \cdot V$, and we may thus write 
\begin{align*}
U^*(A\otimes \one)U=  V^* \otimes \oneapp (A \otimes \oneapp) V \otimes \oneapp
\end{align*}
for all $A \in \lo(\hs)$. This implies that  $L^*(A\otimes \oneapp)L = A \otimes \oneapp$, with $L = U(V^*\otimes \oneapp)$, for all $A$. As such, for all $A$ the following commutation relation must hold: 
\begin{align*}
[A\otimes \oneapp, U(V^*\otimes \oneapp)]=\zero.    
\end{align*} 
 That is, there exists a unitary operator $W$ such that $U=V\otimes W$. Therefore, the unitary channel $\ee$ is a product of unitary channels acting separately in $\hs$ and $\ha$, i.e., $\ee = \ee_1\otimes \ee_2$, where $\ee_1(\cdot) = V \cdot V^*$ acts in $\hs$ and $\ee_2(\cdot) = W \cdot W^*$ acts in $\ha$. It follows that 
\begin{align*}
\E_x &= \Gamma_\xi\circ \ee^*(\onesys \otimes \Z_x) \nonumber \\
&= \Gamma_\xi(\ee_1^*(\onesys) \otimes \ee_2^*(\Z_x)) \nonumber \\
& = \Gamma_\xi(\onesys \otimes W^* \Z_x W) \nonumber \\
&= \tr[ \Z_x W \xi W^*] \onesys,
\end{align*}
which is a trivial observable.

\end{proof}

To show that extremality is indeed possible, and that the bound $\rank{\E_x} \geqslant \sqrt{\dim(\hs)}$ is tight, let us consider the following model. 

Consider $\hs :=\h_1\otimes \h_2$ and $\ha :=\h_3$, where $\h_i = \co^2$. Choose the  apparatus state preparation as an arbitrary full-rank state $\xi$,   and choose a pointer observable  $\Z_x=|x\rangle \langle x|$. Let the measurement interaction be the channel $\ee := \ee_2 \circ \ee_1$, where 
\begin{align*}
\ee_1(A\otimes B \otimes C) &= A \otimes C \otimes B, \\
\ee_2(A\otimes B \otimes C) &= \Phi(A \otimes B) \otimes C.
\end{align*} 
That is, at first a unitary swap channel is applied in $\h_2\otimes \h_3$, and subsequently a channel $\Phi$ is applied in $\h_1 \otimes \h_2$. The operations of the instrument implemented by the measurement scheme $\mm:= (\ha, \xi, \ee, \Z)$ are thus 
\begin{equation}\label{eq:extremal-example-1}
\ii_x^*(A \otimes B) = \Gamma_\xi \circ \ee_1^* (\Phi^*(A \otimes B) \otimes |x\>\<x|).    
\end{equation}
 
 $\Phi$ is a channel defined by 
\begin{eqnarray*}
\Phi(A\otimes B)= \sum_{x=0}^1\sum_{f=0}^1 
K_{x,f} (A\otimes B) K_{x,f}^*,
\end{eqnarray*}
with the Kraus operators defined below: 
\begin{eqnarray*}
K_{x,f}:= V_f \otimes |\varphi_f\rangle \langle x|, 
\end{eqnarray*}
where $|\varphi_0\rangle = |0\rangle $, $|\varphi_1\rangle = |+\rangle
=\frac{1}{\sqrt{2}}(|0\rangle + |1\rangle)$,  and $\{V_f\}$ are Kraus operators for some channel acting in $\h_1$, and are
\begin{eqnarray*}
V_0&=& \left[ \begin{array}{cc}
\sqrt{\frac{1}{4}}&0\\
0& \sqrt{\frac{3}{4}}
\end{array}
\right], 
\\ 
V_1 &=&
\sigma_x \left[ 
\begin{array}{cc}
\sqrt{\frac{3}{4}}&0\\
0& \sqrt{\frac{1}{4}}
\end{array}
\right]
= \left[
\begin{array}{cc}
0& \sqrt{\frac{1}{4}}\\
\sqrt{\frac{3}{4}}& 0
\end{array}
\right].
\end{eqnarray*}

As $\mathcal{E}_1$ is a unitary  channel, it is constrained by the third law. To confirm that $\ee_2$ is constrained by the third law, it is enough to check if $\Phi$ is so constrained. We have 
\begin{eqnarray*}
\Phi\left(\frac{\one}{2}\otimes \frac{\one}{2}\right)
&&=\frac{1}{4}
\sum_x \sum_f 
V_f V_f^* 
\otimes |\varphi_f\rangle 
\langle \varphi_f|
\\
&&=
\frac{1}{2}
\sum_f 
V_f V_f^* 
\otimes |\varphi_f \rangle 
\langle \varphi_f|
\\
&&=\frac{1}{2} \left[\begin{array}{cc}
\frac{1}{4}&0\\
0&\frac{3}{4}
\end{array}\right]
\otimes 
(|0\rangle \langle 0 | 
+ |+ \rangle \langle +|). 
\end{eqnarray*}
As the complete mixture is mapped to a full-rank state, then  $\Phi$, and hence also $\ee_2$, satisfies property (ii) of \lemref{lemma:third-law-channel-condition} and is thus constrained by the third law. Finally, since third-law constrained channels are closed under composition, and  given that the apparatus preparation is full-rank, then $\mm$ is constrained by the third law. Moreover, note that $\ee$ is evidently not a unitary channel. 

Now note that  $\Phi^*$ may also be written as
\begin{align}\label{eq:extremal-example-2}
\Phi^*(A \otimes B ) &=  \sum_{x,f} K_{x,f}^*A \otimes B K_{x,f} \nonumber \\
&=\sum_{x,f}\<\varphi_f| B |\varphi_f\>
V_f^*A V_f \otimes |x\rangle \langle x| \nonumber
\\
&=\sum_f \<\varphi_f| B |\varphi_f\> V_f^* A  V_f \otimes \one. 
\end{align}

By \eq{eq:extremal-example-1} and \eq{eq:extremal-example-2}, we may thus write 

\begin{align*}
\ii_x^*(A \otimes B) & = \Gamma_{\xi}\circ \ee_1^*(\Phi^*(A\otimes B) \otimes 
|x\rangle \langle x|) \\
& = \sum_f \<\varphi_f| B |\varphi_f\> \Gamma_{\xi}\circ \ee_1^*(V_f^* A  V_f \otimes \one \otimes |x\rangle \langle x|) \\
& = \sum_f \<\varphi_f| B |\varphi_f\> \Gamma_{\xi}(V_f^* A  V_f \otimes 
|x\rangle \langle x|\otimes \one) \\
& =  \sum_f \<\varphi_f| B |\varphi_f\> V_f^* A  V_f \otimes 
|x\rangle \langle x| \\
& = \sum_f  K_{x,f}^*(A\otimes B)
K_{x,f}.
\end{align*}
In particular, the measured observable is 
\begin{align*}
\E_x &= \ii_x^*(\one \otimes 
\one) = \sum_f V_f^*V_f \otimes |x\>\<x|  =  \one \otimes |x\rangle 
\langle x|.
\end{align*}
We see that $\E$ is sharp, and that  $\rank{\E_x} = 2$. But since $\hs = \co^2 \otimes \co^2$, and so $\dim(\hs) = 4$, it holds that $\rank{ \E_x}= \sqrt{\dim{ \hs}}$. As such, the necessary condition for extremality is satisfied, with the model saturating the bound $\rank{ \E_x} \geqslant  \sqrt{\dim{ \hs}}$. But by item (iii) of \lemref{lemma:extremality-conditions-general}, $\ii$ is extremal if and only if 
$\{K_{x,f}^*K_{x,g}\}_{x,f,g}$ is a linearly independent set. These operators are written as 
\begin{align*}
K_{x,f}^*K_{x,g}= 
 \langle \varphi_f|\varphi_g\rangle V_f^*V_g
\otimes |x\rangle \langle x|,    
\end{align*} 
where 
\begin{eqnarray*}
V_0^*V_0=
\left[
\begin{array}{cc}
1/4&0\\
0&3/4
\end{array}
\right], 
\quad 
V_1^*V_1=
\left[
\begin{array}{cc}
3/4&0\\
0&1/4
\end{array}
\right], \\
V_0^*V_1=
\left[
\begin{array}{cc}
0&1/4\\
3/4&0
\end{array}
\right], 
\quad 
V_1^*V_0=
\left[
\begin{array}{cc}
0&3/4\\
1/4&0
\end{array}
\right]. 
\end{eqnarray*}
The above operators are linearly independent, and so  $\{K_{x,f}^*K_{x,g}\}_{x,f,g}$ is a linearly independent set. Therefore our model, which is constrained by the third law, implements an extremal instrument, with the rank of the measured observable's effects taking their smallest possible values.


\bibliography{Projects-Thermal-Measurement.bib}

\begin{thebibliography}{74}%
\makeatletter
\providecommand \@ifxundefined [1]{%
 \@ifx{#1\undefined}
}%
\providecommand \@ifnum [1]{%
 \ifnum #1\expandafter \@firstoftwo
 \else \expandafter \@secondoftwo
 \fi
}%
\providecommand \@ifx [1]{%
 \ifx #1\expandafter \@firstoftwo
 \else \expandafter \@secondoftwo
 \fi
}%
\providecommand \natexlab [1]{#1}%
\providecommand \enquote  [1]{``#1''}%
\providecommand \bibnamefont  [1]{#1}%
\providecommand \bibfnamefont [1]{#1}%
\providecommand \citenamefont [1]{#1}%
\providecommand \href@noop [0]{\@secondoftwo}%
\providecommand \href [0]{\begingroup \@sanitize@url \@href}%
\providecommand \@href[1]{\@@startlink{#1}\@@href}%
\providecommand \@@href[1]{\endgroup#1\@@endlink}%
\providecommand \@sanitize@url [0]{\catcode `\\12\catcode `\$12\catcode
  `\&12\catcode `\#12\catcode `\^12\catcode `\_12\catcode `\%12\relax}%
\providecommand \@@startlink[1]{}%
\providecommand \@@endlink[0]{}%
\providecommand \url  [0]{\begingroup\@sanitize@url \@url }%
\providecommand \@url [1]{\endgroup\@href {#1}{\urlprefix }}%
\providecommand \urlprefix  [0]{URL }%
\providecommand \Eprint [0]{\href }%
\providecommand \doibase [0]{https://doi.org/}%
\providecommand \selectlanguage [0]{\@gobble}%
\providecommand \bibinfo  [0]{\@secondoftwo}%
\providecommand \bibfield  [0]{\@secondoftwo}%
\providecommand \translation [1]{[#1]}%
\providecommand \BibitemOpen [0]{}%
\providecommand \bibitemStop [0]{}%
\providecommand \bibitemNoStop [0]{.\EOS\space}%
\providecommand \EOS [0]{\spacefactor3000\relax}%
\providecommand \BibitemShut  [1]{\csname bibitem#1\endcsname}%
\let\auto@bib@innerbib\@empty
\bibitem [{\citenamefont {L{\"{u}}ders}(2006)}]{Luders2006}%
  \BibitemOpen
  \bibfield  {author} {\bibinfo {author} {\bibfnamefont {G.}~\bibnamefont
  {L{\"{u}}ders}},\ }\bibfield  {title} {\bibinfo {title} {{Concerning the
  state‐change due to the measurement process}},\ }\href
  {https://doi.org/10.1002/andp.20065180904} {\bibfield  {journal} {\bibinfo
  {journal} {Ann. Phys.}\ }\textbf {\bibinfo {volume} {518}},\ \bibinfo {pages}
  {663} (\bibinfo {year} {2006})}\BibitemShut {NoStop}%
\bibitem [{\citenamefont {Busch}\ and\ \citenamefont
  {Lahti}(2009)}]{Busch2009a}%
  \BibitemOpen
  \bibfield  {author} {\bibinfo {author} {\bibfnamefont {P.}~\bibnamefont
  {Busch}}\ and\ \bibinfo {author} {\bibfnamefont {P.}~\bibnamefont {Lahti}},\
  }\bibfield  {title} {\bibinfo {title} {{L{\"{u}}ders Rule}},\ }in\ \href
  {https://doi.org/10.1007/978-3-540-70626-7_110} {\emph {\bibinfo {booktitle}
  {Compend. Quantum Phys.}}}\ (\bibinfo  {publisher} {Springer Berlin
  Heidelberg},\ \bibinfo {address} {Berlin, Heidelberg},\ \bibinfo {year}
  {2009})\ pp.\ \bibinfo {pages} {356--358}\BibitemShut {NoStop}%
\bibitem [{\citenamefont {Hegerfeldt}\ and\ \citenamefont {{Sala
  Mayato}}(2012)}]{Hegerfeldt2012}%
  \BibitemOpen
  \bibfield  {author} {\bibinfo {author} {\bibfnamefont {G.~C.}\ \bibnamefont
  {Hegerfeldt}}\ and\ \bibinfo {author} {\bibfnamefont {R.}~\bibnamefont {{Sala
  Mayato}}},\ }\bibfield  {title} {\bibinfo {title} {{Discriminating between
  the von Neumann and L{\"{u}}ders reduction rule}},\ }\href
  {https://doi.org/10.1103/PhysRevA.85.032116} {\bibfield  {journal} {\bibinfo
  {journal} {Phys. Rev. A}\ }\textbf {\bibinfo {volume} {85}},\ \bibinfo
  {pages} {032116} (\bibinfo {year} {2012})}\BibitemShut {NoStop}%
\bibitem [{\citenamefont {von Neumann}(2018)}]{Von-Neumann-Foundations}%
  \BibitemOpen
  \bibfield  {author} {\bibinfo {author} {\bibfnamefont {J.}~\bibnamefont {von
  Neumann}},\ }\href@noop {} {\emph {\bibinfo {title} {{Mathematical
  Foundations of Quantum Mechanics: New Edition}}}}\ (\bibinfo  {publisher}
  {Princeton University Press},\ \bibinfo {year} {2018})\BibitemShut {NoStop}%
\bibitem [{\citenamefont {Wigner}(1952)}]{E.Wigner1952}%
  \BibitemOpen
  \bibfield  {author} {\bibinfo {author} {\bibfnamefont {E.~P.}\ \bibnamefont
  {Wigner}},\ }\bibfield  {title} {\bibinfo {title} {{Die Messung
  quantenmechanischer Operatoren}},\ }\href
  {https://doi.org/10.1007/BF01948686} {\bibfield  {journal} {\bibinfo
  {journal} {Zeitschrift f{\"{u}}r Phys. A Hadron. Nucl.}\ }\textbf {\bibinfo
  {volume} {133}},\ \bibinfo {pages} {101} (\bibinfo {year}
  {1952})}\BibitemShut {NoStop}%
\bibitem [{\citenamefont {Busch}(2010)}]{Busch2010}%
  \BibitemOpen
  \bibfield  {author} {\bibinfo {author} {\bibfnamefont {P.}~\bibnamefont
  {Busch}},\ }\bibfield  {title} {\bibinfo {title} {{Translation of "Die
  Messung quantenmechanischer Operatoren" by E.P. Wigner}},\ }\href
  {http://arxiv.org/abs/1012.4372} {\  (\bibinfo {year} {2010})},\ \Eprint
  {https://arxiv.org/abs/1012.4372} {arXiv:1012.4372} \BibitemShut {NoStop}%
\bibitem [{\citenamefont {Araki}\ and\ \citenamefont
  {Yanase}(1960)}]{Araki1960}%
  \BibitemOpen
  \bibfield  {author} {\bibinfo {author} {\bibfnamefont {H.}~\bibnamefont
  {Araki}}\ and\ \bibinfo {author} {\bibfnamefont {M.~M.}\ \bibnamefont
  {Yanase}},\ }\bibfield  {title} {\bibinfo {title} {{Measurement of Quantum
  Mechanical Operators}},\ }\href {https://doi.org/10.1103/PhysRev.120.622}
  {\bibfield  {journal} {\bibinfo  {journal} {Phys. Rev.}\ }\textbf {\bibinfo
  {volume} {120}},\ \bibinfo {pages} {622} (\bibinfo {year}
  {1960})}\BibitemShut {NoStop}%
\bibitem [{\citenamefont {Ozawa}(2002)}]{Ozawa2002}%
  \BibitemOpen
  \bibfield  {author} {\bibinfo {author} {\bibfnamefont {M.}~\bibnamefont
  {Ozawa}},\ }\bibfield  {title} {\bibinfo {title} {{Conservation Laws,
  Uncertainty Relations, and Quantum Limits of Measurements}},\ }\href
  {https://doi.org/10.1103/PhysRevLett.88.050402} {\bibfield  {journal}
  {\bibinfo  {journal} {Phys. Rev. Lett.}\ }\textbf {\bibinfo {volume} {88}},\
  \bibinfo {pages} {050402} (\bibinfo {year} {2002})}\BibitemShut {NoStop}%
\bibitem [{\citenamefont {Miyadera}\ and\ \citenamefont
  {Imai}(2006)}]{Miyadera2006a}%
  \BibitemOpen
  \bibfield  {author} {\bibinfo {author} {\bibfnamefont {T.}~\bibnamefont
  {Miyadera}}\ and\ \bibinfo {author} {\bibfnamefont {H.}~\bibnamefont
  {Imai}},\ }\bibfield  {title} {\bibinfo {title} {{Wigner-Araki-Yanase theorem
  on distinguishability}},\ }\href {https://doi.org/10.1103/PhysRevA.74.024101}
  {\bibfield  {journal} {\bibinfo  {journal} {Phys. Rev. A}\ }\textbf {\bibinfo
  {volume} {74}},\ \bibinfo {pages} {024101} (\bibinfo {year}
  {2006})}\BibitemShut {NoStop}%
\bibitem [{\citenamefont {Loveridge}\ and\ \citenamefont
  {Busch}(2011)}]{Loveridge2011}%
  \BibitemOpen
  \bibfield  {author} {\bibinfo {author} {\bibfnamefont {L.}~\bibnamefont
  {Loveridge}}\ and\ \bibinfo {author} {\bibfnamefont {P.}~\bibnamefont
  {Busch}},\ }\bibfield  {title} {\bibinfo {title} {{‘Measurement of quantum
  mechanical operators' revisited}},\ }\href
  {https://doi.org/10.1140/epjd/e2011-10714-3} {\bibfield  {journal} {\bibinfo
  {journal} {Eur. Phys. J. D}\ }\textbf {\bibinfo {volume} {62}},\ \bibinfo
  {pages} {297} (\bibinfo {year} {2011})}\BibitemShut {NoStop}%
\bibitem [{\citenamefont {Ahmadi}\ \emph {et~al.}(2013)\citenamefont {Ahmadi},
  \citenamefont {Jennings},\ and\ \citenamefont {Rudolph}}]{Ahmadi2013b}%
  \BibitemOpen
  \bibfield  {author} {\bibinfo {author} {\bibfnamefont {M.}~\bibnamefont
  {Ahmadi}}, \bibinfo {author} {\bibfnamefont {D.}~\bibnamefont {Jennings}},\
  and\ \bibinfo {author} {\bibfnamefont {T.}~\bibnamefont {Rudolph}},\
  }\bibfield  {title} {\bibinfo {title} {{The Wigner–Araki–Yanase theorem
  and the quantum resource theory of asymmetry}},\ }\href
  {https://doi.org/10.1088/1367-2630/15/1/013057} {\bibfield  {journal}
  {\bibinfo  {journal} {New J. Phys.}\ }\textbf {\bibinfo {volume} {15}},\
  \bibinfo {pages} {013057} (\bibinfo {year} {2013})}\BibitemShut {NoStop}%
\bibitem [{\citenamefont {Loveridge}(2020)}]{Loveridge2020a}%
  \BibitemOpen
  \bibfield  {author} {\bibinfo {author} {\bibfnamefont {L.}~\bibnamefont
  {Loveridge}},\ }\bibfield  {title} {\bibinfo {title} {{A relational
  perspective on the Wigner-Araki-Yanase theorem}},\ }\href
  {https://doi.org/10.1088/1742-6596/1638/1/012009} {\bibfield  {journal}
  {\bibinfo  {journal} {J. Phys. Conf. Ser.}\ }\textbf {\bibinfo {volume}
  {1638}},\ \bibinfo {pages} {012009} (\bibinfo {year} {2020})}\BibitemShut
  {NoStop}%
\bibitem [{\citenamefont {Mohammady}\ \emph {et~al.}(2021)\citenamefont
  {Mohammady}, \citenamefont {Miyadera},\ and\ \citenamefont
  {Loveridge}}]{Mohammady2021a}%
  \BibitemOpen
  \bibfield  {author} {\bibinfo {author} {\bibfnamefont {M.~H.}\ \bibnamefont
  {Mohammady}}, \bibinfo {author} {\bibfnamefont {T.}~\bibnamefont
  {Miyadera}},\ and\ \bibinfo {author} {\bibfnamefont {L.}~\bibnamefont
  {Loveridge}},\ }\bibfield  {title} {\bibinfo {title} {{Measurement
  disturbance and conservation laws in quantum mechanics}},\ }\href
  {http://arxiv.org/abs/2110.11705} {\  (\bibinfo {year} {2021})},\ \Eprint
  {https://arxiv.org/abs/2110.11705} {arXiv:2110.11705} \BibitemShut {NoStop}%
\bibitem [{\citenamefont {Kuramochi}\ and\ \citenamefont
  {Tajima}(2022)}]{Kuramochi2022}%
  \BibitemOpen
  \bibfield  {author} {\bibinfo {author} {\bibfnamefont {Y.}~\bibnamefont
  {Kuramochi}}\ and\ \bibinfo {author} {\bibfnamefont {H.}~\bibnamefont
  {Tajima}},\ }\bibfield  {title} {\bibinfo {title} {{Wigner-Araki-Yanase
  theorem for continuous and unbounded conserved observables}},\ }\href
  {http://arxiv.org/abs/2208.13494} {\  (\bibinfo {year} {2022})},\ \Eprint
  {https://arxiv.org/abs/2208.13494} {arXiv:2208.13494} \BibitemShut {NoStop}%
\bibitem [{\citenamefont {Schulman}\ \emph {et~al.}(2005)\citenamefont
  {Schulman}, \citenamefont {Mor},\ and\ \citenamefont
  {Weinstein}}]{Schulman2005}%
  \BibitemOpen
  \bibfield  {author} {\bibinfo {author} {\bibfnamefont {L.~J.}\ \bibnamefont
  {Schulman}}, \bibinfo {author} {\bibfnamefont {T.}~\bibnamefont {Mor}},\ and\
  \bibinfo {author} {\bibfnamefont {Y.}~\bibnamefont {Weinstein}},\ }\bibfield
  {title} {\bibinfo {title} {{Physical Limits of Heat-Bath Algorithmic
  Cooling}},\ }\href {https://doi.org/10.1103/PhysRevLett.94.120501} {\bibfield
   {journal} {\bibinfo  {journal} {Phys. Rev. Lett.}\ }\textbf {\bibinfo
  {volume} {94}},\ \bibinfo {pages} {120501} (\bibinfo {year}
  {2005})}\BibitemShut {NoStop}%
\bibitem [{\citenamefont {Allahverdyan}\ \emph {et~al.}(2011)\citenamefont
  {Allahverdyan}, \citenamefont {Hovhannisyan}, \citenamefont {Janzing},\ and\
  \citenamefont {Mahler}}]{Allahverdyan2011a}%
  \BibitemOpen
  \bibfield  {author} {\bibinfo {author} {\bibfnamefont {A.~E.}\ \bibnamefont
  {Allahverdyan}}, \bibinfo {author} {\bibfnamefont {K.~V.}\ \bibnamefont
  {Hovhannisyan}}, \bibinfo {author} {\bibfnamefont {D.}~\bibnamefont
  {Janzing}},\ and\ \bibinfo {author} {\bibfnamefont {G.}~\bibnamefont
  {Mahler}},\ }\bibfield  {title} {\bibinfo {title} {{Thermodynamic limits of
  dynamic cooling}},\ }\href {https://doi.org/10.1103/PhysRevE.84.041109}
  {\bibfield  {journal} {\bibinfo  {journal} {Phys. Rev. E}\ }\textbf {\bibinfo
  {volume} {84}},\ \bibinfo {pages} {041109} (\bibinfo {year}
  {2011})}\BibitemShut {NoStop}%
\bibitem [{\citenamefont {Reeb}\ and\ \citenamefont {Wolf}(2014)}]{Reeb2013a}%
  \BibitemOpen
  \bibfield  {author} {\bibinfo {author} {\bibfnamefont {D.}~\bibnamefont
  {Reeb}}\ and\ \bibinfo {author} {\bibfnamefont {M.~M.}\ \bibnamefont
  {Wolf}},\ }\bibfield  {title} {\bibinfo {title} {{An improved Landauer
  principle with finite-size corrections}},\ }\href
  {https://doi.org/10.1088/1367-2630/16/10/103011} {\bibfield  {journal}
  {\bibinfo  {journal} {New J. Phys.}\ }\textbf {\bibinfo {volume} {16}},\
  \bibinfo {pages} {103011} (\bibinfo {year} {2014})}\BibitemShut {NoStop}%
\bibitem [{\citenamefont {Masanes}\ and\ \citenamefont
  {Oppenheim}(2017)}]{Masanes2014}%
  \BibitemOpen
  \bibfield  {author} {\bibinfo {author} {\bibfnamefont {L.}~\bibnamefont
  {Masanes}}\ and\ \bibinfo {author} {\bibfnamefont {J.}~\bibnamefont
  {Oppenheim}},\ }\bibfield  {title} {\bibinfo {title} {{A general derivation
  and quantification of the third law of thermodynamics}},\ }\href
  {https://doi.org/10.1038/ncomms14538} {\bibfield  {journal} {\bibinfo
  {journal} {Nat. Commun.}\ }\textbf {\bibinfo {volume} {8}},\ \bibinfo {pages}
  {14538} (\bibinfo {year} {2017})}\BibitemShut {NoStop}%
\bibitem [{\citenamefont {Ticozzi}\ and\ \citenamefont
  {Viola}(2015)}]{Ticozzi2014}%
  \BibitemOpen
  \bibfield  {author} {\bibinfo {author} {\bibfnamefont {F.}~\bibnamefont
  {Ticozzi}}\ and\ \bibinfo {author} {\bibfnamefont {L.}~\bibnamefont
  {Viola}},\ }\bibfield  {title} {\bibinfo {title} {{Quantum resources for
  purification and cooling: fundamental limits and opportunities}},\ }\href
  {https://doi.org/10.1038/srep05192} {\bibfield  {journal} {\bibinfo
  {journal} {Sci. Rep.}\ }\textbf {\bibinfo {volume} {4}},\ \bibinfo {pages}
  {5192} (\bibinfo {year} {2015})}\BibitemShut {NoStop}%
\bibitem [{\citenamefont {Scharlau}\ and\ \citenamefont
  {Mueller}(2018)}]{Scharlau2016a}%
  \BibitemOpen
  \bibfield  {author} {\bibinfo {author} {\bibfnamefont {J.}~\bibnamefont
  {Scharlau}}\ and\ \bibinfo {author} {\bibfnamefont {M.~P.}\ \bibnamefont
  {Mueller}},\ }\bibfield  {title} {\bibinfo {title} {{Quantum Horn's lemma,
  finite heat baths, and the third law of thermodynamics}},\ }\href
  {https://doi.org/10.22331/q-2018-02-22-54} {\bibfield  {journal} {\bibinfo
  {journal} {Quantum}\ }\textbf {\bibinfo {volume} {2}},\ \bibinfo {pages} {54}
  (\bibinfo {year} {2018})}\BibitemShut {NoStop}%
\bibitem [{\citenamefont {Wilming}\ and\ \citenamefont
  {Gallego}(2017)}]{Wilming2017a}%
  \BibitemOpen
  \bibfield  {author} {\bibinfo {author} {\bibfnamefont {H.}~\bibnamefont
  {Wilming}}\ and\ \bibinfo {author} {\bibfnamefont {R.}~\bibnamefont
  {Gallego}},\ }\bibfield  {title} {\bibinfo {title} {{Third Law of
  Thermodynamics as a Single Inequality}},\ }\href
  {https://doi.org/10.1103/PhysRevX.7.041033} {\bibfield  {journal} {\bibinfo
  {journal} {Phys. Rev. X}\ }\textbf {\bibinfo {volume} {7}},\ \bibinfo {pages}
  {041033} (\bibinfo {year} {2017})}\BibitemShut {NoStop}%
\bibitem [{\citenamefont {Freitas}\ \emph {et~al.}(2018)\citenamefont
  {Freitas}, \citenamefont {Gallego}, \citenamefont {Masanes},\ and\
  \citenamefont {Paz}}]{Freitas2018}%
  \BibitemOpen
  \bibfield  {author} {\bibinfo {author} {\bibfnamefont {N.}~\bibnamefont
  {Freitas}}, \bibinfo {author} {\bibfnamefont {R.}~\bibnamefont {Gallego}},
  \bibinfo {author} {\bibfnamefont {L.}~\bibnamefont {Masanes}},\ and\ \bibinfo
  {author} {\bibfnamefont {J.~P.}\ \bibnamefont {Paz}},\ }\bibfield  {title}
  {\bibinfo {title} {{Cooling to Absolute Zero: The Unattainability
  Principle}},\ }in\ \href {https://doi.org/10.1007/978-3-319-99046-0_25}
  {\emph {\bibinfo {booktitle} {Fundam. Theor. Phys.}}},\ Vol.\ \bibinfo
  {volume} {195}\ (\bibinfo  {publisher} {Springer International Publishing},\
  \bibinfo {year} {2018})\ pp.\ \bibinfo {pages} {597--622}\BibitemShut
  {NoStop}%
\bibitem [{\citenamefont {Clivaz}\ \emph {et~al.}(2019)\citenamefont {Clivaz},
  \citenamefont {Silva}, \citenamefont {Haack}, \citenamefont {Brask},
  \citenamefont {Brunner},\ and\ \citenamefont {Huber}}]{Clivaz2019}%
  \BibitemOpen
  \bibfield  {author} {\bibinfo {author} {\bibfnamefont {F.}~\bibnamefont
  {Clivaz}}, \bibinfo {author} {\bibfnamefont {R.}~\bibnamefont {Silva}},
  \bibinfo {author} {\bibfnamefont {G.}~\bibnamefont {Haack}}, \bibinfo
  {author} {\bibfnamefont {J.~B.}\ \bibnamefont {Brask}}, \bibinfo {author}
  {\bibfnamefont {N.}~\bibnamefont {Brunner}},\ and\ \bibinfo {author}
  {\bibfnamefont {M.}~\bibnamefont {Huber}},\ }\bibfield  {title} {\bibinfo
  {title} {{Unifying Paradigms of Quantum Refrigeration: A Universal and
  Attainable Bound on Cooling}},\ }\href
  {https://doi.org/10.1103/PhysRevLett.123.170605} {\bibfield  {journal}
  {\bibinfo  {journal} {Phys. Rev. Lett.}\ }\textbf {\bibinfo {volume} {123}},\
  \bibinfo {pages} {170605} (\bibinfo {year} {2019})}\BibitemShut {NoStop}%
\bibitem [{\citenamefont {Taranto}\ \emph {et~al.}(2021)\citenamefont
  {Taranto}, \citenamefont {Bakhshinezhad}, \citenamefont {Bluhm},
  \citenamefont {Silva}, \citenamefont {Friis}, \citenamefont {Lock},
  \citenamefont {Vitagliano}, \citenamefont {Binder}, \citenamefont {Debarba},
  \citenamefont {Schwarzhans}, \citenamefont {Clivaz},\ and\ \citenamefont
  {Huber}}]{Taranto2021}%
  \BibitemOpen
  \bibfield  {author} {\bibinfo {author} {\bibfnamefont {P.}~\bibnamefont
  {Taranto}}, \bibinfo {author} {\bibfnamefont {F.}~\bibnamefont
  {Bakhshinezhad}}, \bibinfo {author} {\bibfnamefont {A.}~\bibnamefont
  {Bluhm}}, \bibinfo {author} {\bibfnamefont {R.}~\bibnamefont {Silva}},
  \bibinfo {author} {\bibfnamefont {N.}~\bibnamefont {Friis}}, \bibinfo
  {author} {\bibfnamefont {M.~P.~E.}\ \bibnamefont {Lock}}, \bibinfo {author}
  {\bibfnamefont {G.}~\bibnamefont {Vitagliano}}, \bibinfo {author}
  {\bibfnamefont {F.~C.}\ \bibnamefont {Binder}}, \bibinfo {author}
  {\bibfnamefont {T.}~\bibnamefont {Debarba}}, \bibinfo {author} {\bibfnamefont
  {E.}~\bibnamefont {Schwarzhans}}, \bibinfo {author} {\bibfnamefont
  {F.}~\bibnamefont {Clivaz}},\ and\ \bibinfo {author} {\bibfnamefont
  {M.}~\bibnamefont {Huber}},\ }\bibfield  {title} {\bibinfo {title} {{Landauer
  vs. Nernst: What is the True Cost of Cooling a Quantum System?}},\ }\href
  {http://arxiv.org/abs/2106.05151} {\  (\bibinfo {year} {2021})},\ \Eprint
  {https://arxiv.org/abs/2106.05151} {arXiv:2106.05151} \BibitemShut {NoStop}%
\bibitem [{\citenamefont {Buffoni}\ \emph {et~al.}(2022)\citenamefont
  {Buffoni}, \citenamefont {Gherardini}, \citenamefont {{Zambrini Cruzeiro}},\
  and\ \citenamefont {Omar}}]{Buffoni2022}%
  \BibitemOpen
  \bibfield  {author} {\bibinfo {author} {\bibfnamefont {L.}~\bibnamefont
  {Buffoni}}, \bibinfo {author} {\bibfnamefont {S.}~\bibnamefont {Gherardini}},
  \bibinfo {author} {\bibfnamefont {E.}~\bibnamefont {{Zambrini Cruzeiro}}},\
  and\ \bibinfo {author} {\bibfnamefont {Y.}~\bibnamefont {Omar}},\ }\bibfield
  {title} {\bibinfo {title} {{Third Law of Thermodynamics and the Scaling of
  Quantum Computers}},\ }\href {https://doi.org/10.1103/PhysRevLett.129.150602}
  {\bibfield  {journal} {\bibinfo  {journal} {Phys. Rev. Lett.}\ }\textbf
  {\bibinfo {volume} {129}},\ \bibinfo {pages} {150602} (\bibinfo {year}
  {2022})}\BibitemShut {NoStop}%
\bibitem [{\citenamefont {Guryanova}\ \emph {et~al.}(2020)\citenamefont
  {Guryanova}, \citenamefont {Friis},\ and\ \citenamefont
  {Huber}}]{Guryanova2018}%
  \BibitemOpen
  \bibfield  {author} {\bibinfo {author} {\bibfnamefont {Y.}~\bibnamefont
  {Guryanova}}, \bibinfo {author} {\bibfnamefont {N.}~\bibnamefont {Friis}},\
  and\ \bibinfo {author} {\bibfnamefont {M.}~\bibnamefont {Huber}},\ }\bibfield
   {title} {\bibinfo {title} {{Ideal Projective Measurements Have Infinite
  Resource Costs}},\ }\href {https://doi.org/10.22331/q-2020-01-13-222}
  {\bibfield  {journal} {\bibinfo  {journal} {Quantum}\ }\textbf {\bibinfo
  {volume} {4}},\ \bibinfo {pages} {222} (\bibinfo {year} {2020})}\BibitemShut
  {NoStop}%
\bibitem [{\citenamefont {Busch}\ \emph
  {et~al.}(1995{\natexlab{a}})\citenamefont {Busch}, \citenamefont
  {Grabowski},\ and\ \citenamefont {Lahti}}]{PaulBuschMarianGrabowski1995}%
  \BibitemOpen
  \bibfield  {author} {\bibinfo {author} {\bibfnamefont {P.}~\bibnamefont
  {Busch}}, \bibinfo {author} {\bibfnamefont {M.}~\bibnamefont {Grabowski}},\
  and\ \bibinfo {author} {\bibfnamefont {P.~J.}\ \bibnamefont {Lahti}},\ }\href
  {https://doi.org/10.1007/978-3-540-49239-9} {\emph {\bibinfo {title}
  {{Operational Quantum Physics}}}},\ \bibinfo {series} {Lecture Notes in
  Physics Monographs}, Vol.~\bibinfo {volume} {31}\ (\bibinfo  {publisher}
  {Springer Berlin Heidelberg},\ \bibinfo {address} {Berlin, Heidelberg},\
  \bibinfo {year} {1995})\BibitemShut {NoStop}%
\bibitem [{\citenamefont {Busch}\ \emph {et~al.}(1996)\citenamefont {Busch},
  \citenamefont {Lahti},\ and\ \citenamefont {{Peter
  Mittelstaedt}}}]{Busch1996}%
  \BibitemOpen
  \bibfield  {author} {\bibinfo {author} {\bibfnamefont {P.}~\bibnamefont
  {Busch}}, \bibinfo {author} {\bibfnamefont {P.~J.}\ \bibnamefont {Lahti}},\
  and\ \bibinfo {author} {\bibnamefont {{Peter Mittelstaedt}}},\ }\href
  {https://doi.org/10.1007/978-3-540-37205-9} {\emph {\bibinfo {title} {{The
  Quantum Theory of Measurement}}}},\ \bibinfo {series} {Lecture Notes in
  Physics Monographs}, Vol.~\bibinfo {volume} {2}\ (\bibinfo  {publisher}
  {Springer Berlin Heidelberg},\ \bibinfo {address} {Berlin, Heidelberg},\
  \bibinfo {year} {1996})\BibitemShut {NoStop}%
\bibitem [{\citenamefont {Heinosaari}\ and\ \citenamefont
  {Ziman}(2011)}]{Heinosaari2011}%
  \BibitemOpen
  \bibfield  {author} {\bibinfo {author} {\bibfnamefont {T.}~\bibnamefont
  {Heinosaari}}\ and\ \bibinfo {author} {\bibfnamefont {M.}~\bibnamefont
  {Ziman}},\ }\href {https://doi.org/10.1017/CBO9781139031103} {\emph {\bibinfo
  {title} {{The Mathematical language of Quantum Theory}}}}\ (\bibinfo
  {publisher} {Cambridge University Press},\ \bibinfo {address} {Cambridge},\
  \bibinfo {year} {2011})\BibitemShut {NoStop}%
\bibitem [{\citenamefont {Busch}\ \emph {et~al.}(2016)\citenamefont {Busch},
  \citenamefont {Lahti}, \citenamefont {Pellonp{\"{a}}{\"{a}}},\ and\
  \citenamefont {Ylinen}}]{Busch2016a}%
  \BibitemOpen
  \bibfield  {author} {\bibinfo {author} {\bibfnamefont {P.}~\bibnamefont
  {Busch}}, \bibinfo {author} {\bibfnamefont {P.}~\bibnamefont {Lahti}},
  \bibinfo {author} {\bibfnamefont {J.-P.}\ \bibnamefont
  {Pellonp{\"{a}}{\"{a}}}},\ and\ \bibinfo {author} {\bibfnamefont
  {K.}~\bibnamefont {Ylinen}},\ }\href
  {https://doi.org/10.1007/978-3-319-43389-9} {\emph {\bibinfo {title}
  {{Quantum Measurement}}}},\ Theoretical and Mathematical Physics\ (\bibinfo
  {publisher} {Springer International Publishing},\ \bibinfo {address} {Cham},\
  \bibinfo {year} {2016})\BibitemShut {NoStop}%
\bibitem [{\citenamefont {Busch}\ and\ \citenamefont
  {Jaeger}(2010)}]{Busch2010a}%
  \BibitemOpen
  \bibfield  {author} {\bibinfo {author} {\bibfnamefont {P.}~\bibnamefont
  {Busch}}\ and\ \bibinfo {author} {\bibfnamefont {G.}~\bibnamefont {Jaeger}},\
  }\bibfield  {title} {\bibinfo {title} {{Unsharp Quantum Reality}},\ }\href
  {https://doi.org/10.1007/s10701-010-9497-0} {\bibfield  {journal} {\bibinfo
  {journal} {Found. Phys.}\ }\textbf {\bibinfo {volume} {40}},\ \bibinfo
  {pages} {1341} (\bibinfo {year} {2010})}\BibitemShut {NoStop}%
\bibitem [{\citenamefont {Jaeger}(2019)}]{Jaeger2019}%
  \BibitemOpen
  \bibfield  {author} {\bibinfo {author} {\bibfnamefont {G.}~\bibnamefont
  {Jaeger}},\ }\bibfield  {title} {\bibinfo {title} {{Quantum Unsharpness,
  Potentiality, and Reality}},\ }\href
  {https://doi.org/10.1007/s10701-019-00273-z} {\bibfield  {journal} {\bibinfo
  {journal} {Found. Phys.}\ }\textbf {\bibinfo {volume} {49}},\ \bibinfo
  {pages} {663} (\bibinfo {year} {2019})}\BibitemShut {NoStop}%
\bibitem [{\citenamefont {Davies}\ and\ \citenamefont
  {Lewis}(1970)}]{Davies1970}%
  \BibitemOpen
  \bibfield  {author} {\bibinfo {author} {\bibfnamefont {E.~B.}\ \bibnamefont
  {Davies}}\ and\ \bibinfo {author} {\bibfnamefont {J.~T.}\ \bibnamefont
  {Lewis}},\ }\bibfield  {title} {\bibinfo {title} {{An operational approach to
  quantum probability}},\ }\href {https://doi.org/10.1007/BF01647093}
  {\bibfield  {journal} {\bibinfo  {journal} {Commun. Math. Phys.}\ }\textbf
  {\bibinfo {volume} {17}},\ \bibinfo {pages} {239} (\bibinfo {year}
  {1970})}\BibitemShut {NoStop}%
\bibitem [{\citenamefont {Lahti}\ \emph {et~al.}(1991)\citenamefont {Lahti},
  \citenamefont {Busch},\ and\ \citenamefont {Mittelstaedt}}]{Lahti1991}%
  \BibitemOpen
  \bibfield  {author} {\bibinfo {author} {\bibfnamefont {P.~J.}\ \bibnamefont
  {Lahti}}, \bibinfo {author} {\bibfnamefont {P.}~\bibnamefont {Busch}},\ and\
  \bibinfo {author} {\bibfnamefont {P.}~\bibnamefont {Mittelstaedt}},\
  }\bibfield  {title} {\bibinfo {title} {{Some important classes of quantum
  measurements and their information gain}},\ }\href
  {https://doi.org/10.1063/1.529504} {\bibfield  {journal} {\bibinfo  {journal}
  {J. Math. Phys.}\ }\textbf {\bibinfo {volume} {32}},\ \bibinfo {pages} {2770}
  (\bibinfo {year} {1991})}\BibitemShut {NoStop}%
\bibitem [{\citenamefont {Heinosaari}\ and\ \citenamefont
  {Wolf}(2010)}]{Heinosaari2010}%
  \BibitemOpen
  \bibfield  {author} {\bibinfo {author} {\bibfnamefont {T.}~\bibnamefont
  {Heinosaari}}\ and\ \bibinfo {author} {\bibfnamefont {M.~M.}\ \bibnamefont
  {Wolf}},\ }\bibfield  {title} {\bibinfo {title} {{Nondisturbing quantum
  measurements}},\ }\href {https://doi.org/10.1063/1.3480658} {\bibfield
  {journal} {\bibinfo  {journal} {J. Math. Phys.}\ }\textbf {\bibinfo {volume}
  {51}},\ \bibinfo {pages} {092201} (\bibinfo {year} {2010})}\BibitemShut
  {NoStop}%
\bibitem [{\citenamefont {D'Ariano}\ \emph {et~al.}(2011)\citenamefont
  {D'Ariano}, \citenamefont {Perinotti},\ and\ \citenamefont
  {Sedl{\'{a}}k}}]{DAriano2011}%
  \BibitemOpen
  \bibfield  {author} {\bibinfo {author} {\bibfnamefont {G.~M.}\ \bibnamefont
  {D'Ariano}}, \bibinfo {author} {\bibfnamefont {P.}~\bibnamefont
  {Perinotti}},\ and\ \bibinfo {author} {\bibfnamefont {M.}~\bibnamefont
  {Sedl{\'{a}}k}},\ }\bibfield  {title} {\bibinfo {title} {{Extremal quantum
  protocols}},\ }\href {https://doi.org/10.1063/1.3610676} {\bibfield
  {journal} {\bibinfo  {journal} {J. Math. Phys.}\ }\textbf {\bibinfo {volume}
  {52}},\ \bibinfo {pages} {82202} (\bibinfo {year} {2011})}\BibitemShut
  {NoStop}%
\bibitem [{\citenamefont {Gour}\ and\ \citenamefont {Wilde}(2020)}]{Gour2020b}%
  \BibitemOpen
  \bibfield  {author} {\bibinfo {author} {\bibfnamefont {G.}~\bibnamefont
  {Gour}}\ and\ \bibinfo {author} {\bibfnamefont {M.~M.}\ \bibnamefont
  {Wilde}},\ }\bibfield  {title} {\bibinfo {title} {{Entropy of a Quantum
  Channel: Definition, Properties, and Application}},\ }in\ \href
  {https://doi.org/10.1109/ISIT44484.2020.9174135} {\emph {\bibinfo {booktitle}
  {2020 IEEE Int. Symp. Inf. Theory}}},\ Vol.\ \bibinfo {volume} {2020-June}\
  (\bibinfo  {publisher} {IEEE},\ \bibinfo {year} {2020})\ pp.\ \bibinfo
  {pages} {1903--1908}\BibitemShut {NoStop}%
\bibitem [{\citenamefont {Martens}\ and\ \citenamefont
  {de~Muynck}(1990)}]{Holland1990}%
  \BibitemOpen
  \bibfield  {author} {\bibinfo {author} {\bibfnamefont {H.}~\bibnamefont
  {Martens}}\ and\ \bibinfo {author} {\bibfnamefont {W.~M.}\ \bibnamefont
  {de~Muynck}},\ }\bibfield  {title} {\bibinfo {title} {{Nonideal quantum
  measurements}},\ }\href {https://doi.org/10.1007/BF00731693} {\bibfield
  {journal} {\bibinfo  {journal} {Found. Phys.}\ }\textbf {\bibinfo {volume}
  {20}},\ \bibinfo {pages} {255} (\bibinfo {year} {1990})}\BibitemShut
  {NoStop}%
\bibitem [{\citenamefont {Pellonp{\"{a}}{\"{a}}}(2014)}]{Pellonpaa2014}%
  \BibitemOpen
  \bibfield  {author} {\bibinfo {author} {\bibfnamefont {J.-P.}\ \bibnamefont
  {Pellonp{\"{a}}{\"{a}}}},\ }\bibfield  {title} {\bibinfo {title} {{Complete
  Measurements of Quantum Observables}},\ }\href
  {https://doi.org/10.1007/s10701-013-9764-y} {\bibfield  {journal} {\bibinfo
  {journal} {Found. Phys.}\ }\textbf {\bibinfo {volume} {44}},\ \bibinfo
  {pages} {71} (\bibinfo {year} {2014})}\BibitemShut {NoStop}%
\bibitem [{\citenamefont {Heinosaari}\ \emph {et~al.}(2016)\citenamefont
  {Heinosaari}, \citenamefont {Miyadera},\ and\ \citenamefont
  {Ziman}}]{Heinosaari2015}%
  \BibitemOpen
  \bibfield  {author} {\bibinfo {author} {\bibfnamefont {T.}~\bibnamefont
  {Heinosaari}}, \bibinfo {author} {\bibfnamefont {T.}~\bibnamefont
  {Miyadera}},\ and\ \bibinfo {author} {\bibfnamefont {M.}~\bibnamefont
  {Ziman}},\ }\bibfield  {title} {\bibinfo {title} {{An invitation to quantum
  incompatibility}},\ }\href {https://doi.org/10.1088/1751-8113/49/12/123001}
  {\bibfield  {journal} {\bibinfo  {journal} {J. Phys. A Math. Theor.}\
  }\textbf {\bibinfo {volume} {49}},\ \bibinfo {pages} {123001} (\bibinfo
  {year} {2016})}\BibitemShut {NoStop}%
\bibitem [{\citenamefont {Ozawa}(2001)}]{Ozawa2001}%
  \BibitemOpen
  \bibfield  {author} {\bibinfo {author} {\bibfnamefont {M.}~\bibnamefont
  {Ozawa}},\ }\bibfield  {title} {\bibinfo {title} {{Operations, disturbance,
  and simultaneous measurability}},\ }\href
  {https://doi.org/10.1103/PhysRevA.63.032109} {\bibfield  {journal} {\bibinfo
  {journal} {Phys. Rev. A}\ }\textbf {\bibinfo {volume} {63}},\ \bibinfo
  {pages} {032109} (\bibinfo {year} {2001})}\BibitemShut {NoStop}%
\bibitem [{\citenamefont
  {Pellonp{\"{a}}{\"{a}}}(2013{\natexlab{a}})}]{Pellonpaa2013a}%
  \BibitemOpen
  \bibfield  {author} {\bibinfo {author} {\bibfnamefont {J.-P.}\ \bibnamefont
  {Pellonp{\"{a}}{\"{a}}}},\ }\bibfield  {title} {\bibinfo {title} {{Quantum
  instruments: II. Measurement theory}},\ }\href
  {https://doi.org/10.1088/1751-8113/46/2/025303} {\bibfield  {journal}
  {\bibinfo  {journal} {J. Phys. A Math. Theor.}\ }\textbf {\bibinfo {volume}
  {46}},\ \bibinfo {pages} {025303} (\bibinfo {year}
  {2013}{\natexlab{a}})}\BibitemShut {NoStop}%
\bibitem [{\citenamefont {Ozawa}(1984)}]{Ozawa1984}%
  \BibitemOpen
  \bibfield  {author} {\bibinfo {author} {\bibfnamefont {M.}~\bibnamefont
  {Ozawa}},\ }\bibfield  {title} {\bibinfo {title} {{Quantum measuring
  processes of continuous observables}},\ }\href
  {https://doi.org/10.1063/1.526000} {\bibfield  {journal} {\bibinfo  {journal}
  {J. Math. Phys.}\ }\textbf {\bibinfo {volume} {25}},\ \bibinfo {pages} {79}
  (\bibinfo {year} {1984})}\BibitemShut {NoStop}%
\bibitem [{\citenamefont {Busch}\ and\ \citenamefont
  {Singh}(1998)}]{Busch1998}%
  \BibitemOpen
  \bibfield  {author} {\bibinfo {author} {\bibfnamefont {P.}~\bibnamefont
  {Busch}}\ and\ \bibinfo {author} {\bibfnamefont {J.}~\bibnamefont {Singh}},\
  }\bibfield  {title} {\bibinfo {title} {{L{\"{u}}ders theorem for unsharp
  quantum measurements}},\ }\href
  {https://doi.org/10.1016/S0375-9601(98)00704-X} {\bibfield  {journal}
  {\bibinfo  {journal} {Phys. Lett. A}\ }\textbf {\bibinfo {volume} {249}},\
  \bibinfo {pages} {10} (\bibinfo {year} {1998})}\BibitemShut {NoStop}%
\bibitem [{\citenamefont {Miyadera}\ and\ \citenamefont
  {Imai}(2008)}]{Miyadera2008}%
  \BibitemOpen
  \bibfield  {author} {\bibinfo {author} {\bibfnamefont {T.}~\bibnamefont
  {Miyadera}}\ and\ \bibinfo {author} {\bibfnamefont {H.}~\bibnamefont
  {Imai}},\ }\bibfield  {title} {\bibinfo {title} {{Heisenberg's uncertainty
  principle for simultaneous measurement of positive-operator-valued
  measures}},\ }\href {https://doi.org/10.1103/PhysRevA.78.052119} {\bibfield
  {journal} {\bibinfo  {journal} {Phys. Rev. A}\ }\textbf {\bibinfo {volume}
  {78}},\ \bibinfo {pages} {052119} (\bibinfo {year} {2008})}\BibitemShut
  {NoStop}%
\bibitem [{\citenamefont {Heinosaari}\ and\ \citenamefont
  {Pellonp{\"{a}}{\"{a}}}(2011)}]{Heinosaari2011a}%
  \BibitemOpen
  \bibfield  {author} {\bibinfo {author} {\bibfnamefont {T.}~\bibnamefont
  {Heinosaari}}\ and\ \bibinfo {author} {\bibfnamefont {J.-P.}\ \bibnamefont
  {Pellonp{\"{a}}{\"{a}}}},\ }\bibfield  {title} {\bibinfo {title} {{Extreme
  commutative quantum observables are sharp}},\ }\href
  {https://doi.org/10.1088/1751-8113/44/31/315303} {\bibfield  {journal}
  {\bibinfo  {journal} {J. Phys. A Math. Theor.}\ }\textbf {\bibinfo {volume}
  {44}},\ \bibinfo {pages} {315303} (\bibinfo {year} {2011})}\BibitemShut
  {NoStop}%
\bibitem [{\citenamefont {Busch}\ \emph
  {et~al.}(1995{\natexlab{b}})\citenamefont {Busch}, \citenamefont
  {Grabowski},\ and\ \citenamefont {Lahti}}]{Busch1995}%
  \BibitemOpen
  \bibfield  {author} {\bibinfo {author} {\bibfnamefont {P.}~\bibnamefont
  {Busch}}, \bibinfo {author} {\bibfnamefont {M.}~\bibnamefont {Grabowski}},\
  and\ \bibinfo {author} {\bibfnamefont {P.~J.}\ \bibnamefont {Lahti}},\
  }\bibfield  {title} {\bibinfo {title} {{Repeatable measurements in quantum
  theory: Their role and feasibility}},\ }\href
  {https://doi.org/10.1007/BF02055331} {\bibfield  {journal} {\bibinfo
  {journal} {Found. Phys.}\ }\textbf {\bibinfo {volume} {25}},\ \bibinfo
  {pages} {1239} (\bibinfo {year} {1995}{\natexlab{b}})}\BibitemShut {NoStop}%
\bibitem [{\citenamefont {Busch}\ and\ \citenamefont
  {Lahti}(1996)}]{Busch1996b}%
  \BibitemOpen
  \bibfield  {author} {\bibinfo {author} {\bibfnamefont {P.}~\bibnamefont
  {Busch}}\ and\ \bibinfo {author} {\bibfnamefont {P.~J.}\ \bibnamefont
  {Lahti}},\ }\bibfield  {title} {\bibinfo {title} {{Correlation properties of
  quantum measurements}},\ }\href {https://doi.org/10.1063/1.531530} {\bibfield
   {journal} {\bibinfo  {journal} {J. Math. Phys.}\ }\textbf {\bibinfo {volume}
  {37}},\ \bibinfo {pages} {2585} (\bibinfo {year} {1996})}\BibitemShut
  {NoStop}%
\bibitem [{\citenamefont {Busch}\ \emph {et~al.}(1990)\citenamefont {Busch},
  \citenamefont {Cassinelli},\ and\ \citenamefont {Lahti}}]{Busch1990}%
  \BibitemOpen
  \bibfield  {author} {\bibinfo {author} {\bibfnamefont {P.}~\bibnamefont
  {Busch}}, \bibinfo {author} {\bibfnamefont {G.}~\bibnamefont {Cassinelli}},\
  and\ \bibinfo {author} {\bibfnamefont {P.~J.}\ \bibnamefont {Lahti}},\
  }\bibfield  {title} {\bibinfo {title} {{On the quantum theory of sequential
  measurements}},\ }\href {https://doi.org/10.1007/BF01889690} {\bibfield
  {journal} {\bibinfo  {journal} {Found. Phys.}\ }\textbf {\bibinfo {volume}
  {20}},\ \bibinfo {pages} {757} (\bibinfo {year} {1990})}\BibitemShut
  {NoStop}%
\bibitem [{\citenamefont
  {Pellonp{\"{a}}{\"{a}}}(2013{\natexlab{b}})}]{Pellonpaa2013}%
  \BibitemOpen
  \bibfield  {author} {\bibinfo {author} {\bibfnamefont {J.-P.}\ \bibnamefont
  {Pellonp{\"{a}}{\"{a}}}},\ }\bibfield  {title} {\bibinfo {title} {{Quantum
  instruments: I. Extreme instruments}},\ }\href
  {https://doi.org/10.1088/1751-8113/46/2/025302} {\bibfield  {journal}
  {\bibinfo  {journal} {J. Phys. A Math. Theor.}\ }\textbf {\bibinfo {volume}
  {46}},\ \bibinfo {pages} {025302} (\bibinfo {year}
  {2013}{\natexlab{b}})}\BibitemShut {NoStop}%
\bibitem [{\citenamefont {Einstein}\ \emph {et~al.}(1935)\citenamefont
  {Einstein}, \citenamefont {Podolsky},\ and\ \citenamefont
  {Rosen}}]{Einstein1935}%
  \BibitemOpen
  \bibfield  {author} {\bibinfo {author} {\bibfnamefont {A.}~\bibnamefont
  {Einstein}}, \bibinfo {author} {\bibfnamefont {B.}~\bibnamefont {Podolsky}},\
  and\ \bibinfo {author} {\bibfnamefont {N.}~\bibnamefont {Rosen}},\ }\bibfield
   {title} {\bibinfo {title} {{Can Quantum-Mechanical Description of Physical
  Reality Be Considered Complete?}},\ }\href
  {https://doi.org/10.1103/PhysRev.47.777} {\bibfield  {journal} {\bibinfo
  {journal} {Phys. Rev.}\ }\textbf {\bibinfo {volume} {47}},\ \bibinfo {pages}
  {777} (\bibinfo {year} {1935})}\BibitemShut {NoStop}%
\bibitem [{\citenamefont {Leggett}\ and\ \citenamefont
  {Garg}(1985)}]{Leggett1985}%
  \BibitemOpen
  \bibfield  {author} {\bibinfo {author} {\bibfnamefont {A.~J.}\ \bibnamefont
  {Leggett}}\ and\ \bibinfo {author} {\bibfnamefont {A.}~\bibnamefont {Garg}},\
  }\bibfield  {title} {\bibinfo {title} {{Quantum mechanics versus macroscopic
  realism: Is the flux there when nobody looks?}},\ }\href
  {https://doi.org/10.1103/PhysRevLett.54.857} {\bibfield  {journal} {\bibinfo
  {journal} {Phys. Rev. Lett.}\ }\textbf {\bibinfo {volume} {54}},\ \bibinfo
  {pages} {857} (\bibinfo {year} {1985})}\BibitemShut {NoStop}%
\bibitem [{\citenamefont {Sagawa}\ and\ \citenamefont
  {Ueda}(2009)}]{Sagawa2009b}%
  \BibitemOpen
  \bibfield  {author} {\bibinfo {author} {\bibfnamefont {T.}~\bibnamefont
  {Sagawa}}\ and\ \bibinfo {author} {\bibfnamefont {M.}~\bibnamefont {Ueda}},\
  }\bibfield  {title} {\bibinfo {title} {{Minimal Energy Cost for Thermodynamic
  Information Processing: Measurement and Information Erasure}},\ }\href
  {https://doi.org/10.1103/PhysRevLett.102.250602} {\bibfield  {journal}
  {\bibinfo  {journal} {Phys. Rev. Lett.}\ }\textbf {\bibinfo {volume} {102}},\
  \bibinfo {pages} {250602} (\bibinfo {year} {2009})}\BibitemShut {NoStop}%
\bibitem [{\citenamefont {Miyadera}(2011)}]{Miyadera2011d}%
  \BibitemOpen
  \bibfield  {author} {\bibinfo {author} {\bibfnamefont {T.}~\bibnamefont
  {Miyadera}},\ }\bibfield  {title} {\bibinfo {title} {{Relation between
  strength of interaction and accuracy of measurement for a quantum
  measurement}},\ }\href {https://doi.org/10.1103/PhysRevA.83.052119}
  {\bibfield  {journal} {\bibinfo  {journal} {Phys. Rev. A}\ }\textbf {\bibinfo
  {volume} {83}},\ \bibinfo {pages} {052119} (\bibinfo {year}
  {2011})}\BibitemShut {NoStop}%
\bibitem [{\citenamefont {Jacobs}(2012)}]{Jacobs2012a}%
  \BibitemOpen
  \bibfield  {author} {\bibinfo {author} {\bibfnamefont {K.}~\bibnamefont
  {Jacobs}},\ }\bibfield  {title} {\bibinfo {title} {{Quantum measurement and
  the first law of thermodynamics: The energy cost of measurement is the work
  value of the acquired information}},\ }\href
  {https://doi.org/10.1103/PhysRevE.86.040106} {\bibfield  {journal} {\bibinfo
  {journal} {Phys. Rev. E}\ }\textbf {\bibinfo {volume} {86}},\ \bibinfo
  {pages} {040106(R)} (\bibinfo {year} {2012})}\BibitemShut {NoStop}%
\bibitem [{\citenamefont {Funo}\ \emph {et~al.}(2013)\citenamefont {Funo},
  \citenamefont {Watanabe},\ and\ \citenamefont {Ueda}}]{Funo2013}%
  \BibitemOpen
  \bibfield  {author} {\bibinfo {author} {\bibfnamefont {K.}~\bibnamefont
  {Funo}}, \bibinfo {author} {\bibfnamefont {Y.}~\bibnamefont {Watanabe}},\
  and\ \bibinfo {author} {\bibfnamefont {M.}~\bibnamefont {Ueda}},\ }\bibfield
  {title} {\bibinfo {title} {{Integral quantum fluctuation theorems under
  measurement and feedback control}},\ }\href
  {https://doi.org/10.1103/PhysRevE.88.052121} {\bibfield  {journal} {\bibinfo
  {journal} {Phys. Rev. E}\ }\textbf {\bibinfo {volume} {88}},\ \bibinfo
  {pages} {052121} (\bibinfo {year} {2013})}\BibitemShut {NoStop}%
\bibitem [{\citenamefont {Navascu{\'{e}}s}\ and\ \citenamefont
  {Popescu}(2014)}]{Navascues2014a}%
  \BibitemOpen
  \bibfield  {author} {\bibinfo {author} {\bibfnamefont {M.}~\bibnamefont
  {Navascu{\'{e}}s}}\ and\ \bibinfo {author} {\bibfnamefont {S.}~\bibnamefont
  {Popescu}},\ }\bibfield  {title} {\bibinfo {title} {{How Energy Conservation
  Limits Our Measurements}},\ }\href
  {https://doi.org/10.1103/PhysRevLett.112.140502} {\bibfield  {journal}
  {\bibinfo  {journal} {Phys. Rev. Lett.}\ }\textbf {\bibinfo {volume} {112}},\
  \bibinfo {pages} {140502} (\bibinfo {year} {2014})}\BibitemShut {NoStop}%
\bibitem [{\citenamefont {Miyadera}(2016)}]{Miyadera2015a}%
  \BibitemOpen
  \bibfield  {author} {\bibinfo {author} {\bibfnamefont {T.}~\bibnamefont
  {Miyadera}},\ }\bibfield  {title} {\bibinfo {title} {{Energy-Time Uncertainty
  Relations in Quantum Measurements}},\ }\href
  {https://doi.org/10.1007/s10701-016-0027-6} {\bibfield  {journal} {\bibinfo
  {journal} {Found. Phys.}\ }\textbf {\bibinfo {volume} {46}},\ \bibinfo
  {pages} {1522} (\bibinfo {year} {2016})}\BibitemShut {NoStop}%
\bibitem [{\citenamefont {Abdelkhalek}\ \emph {et~al.}(2016)\citenamefont
  {Abdelkhalek}, \citenamefont {Nakata},\ and\ \citenamefont
  {Reeb}}]{Abdelkhalek2016}%
  \BibitemOpen
  \bibfield  {author} {\bibinfo {author} {\bibfnamefont {K.}~\bibnamefont
  {Abdelkhalek}}, \bibinfo {author} {\bibfnamefont {Y.}~\bibnamefont
  {Nakata}},\ and\ \bibinfo {author} {\bibfnamefont {D.}~\bibnamefont {Reeb}},\
  }\bibfield  {title} {\bibinfo {title} {{Fundamental energy cost for quantum
  measurement}},\ }\href {http://arxiv.org/abs/1609.06981} {\  (\bibinfo {year}
  {2016})},\ \Eprint {https://arxiv.org/abs/1609.06981} {arXiv:1609.06981}
  \BibitemShut {NoStop}%
\bibitem [{\citenamefont {Hayashi}\ and\ \citenamefont
  {Tajima}(2017)}]{Hayashi2017}%
  \BibitemOpen
  \bibfield  {author} {\bibinfo {author} {\bibfnamefont {M.}~\bibnamefont
  {Hayashi}}\ and\ \bibinfo {author} {\bibfnamefont {H.}~\bibnamefont
  {Tajima}},\ }\bibfield  {title} {\bibinfo {title} {{Measurement-based
  formulation of quantum heat engines}},\ }\href
  {https://doi.org/10.1103/PhysRevA.95.032132} {\bibfield  {journal} {\bibinfo
  {journal} {Phys. Rev. A}\ }\textbf {\bibinfo {volume} {95}},\ \bibinfo
  {pages} {032132} (\bibinfo {year} {2017})}\BibitemShut {NoStop}%
\bibitem [{\citenamefont {Lipka-Bartosik}\ and\ \citenamefont
  {Demkowicz-Dobrza{\'{n}}ski}(2018)}]{Lipka-Bartosik2018}%
  \BibitemOpen
  \bibfield  {author} {\bibinfo {author} {\bibfnamefont {P.}~\bibnamefont
  {Lipka-Bartosik}}\ and\ \bibinfo {author} {\bibfnamefont {R.}~\bibnamefont
  {Demkowicz-Dobrza{\'{n}}ski}},\ }\bibfield  {title} {\bibinfo {title}
  {{Thermodynamic work cost of quantum estimation protocols}},\ }\href
  {https://doi.org/10.1088/1751-8121/aae664} {\bibfield  {journal} {\bibinfo
  {journal} {J. Phys. A Math. Theor.}\ }\textbf {\bibinfo {volume} {51}},\
  \bibinfo {pages} {474001} (\bibinfo {year} {2018})}\BibitemShut {NoStop}%
\bibitem [{\citenamefont {Solfanelli}\ \emph {et~al.}(2019)\citenamefont
  {Solfanelli}, \citenamefont {Buffoni}, \citenamefont {Cuccoli},\ and\
  \citenamefont {Campisi}}]{Solfanelli2019}%
  \BibitemOpen
  \bibfield  {author} {\bibinfo {author} {\bibfnamefont {A.}~\bibnamefont
  {Solfanelli}}, \bibinfo {author} {\bibfnamefont {L.}~\bibnamefont {Buffoni}},
  \bibinfo {author} {\bibfnamefont {A.}~\bibnamefont {Cuccoli}},\ and\ \bibinfo
  {author} {\bibfnamefont {M.}~\bibnamefont {Campisi}},\ }\bibfield  {title}
  {\bibinfo {title} {{Maximal energy extraction via quantum measurement}},\
  }\href {https://doi.org/10.1088/1742-5468/ab3721} {\bibfield  {journal}
  {\bibinfo  {journal} {J. Stat. Mech. Theory Exp.}\ }\textbf {\bibinfo
  {volume} {2019}},\ \bibinfo {pages} {094003} (\bibinfo {year}
  {2019})}\BibitemShut {NoStop}%
\bibitem [{\citenamefont {Purves}\ and\ \citenamefont
  {Short}(2021)}]{Purves-2020}%
  \BibitemOpen
  \bibfield  {author} {\bibinfo {author} {\bibfnamefont {T.}~\bibnamefont
  {Purves}}\ and\ \bibinfo {author} {\bibfnamefont {A.~J.}\ \bibnamefont
  {Short}},\ }\bibfield  {title} {\bibinfo {title} {{Channels, measurements,
  and postselection in quantum thermodynamics}},\ }\href
  {https://doi.org/10.1103/PhysRevE.104.014111} {\bibfield  {journal} {\bibinfo
   {journal} {Phys. Rev. E}\ }\textbf {\bibinfo {volume} {104}},\ \bibinfo
  {pages} {014111} (\bibinfo {year} {2021})}\BibitemShut {NoStop}%
\bibitem [{\citenamefont {Hovhannisyan}\ and\ \citenamefont
  {Imparato}(2021)}]{Hovhannisyan2021}%
  \BibitemOpen
  \bibfield  {author} {\bibinfo {author} {\bibfnamefont {K.~V.}\ \bibnamefont
  {Hovhannisyan}}\ and\ \bibinfo {author} {\bibfnamefont {A.}~\bibnamefont
  {Imparato}},\ }\bibfield  {title} {\bibinfo {title} {{Energy conservation and
  fluctuation theorem are incompatible for quantum work}},\ }\href
  {http://arxiv.org/abs/2104.09364} {\  (\bibinfo {year} {2021})},\ \Eprint
  {https://arxiv.org/abs/2104.09364} {arXiv:2104.09364} \BibitemShut {NoStop}%
\bibitem [{\citenamefont {Aw}\ \emph {et~al.}(2021)\citenamefont {Aw},
  \citenamefont {Buscemi},\ and\ \citenamefont {Scarani}}]{Aw2021}%
  \BibitemOpen
  \bibfield  {author} {\bibinfo {author} {\bibfnamefont {C.~C.}\ \bibnamefont
  {Aw}}, \bibinfo {author} {\bibfnamefont {F.}~\bibnamefont {Buscemi}},\ and\
  \bibinfo {author} {\bibfnamefont {V.}~\bibnamefont {Scarani}},\ }\bibfield
  {title} {\bibinfo {title} {{Fluctuation theorems with retrodiction rather
  than reverse processes}},\ }\href {https://doi.org/10.1116/5.0060893}
  {\bibfield  {journal} {\bibinfo  {journal} {AVS Quantum Sci.}\ }\textbf
  {\bibinfo {volume} {3}},\ \bibinfo {pages} {045601} (\bibinfo {year}
  {2021})}\BibitemShut {NoStop}%
\bibitem [{\citenamefont {Beyer}\ \emph {et~al.}(2022)\citenamefont {Beyer},
  \citenamefont {Uola}, \citenamefont {Luoma},\ and\ \citenamefont
  {Strunz}}]{Beyer2021}%
  \BibitemOpen
  \bibfield  {author} {\bibinfo {author} {\bibfnamefont {K.}~\bibnamefont
  {Beyer}}, \bibinfo {author} {\bibfnamefont {R.}~\bibnamefont {Uola}},
  \bibinfo {author} {\bibfnamefont {K.}~\bibnamefont {Luoma}},\ and\ \bibinfo
  {author} {\bibfnamefont {W.~T.}\ \bibnamefont {Strunz}},\ }\bibfield  {title}
  {\bibinfo {title} {{Joint measurability in nonequilibrium quantum
  thermodynamics}},\ }\href {https://doi.org/10.1103/PhysRevE.106.L022101}
  {\bibfield  {journal} {\bibinfo  {journal} {Phys. Rev. E}\ }\textbf {\bibinfo
  {volume} {106}},\ \bibinfo {pages} {L022101} (\bibinfo {year}
  {2022})}\BibitemShut {NoStop}%
\bibitem [{\citenamefont {Danageozian}\ \emph {et~al.}(2022)\citenamefont
  {Danageozian}, \citenamefont {Wilde},\ and\ \citenamefont
  {Buscemi}}]{Danageozian2022}%
  \BibitemOpen
  \bibfield  {author} {\bibinfo {author} {\bibfnamefont {A.}~\bibnamefont
  {Danageozian}}, \bibinfo {author} {\bibfnamefont {M.~M.}\ \bibnamefont
  {Wilde}},\ and\ \bibinfo {author} {\bibfnamefont {F.}~\bibnamefont
  {Buscemi}},\ }\bibfield  {title} {\bibinfo {title} {{Thermodynamic
  Constraints on Quantum Information Gain and Error Correction: A Triple
  Trade-Off}},\ }\href {https://doi.org/10.1103/PRXQuantum.3.020318} {\bibfield
   {journal} {\bibinfo  {journal} {PRX Quantum}\ }\textbf {\bibinfo {volume}
  {3}},\ \bibinfo {pages} {020318} (\bibinfo {year} {2022})}\BibitemShut
  {NoStop}%
\bibitem [{\citenamefont {Mohammady}(2022)}]{Mohammady2022}%
  \BibitemOpen
  \bibfield  {author} {\bibinfo {author} {\bibfnamefont {M.~H.}\ \bibnamefont
  {Mohammady}},\ }\bibfield  {title} {\bibinfo {title} {{Thermodynamically free
  quantum measurements}},\ }\href {https://doi.org/10.1088/1751-8121/acad4a}
  {\bibfield  {journal} {\bibinfo  {journal} {J. Phys. A Math. Theor.}\
  }\textbf {\bibinfo {volume} {55}},\ \bibinfo {pages} {505304} (\bibinfo
  {year} {2022})}\BibitemShut {NoStop}%
\bibitem [{\citenamefont {Stevens}\ \emph {et~al.}(2022)\citenamefont
  {Stevens}, \citenamefont {Szombati}, \citenamefont {Maffei}, \citenamefont
  {Elouard}, \citenamefont {Assouly}, \citenamefont {Cottet}, \citenamefont
  {Dassonneville}, \citenamefont {Ficheux}, \citenamefont {Zeppetzauer},
  \citenamefont {Bienfait}, \citenamefont {Jordan}, \citenamefont
  {Auff{\`{e}}ves},\ and\ \citenamefont {Huard}}]{Stevens2021}%
  \BibitemOpen
  \bibfield  {author} {\bibinfo {author} {\bibfnamefont {J.}~\bibnamefont
  {Stevens}}, \bibinfo {author} {\bibfnamefont {D.}~\bibnamefont {Szombati}},
  \bibinfo {author} {\bibfnamefont {M.}~\bibnamefont {Maffei}}, \bibinfo
  {author} {\bibfnamefont {C.}~\bibnamefont {Elouard}}, \bibinfo {author}
  {\bibfnamefont {R.}~\bibnamefont {Assouly}}, \bibinfo {author} {\bibfnamefont
  {N.}~\bibnamefont {Cottet}}, \bibinfo {author} {\bibfnamefont
  {R.}~\bibnamefont {Dassonneville}}, \bibinfo {author} {\bibfnamefont
  {Q.}~\bibnamefont {Ficheux}}, \bibinfo {author} {\bibfnamefont
  {S.}~\bibnamefont {Zeppetzauer}}, \bibinfo {author} {\bibfnamefont
  {A.}~\bibnamefont {Bienfait}}, \bibinfo {author} {\bibfnamefont {A.~N.}\
  \bibnamefont {Jordan}}, \bibinfo {author} {\bibfnamefont {A.}~\bibnamefont
  {Auff{\`{e}}ves}},\ and\ \bibinfo {author} {\bibfnamefont {B.}~\bibnamefont
  {Huard}},\ }\bibfield  {title} {\bibinfo {title} {{Energetics of a Single
  Qubit Gate}},\ }\href {https://doi.org/10.1103/PhysRevLett.129.110601}
  {\bibfield  {journal} {\bibinfo  {journal} {Phys. Rev. Lett.}\ }\textbf
  {\bibinfo {volume} {129}},\ \bibinfo {pages} {110601} (\bibinfo {year}
  {2022})}\BibitemShut {NoStop}%
\bibitem [{\citenamefont {Bratteli}\ \emph {et~al.}(2000)\citenamefont
  {Bratteli}, \citenamefont {Jorgensen}, \citenamefont {Kishimoto},\ and\
  \citenamefont {Werner}}]{Bratteli1998}%
  \BibitemOpen
  \bibfield  {author} {\bibinfo {author} {\bibfnamefont {O.}~\bibnamefont
  {Bratteli}}, \bibinfo {author} {\bibfnamefont {P.~E.~T.}\ \bibnamefont
  {Jorgensen}}, \bibinfo {author} {\bibfnamefont {A.}~\bibnamefont
  {Kishimoto}},\ and\ \bibinfo {author} {\bibfnamefont {R.~F.}\ \bibnamefont
  {Werner}},\ }\bibfield  {title} {\bibinfo {title} {{Pure states on O{\_}d}},\
  }\href {https://www.jstor.org/stable/24715231} {\bibfield  {journal}
  {\bibinfo  {journal} {J. Oper. Theory}\ }\textbf {\bibinfo {volume} {43}},\
  \bibinfo {pages} {97} (\bibinfo {year} {2000})}\BibitemShut {NoStop}%
\bibitem [{\citenamefont {Arias}\ \emph {et~al.}(2002)\citenamefont {Arias},
  \citenamefont {Gheondea},\ and\ \citenamefont {Gudder}}]{Arias2002}%
  \BibitemOpen
  \bibfield  {author} {\bibinfo {author} {\bibfnamefont {A.}~\bibnamefont
  {Arias}}, \bibinfo {author} {\bibfnamefont {A.}~\bibnamefont {Gheondea}},\
  and\ \bibinfo {author} {\bibfnamefont {S.}~\bibnamefont {Gudder}},\
  }\bibfield  {title} {\bibinfo {title} {{Fixed points of quantum
  operations}},\ }\href {https://doi.org/10.1063/1.1519669} {\bibfield
  {journal} {\bibinfo  {journal} {J. Math. Phys.}\ }\textbf {\bibinfo {volume}
  {43}},\ \bibinfo {pages} {5872} (\bibinfo {year} {2002})}\BibitemShut
  {NoStop}%
\bibitem [{\citenamefont {Lindblad}(1999)}]{Lindblad1999a}%
  \BibitemOpen
  \bibfield  {author} {\bibinfo {author} {\bibfnamefont {G.}~\bibnamefont
  {Lindblad}},\ }\bibfield  {title} {\bibinfo {title} {{A general no-cloning
  theorem}},\ }\href {https://doi.org/10.1023/A:1007581027660} {\bibfield
  {journal} {\bibinfo  {journal} {Lett. Math. Phys.}\ }\textbf {\bibinfo
  {volume} {47}},\ \bibinfo {pages} {189} (\bibinfo {year} {1999})}\BibitemShut
  {NoStop}%
\bibitem [{\citenamefont {Choi}(1974)}]{Choi1974}%
  \BibitemOpen
  \bibfield  {author} {\bibinfo {author} {\bibfnamefont {M.-D.}\ \bibnamefont
  {Choi}},\ }\bibfield  {title} {\bibinfo {title} {{A schwarz inequality for
  positive linear maps on C*-algebras}},\ }\href
  {https://doi.org/10.1215/ijm/1256051007} {\bibfield  {journal} {\bibinfo
  {journal} {Illinois J. Math.}\ }\textbf {\bibinfo {volume} {18}},\ \bibinfo
  {pages} {565} (\bibinfo {year} {1974})}\BibitemShut {NoStop}%
\bibitem [{\citenamefont {Choi}(1975)}]{Choi1975}%
  \BibitemOpen
  \bibfield  {author} {\bibinfo {author} {\bibfnamefont {M.~D.}\ \bibnamefont
  {Choi}},\ }\bibfield  {title} {\bibinfo {title} {{Positive Linear Maps on
  Complex Matrices}},\ }\href
  {http://www.sciencedirect.com/science/article/pii/0024379575900750{\#}}
  {\bibfield  {journal} {\bibinfo  {journal} {Linear Algebra Appl.}\ }\textbf
  {\bibinfo {volume} {10}},\ \bibinfo {pages} {285} (\bibinfo {year}
  {1975})}\BibitemShut {NoStop}%
\end{thebibliography}%


%

\end{document}